\documentclass[11pt,a4paper,fleqn]{article}%,draft

\usepackage{amsmath,amssymb,amsthm}%,refcheck
\usepackage[mathscr]{eucal}

\usepackage{hyperref}
\hypersetup{colorlinks, linkcolor=blue, citecolor=blue, urlcolor=blue}

\allowdisplaybreaks

\newcommand{\p}{\partial}

\newcommand{\lsemioplus}{\mathbin{\mbox{$\lefteqn{\hspace{.77ex}\rule{.4pt}{1.2ex}}{\in}$}}}
\newcommand{\sgn}{\mathop{\rm sgn}\nolimits}
\newcommand{\rank}{\mathop{\rm rank}\nolimits}
\newcommand{\spanindex}{{\mbox{\tiny$\langle\,\rangle$}}}

\newcommand{\todo}[1][\null]{\ensuremath{\clubsuit}}

\newcommand{\noprint}[1]{}
\newcommand{\checked}[1][\null]{\ensuremath{\boldsymbol{\surd}}}

\newtheorem{theorem}{Theorem}%[section]
\newtheorem{lemma}[theorem]{Lemma}
\newtheorem{corollary}[theorem]{Corollary}
\newtheorem{proposition}[theorem]{Proposition}
\newtheorem*{proposition*}{Proposition}
{\theoremstyle{definition}
\newtheorem{definition}[theorem]{Definition}
\newtheorem{example}[theorem]{Example}
\newtheorem{remark}[theorem]{Remark}
\newtheorem*{notation*}{Notation}
\newtheorem*{conventions*}{Conventions}
}

\begin{document}

\par\noindent {\LARGE\bf
Algebraic method of group classification\\ 
for semi-normalized classes of differential equations
%Group classification of multidimensional linear\\ Schr\"odinger equations with the algebraic method
\par}

\vspace{6mm}\par\noindent{\large
C\'elestin Kurujyibwami$^\dag$, Dmytro R. Popovych$^{\ddag\S}$ and Roman O. Popovych$^{\ddag\Diamond}$
}\par\vspace{3mm}\par

{\small\it
\vspace{2mm}\par\noindent
$^\dag$\,College of Science and Technology, University of Rwanda, P.O.\,Box: 3900, Kigali, Rwanda

\vspace{2mm}\par\noindent
$^\ddag$\,Department of Mathematics and Statistics, Memorial University of Newfoundland,\\
$\phantom{^\ddag}$\,St.\ John's (NL) A1C 5S7, Canada

\vspace{2mm}\par\noindent
$^\S$\,Institute of Mathematics of NAS of Ukraine, 3 Tereshchenkivska Str., 01004 Kyiv, Ukraine

\vspace{2mm}\par\noindent
$^\Diamond$\,Mathematical Institute, Silesian University in Opava, Na Rybn\'\i{}\v{c}ku 1, 746 01 Opava, Czechia

}

\vspace{6mm}\par\noindent
E-mails:
celeku@yahoo.fr,
dpopovych@mun.ca, 
rop@imath.kiev.ua

\vspace{8mm}\par\noindent\hspace*{10mm}\parbox{140mm}{\small
We generalize the notion of semi-normalized classes of systems of differential equations, 
study properties of such classes and extend the algebraic method of group classification to them.
In particular, we prove the important theorems 
on factoring out symmetry groups and invariance algebras of systems from semi-normalized classes 
and on splitting such groups and algebras within disjointedly semi-normalized classes. 
Nontrivial particular examples of classes that arise in real-world applications 
and showcase the relevance of the developed theory are provided.
To convincingly illustrate the efficiency of the proposed method, 
we apply it to the group classification problem for the class
of linear Schr\"odinger equations with complex-valued potentials
and the general value of the space dimension.
We compute the equivalence groupoid of the class by the direct method
and thus show that this class is uniformly semi-normalized with respect to the linear superposition of solutions.
This is why the group classification problem reduces to the classification
of specific low-dimensional subalgebras of the associated equivalence algebra,
which is completely realized for the case of space dimension two.
Splitting into different classification cases is based on three integer parameters
that are invariant with respect to equivalence transformations.
We also single out those of the obtained results 
that are relevant to linear Schr\"odinger equations with real-valued potentials.
}

\vspace{4mm}

\section{Introduction}

Systems of differential equations that appear in real-world applications 
generally contain constant or functional parameters 
that are not fixed and should be found using experimental data or additional arguments. 
An example of such parameters called arbitrary elements is given by potentials in Shr\"odinger equations.
The scope of relevant models is often restricted by certain requirements to their symmetry properties 
following from the corresponding theory, 
like certain scale invariance or the Galilean or special relativity principles in physics. 
In mathematical terminology, this means that 
modeling differential equations have to possess 
a certain symmetry group or symmetry groups with certain properties
or the widest symmetry group among the possible ones~\cite{ovsy1982A}.  
Thus, to single out such models from a parameterized class of systems of differential equations,
one should solve the group classification problem for this class. 
This is why the theory of group classification plays nowadays a central role 
in symmetry analysis of differential equations and its application to various sciences.

This theory has its origins in the work of Sophus Lie. 
His greatest achievement here was solving the group classification problems 
for three classes differential equations, 
second-order ordinary differential equations~\cite{lie1893A}, 
second-order linear partial differential equations in two independent variables~\cite{lie1881a}
and nonlinear Klein--Gordon equations of the form $d^2z/dx\,dy=F(z)$~\cite{lie1881b}. 
Both the main approaches to solving group classification problems, 
which are the direct and the algebraic methods, also originated from these Lie's studies. 
The modern development of the theory of group classification was started by Ovsyannikov \cite{ovsy1959a,ovsy1982A}, 
and this has been the field of intensive research ever since, 
see \cite{CRChandbook1994V1,CRChandbook1995V2,lahn2006a,vane2020b} and references therein.

The group classification problem for a class of differential equations
is complex beyond constructing the Lie symmetry group of a single system of differential equations.
It involves classifying solutions of a cumbersome overdetermined system of partial differential equations,
where the unknown functions are not only the coefficients of Lie-symmetry vector fields 
but also the arbitrary elements of the class.
This is why the progress in implementing the group classification framework in software
is not significant despite there existing many specialized packages for finding Lie symmetries
in various symbolic computation systems 
\cite{BaranMarvan,carm2000a,chev2007a,chev2010d,marv2009a,Wittkopf2004PhD}.
Moreover, a considerable fraction of group classification problems
presented in the literature is tackled using a brute-force version of the direct method, 
which is advisable only for some simple classes of differential equations, 
see a discussion in~\cite[Section~4]{opan2020b}.
This definitely inefficient approach requires a rigorous effort to keep track of various cases
of integrating the corresponding determining equations for Lie symmetries
and furthermore to check inequivalence of these cases.
Any lapse in attention leads to missed, overlapping, repeated and/or incorrect classification cases
afflicting most papers following this approach.
Due to their unhandiness, the corresponding computations are not properly presented if presented at all,
which makes the obtained classification results unverifiable and thus not trustworthy.

In the meanwhile, a number of novel techniques for group classification have been developed recently,
which include variations to the advanced modification of the direct method 
called the \emph{method of furcate splitting} \cite{ivan2010a,opan2020b,niki2001a}
and various flavors of the \emph{algebraic method} \cite{bihl2012b,card2011a,kuru2018a,opan2017a,popo2010a,vane2020b}. 
The former method was suggested and applied for the first time in~\cite{niki2001a}
but was given its name only in~\cite{ivan2010a} and formalized to its full extend in~\cite{opan2020b}.
See also \cite{bihl2020a,popo2004b,vane2012b,vane2015a} and references therein 
for examples of its efficient application.
It is advisable only for classes of simple structure, 
which is, at the same time, not appropriate to apply the more powerful algebraic method.
In particular, arbitrary elements in such classes should be constants or depend on few arguments, 
and the reasonable upper bound for the number of arbitrary-element arguments, 
which was handled in the literature~\cite{niki2001a}, is two. 
At the same time, the latter method is not so sensitive to the number of arguments of arbitrary elements.

Despite its invention by Lie, 
the algebraic method of group classification became a common tool of group analysis of differential equations
considerably later, only since the 1990s 
\cite{basa2001a,gagn1993a,gaze1992a,gung2004a,lahn2005a,lahn2006a,maga1993a,ozem2013a,wint1992a,zhda1999d}.
Moreover, the fact that the class of differential equations under consideration 
should admit special (normalization) properties for the proper applicability of this method 
was understood even later~\cite{popo2006b,popo2005c,popo2010a}.  
The straightforward application of the algebraic method to non-normalized classes of differential equations 
results in the so-called preliminary group classification of such classes \cite{akha1991a,card2011a,ibra1991a}, 
see the discussion in \cite[Section~XI]{bihl2012b}.
To extend the scope of applicability of the algebraic method, 
a number of various techniques and approaches were invented, 
including 
splitting the class under consideration into normalized subclasses \cite{kuru2020a,popo2010a},
constructing its normalized superclass or even a hierarchy normalized classes 
related to it \cite{bihlo2017a,kuru2020a,opan2017a,popo2010a}, 
singling out its singular subclasses \cite{bihl2012b}, 
gauging its arbitrary elements using a wide subgroupoid of its equivalence groupoid 
\cite{bihlo2017a,kuru2020a,opan2017a,opan2020b,popo2010a},
mappings between classes that are generated by point transformations~\cite{opan2020b} and  
partitioning classification cases into regular and singular ones~\cite{vane2020b}.
One of such approaches is to extend the algebraic method to classes with weaker normalization properties, 
e.g., to classes that are semi-normalized or normalized in the generalized or the extended sense. 
The first version of the algebraic method for group classification of semi-normalized classes 
was suggested in~\cite{kuru2018a} and applied therein 
for solving the group classification problem for the class 
of (1+1)-dimensional linear Schr\"odinger equations with complex-valued potentials.
Then this version was used to classify Lie symmetries 
of (1+1)-dimensional linear evolution equations of arbitrary order greater than two~\cite{bihlo2017a}. 
It was shown in~\cite{opan2020b} that further development of the algebraic method of group classification 
requires a deep revision of foundations of the theory of group classification, 
including the statement of group classification problems. 

In the present paper, we enhance and further generalize the above version of this method 
and apply it to the class~$\mathcal F$ of (1+$n$)-dimensional ($n\geqslant1$) linear Schr\"odinger equations
with complex-valued potentials 
\begin{gather}\label{MLinSchEqs}
i\psi_t+\psi_{aa}+V(t,x)\psi=0,
\end{gather}
where $t$ and $x=(x_1,\dots,x_n)$ are the real independent variables,
$\psi$ is the complex dependent variable
and $V$ is an arbitrary smooth complex-valued potential depending on $t$ and $x$.
Here and in what follows, the subscripts $t$ and $a$ denote differentiation with respect to $t$ and $x_a$, respectively,
the indices $a$, $b$, $c$, and $d$ run from $1$ to $n$,
and we use summation convention over repeated indices.
By~$\mathcal F_{\mathbb R}$ we denote the subclass of the class~$\mathcal F$ 
that consists of the equations of the form~\eqref{MLinSchEqs} with real-valued potentials~$V$.

More specifically, we generalize the notion of semi-normalized classes of differential equations, 
study properties of such classes and extend the algebraic method of group classification to them.
Basic for this extension are the new theorems 
on factoring out symmetry groups and invariance algebras of systems from semi-normalized classes 
and on splitting such groups and algebras within disjointedly semi-normalized classes. 
For the class~$\mathcal F$ with the general value of~$n$,
we compute its equivalence groupoid~$\mathcal G^\sim$ by the direct method
and thus prove its uniform semi-normalization with respect to the linear superposition of solutions.
In this way, we reduce the group classification problem for the class~$\mathcal F$
to the classification of appropriate low-dimensional subalgebras of its equivalence algebra.
This problem is completely solved in the case of $n=2$.
Splitting into different classification cases is based on three integer parameters.
These parameters are invariant with respect to equivalence transformations 
and characterize the dimensions of parts of the corresponding Lie symmetry algebras,
which are related to transformation of time, rotations 
and generalized time-dependent shifts with respect to space variables, respectively.
The analogous results for the subclass~$\mathcal F_{\mathbb R}$ are singled out from those above. 
In the study of the class~$\mathcal F$, the present paper succeeds~\cite{kuru2018a}, 
where the case $n=1$ was comprehensively studied, and~\cite{kuru2015a}, 
where the groupoid~$\mathcal G^\sim$ was computed for the case $n=2$.

The study of symmetries of Schr\"odinger equations began in the 1970s 
with Lie symmetries of linear Schr\"odinger equations with real-valued time-independent potentials
\cite{Boyer1974,Miller1977,Niederer1972,Niederer1973a,Niederer1973b,Niederer1974},
but later the focus shifted to nonlinear Schr\"odinger equations
\cite{belm2008a,BoyerSharpWinternitz1976,Fushchich&Moskaliuk1981,Gagnon88a,Gagnon89a,Gagnon89b,Gagnon89c,gagn1993a,kuru2020a,
ozem2013a,Popovych&Ivanova&Eshragi2004,popo2010a,Zhdanov&Roman2000},
see also detailed reviews in~\cite[Section~4]{popo2010a} and~\cite[Section~1]{kuru2018a}.
At the same time, Lie symmetries of linear Schr\"odinger equations and their generalizations
were not properly classified.
In particular, this is the case for linear Schr\"odinger equations with complex-valued potentials  
although they also arise as models in quantum mechanics,
scattering theory, condensed matter physics, quantum field theory and so forth
\cite{Bender2002,Fernandez&Guardiola&Ros&Znojil1999,Muga&Palao&Navarro&Egusquiza2004}.
Moreover, the existing results for the standard case of real-valued time-independent potentials 
are neither exhaustive nor completely correct, which has been found out in the course of the consideration 
of more general Schr\"odinger equations.
Thus, in~\cite{niki2016a}, Schr\"odinger equations with position-dependent mass and stationary real-valued potentials 
were considered from the Lie-symmetry point of view for arbitrary space dimension, 
and their Lie symmetries were completely classified for space dimension two. 
The obtained classification list contains three case families formally intersecting the case of constant mass. 
As a result, a Lie-symmetry classification case 
for Schr\"odinger equations with constant mass and real-valued time-independent potentials
that had been missed in~\cite{Boyer1974} was identified; 
see the conclusion of the present paper for more details.   
Hence there are a number of open problems even on Lie symmetries of various Schr\"odinger equations, 
including linear ones. 
One can consider wider or new classes of Schr\"odinger equations~\cite{niki2020a,niki2016a} 
or carry out complete, partial or inverse group classification 
of models with additional properties such as superintegrability~\cite{niki2022b,niki2023a}.

The structure of this paper is as follows.
Basic notions related to classes of differential equations
are reviewed in Section~\ref{sec:BasicsNotionsOfGroupClassification}.
In Section~\ref{sec:Semi-normalizedClasses}, 
we extend the notion of a semi-normalized class of differential equations 
to the semi-normalization with respect to 
a uniform family of point-symmetry groups of systems from the class under consideration
and a subgroup of the corresponding equivalence group, 
study properties of newly defined semi-normalized classes 
and prove the theorems on factoring out point-symmetry groups and maximal Lie invariance algebras
for systems from such classes.
In Section~\ref{sec:DisjointSemi-normalization}, 
we distinguish between the disjoint and the non-disjoint semi-normalizations, 
study specific properties of disjointedly semi-normalized classes of differential equations 
and prove the stronger theorems on splitting point-symmetry groups and maximal Lie invariance algebras 
within such classes.
We also present two important examples, 
the class that is disjointedly semi-normalized
only with respect to a proper subgroup of its equivalence group
and a semi-normalized class that is not disjointedly semi-normalized.
The additional attribute ``uniform'' for semi-normalization is introduced in Section~\ref{sec:UniformSemi-norm}. 
Therein we discuss an example of a class that is semi-normalized but not uniformly semi-normalized. 
More attention is paid to the particularly important case of uniform semi-normalization
given by the classes of homogeneous linear systems of differential equations 
that are uniformly semi-normalized with respect to the linear superposition of solutions. 
This specific semi-normalization is especially relevant to the symmetry analysis 
of the linear Schr\"odinger equations in the second part of the present paper.
The algebraic method of group classification is extended to semi-normalized classes 
in Section~\ref{sec:AlgMethodForSemi-normalizedClasses}.%\looseness=-1

Section~\ref{sec:PreliminarySymAnalisysOfMultiDSchEqs} is devoted to
the preliminary symmetry analysis of the class~$\mathcal F$ of (1+$n$)-dimensional linear Schr\"odinger equations
 of the form~\eqref{MLinSchEqs} with complex-valued potentials and with an arbitrary value $n\geqslant1$. 
We compute the equivalence groupoid~$\mathcal G^\sim$, 
the equivalence group~$G^\sim$ and the equivalence algebra~$\mathfrak g^\sim$ of this class. 
It is proved that
the class~$\mathcal F$ is uniformly semi-normalized with respect to the linear superposition of solutions.
We derive and preliminarily integrate the determining equations
for Lie symmetries of equations from the class~$\mathcal F$, 
which allows us to obtain the general form of Lie-symmetry vector fields,
the kernel Lie invariance algebra~$\mathfrak g^\cap$ of the class~$\mathcal F$
as well as the classifying condition for such Lie symmetries.
Then we analyze properties of subalgebras of~$\mathfrak g^\sim$
(more precisely, of the projection~$\pi_*\mathfrak g^\sim$ of~$\mathfrak g^\sim$ to the space of independent and dependent variables)
that are appropriate as maximal Lie invariance algebras of equations from the class~$\mathcal F$. 
In particular, we obtain that the least upper bound for the dimensions of these subalgebras equals $n(n+3)/2+5$.
As a result, we show that 
the group classification problem for the class~$\mathcal F$ 
reduces to the classification of appropriate subalgebras of the algebra~$\pi_*\mathfrak g$ 
with respect to the equivalence relation generated by the action of the projected group~$\pi_*G^\sim$.
The group classification problem for the class~$\mathcal F$ in the case $n=2$
is completely solved in Section~\ref{sec:GroupClassificationOf(1+2)DLinSchEqs}.
We use the above results on the class~$\mathcal F$ in Section~\ref{sec:Real-ValuedPotentials} 
to derive their counterparts for the class~$\mathcal F_{\mathbb R}$.\looseness=-1 

Lastly, we conclude with a summary of results, their discussion and suggestions for future considerations 
in Section~\ref{sec:Conclusion}.

\begin{conventions*}
Throughout the paper,
we assume summation over repeated indices.
Subscripts of functions denote differentiation with respect to the corresponding variables.
Given a class of differential equations,
$\pi$ denotes the natural projection of 
the space run by the tuple of the independent and dependent variables, relevant jet variables and the arbitrary elements of the class
onto the space run by the tuple of the independent and dependent variables only. 
When talking about pseudogroups, we omit the prefix ``pseudo''.%
\footnote{%
As for many other classes of differential equations, 
the equivalence group of the class~$\mathcal F$ 
and the point symmetry groups of equations from this class 
are in fact pseudogroups but not groups.  
The theory developed should be formulated involving the notions pseudogroups and pseudoalgebras. 
At the same time, such a modification essentially complicates the consideration and 
requires revising and enhancing 
the entire framework of symmetry analysis of differential equations, 
beginning with the notions of a system of differential equations and of a class of such systems, 
which we intend to carry out in the future. 
} 
The notation $\rm id$ is used for the identity transformation in the relevant space.
In Sections~\ref{sec:BasicsNotionsOfGroupClassification}--\ref{sec:AlgMethodForSemi-normalizedClasses} 
with theoretical results on general systems of differential equations 
and in Sections~\ref{sec:PreliminarySymAnalisysOfMultiDSchEqs} and~\ref{sec:GroupClassificationOf(1+2)DLinSchEqs}
with specific results on Schr\"odinger equations, we use different locally specified notations.  
\end{conventions*}

\section{Basics notions on classes of differential equations}\label{sec:BasicsNotionsOfGroupClassification}

\newcommand{\EqOrd}{p}

In order to make the presentation self-contained,
we briefly define the notion of a class (of systems) of differential equations
and various objects related to point transformations in such classes, 
mainly following \cite[Section~2]{vane2020b}. 
More details can be found in 
\cite{bihl2012b,opan2022a,popo2006b,popo2010a,vane2020b} 
and references therein.

\begin{notation*}
Up to Section~\ref{sec:AlgMethodForSemi-normalizedClasses},
we use a notation for independent and dependent variables
that differs from that of the other sections:
$x=(x_1,\dots,x_n)$ denotes the complete $n$-tuple of independent variables
and $u=(u^1,\dots,u^m)$ is the $m$-tuple of dependent variables (that is, unknown functions of the independent variables), 
except particular examples, where natural specific notations of variables are used.
The indices $i$ and $j$ run from $1$ to $n$ and the indices $a$ and $b$ run from $1$ to $m$. 
\end{notation*}

Consider a system of differential equations $\mathcal L_\theta$: $L(x,u_{(p)},\theta_{(q)}(x,u_{(p)}))=0$
parameterized by a tuple of arbitrary elements
$\theta(x,u_{(p)})=(\theta^1(x,u_{(p)}),\dots,\theta^k(x,u_{(p)})),$
where $u_{(p)}$ stands for the tuple of derivatives of~$u$ with respect to~$x$ up to order~$p$,
including~$u$ as the zeroth order derivatives,
and $\theta_{(q)}$ denotes the tuple of derivatives of the tuple~$\theta$
with respect to its arguments~$(x,u_{(p)})$ up to order~$q$.
The arbitrary-element tuple $\theta$ runs through the set $\mathcal S$ of solutions
of an auxiliary system of differential equations $S(x,u_{(p)},\theta_{(q')}(x,u_{(p)}))=0$
and differential inequalities of the form $\Sigma(x,u_{(p)},\theta_{(q')}(x,u_{(p)}))\ne 0$,
where both $x$ and $u_{(p)}$ play the role of independent variables,
and $S$ and $\Sigma$ are tuples of smooth functions depending on $x,u_{(p)}$ and $\theta_{(q')}$.
Other kinds of inequalities are also possible.
By~$\mathcal L|_{\mathcal S}$ we denote the class of systems $\mathcal L_\theta$
with arbitrary elements $\theta$ running through~$\mathcal S$,
$\mathcal L|_{\mathcal S}:=\{\mathcal L_\theta\mid\theta\in\mathcal S\}$.

The \emph{(point) equivalence groupoid} $\mathcal G^\sim$ of the class~$\mathcal L|_{\mathcal S}$ 
is the small category with $\mathcal S$ 
as the set of objects and with the set of point transformations%
\footnote{%
In the case of single dependent variable, 
one can also consider the contact equivalence groupoid of the class~$\mathcal L|_{\mathcal S}$ 
that consists of the admissible contact transformations within this class. 
}
of~$(x,u)$, 
i.e., of local diffeomorphisms in the space with the coordinates~$(x,u)$, 
between pairs of systems from~$\mathcal L|_{\mathcal S}$ as the set of arrows, 
\[
\mathcal G^\sim=\big\{\mathcal T=(\theta,\Phi,\tilde\theta)\mid
\theta,\tilde\theta\in\mathcal S,\,\Phi\in{\rm Diff}^{\rm loc}_{(x,u)}\colon\Phi_*\mathcal L_\theta=\mathcal L_{\tilde\theta}\big\}.
\]
More specifically, 
elements~$\mathcal T$ of~$\mathcal G^\sim$ are called \emph{admissible (point) transformations} within the class~$\mathcal L|_{\mathcal S}$. 
The \emph{source} and \emph{target} maps ${\rm s},{\rm t}\colon\mathcal G^\sim\to\mathcal S$ 
are defined by ${\rm s},{\rm t}\colon\mathcal G^\sim\ni\mathcal T=(\theta,\Phi,\tilde\theta)\mapsto\theta,\tilde\theta$.
This map pair inspires the groupoid notation $\mathcal G^\sim\rightrightarrows\mathcal S$. 
Admissible transformations~$\mathcal T=(\theta,\Phi,\tilde\theta)$ and $\mathcal T'=(\theta',\Phi',\tilde\theta')$ 
are \emph{composable} 
if $\tilde\theta=\theta'$, and then their \emph{composition} is $\mathcal T\star\mathcal T'=(\theta,\Phi'\circ\Phi,\tilde\theta')$, 
which defines a natural partial multiplication on~$\mathcal G^\sim$.
For any $\theta\in\mathcal S$, the \emph{unit at~$\theta$} is given by \smash{${\rm id}_\theta:=(\theta,{\rm id},\theta)$}, 
where ${\rm id}$ is the identity transformation of~$(x,u)$. 
This defines the \emph{object inclusion map} $\mathcal S\ni\theta\mapsto{\rm id}_\theta\in\mathcal G^\sim$, 
i.e., the object set~$\mathcal S$ can be regarded 
to coincide with the base groupoid $\mathcal S\rightrightarrows\mathcal S:=\{{\rm id}_\theta\mid\theta\in\mathcal S\}$. 
The \emph{inverse} of~$\mathcal T$ is $\mathcal T^{-1}:=(\tilde\theta,\Phi^{-1},\theta)$, where 
$\Phi^{-1}$ is the inverse of~$\Phi$.
All required groupoid properties like associativity of the partial multiplication, 
its consistency with the source and target maps,  
natural properties of units and inverses are obviously satisfied. 
The subset $\mathcal G(\theta,\tilde\theta):={\rm s}^{-1}(\theta)\cap{\rm t}^{-1}(\tilde\theta)$ of~$\mathcal G^\sim$ 
with $\theta,\tilde\theta\in\mathcal S$ 
corresponds to the set of point transformations mapping the system~$\mathcal L_\theta$ to the system~$\mathcal L_{\tilde\theta}$.
The \emph{vertex group} $\mathcal G_\theta:=\mathcal G(\theta,\theta)={\rm s}^{-1}(\theta)\cap{\rm t}^{-1}(\theta)$ 
is associated with the point symmetry (pseudo)group~$G_\theta$ of the system~$\mathcal L_\theta$, 
$G_\theta=\big\{\Phi\in{\rm Diff}^{\rm loc}_{(x,u)}\mid(\theta,\Phi,\theta)\in\mathcal G_\theta\big\}$.
The \emph{fundamental groupoid}~$\mathcal G^{\rm f}$ of the class~$\mathcal L|_{\mathcal S}$ 
is the disjoint union of the vertex groups $\mathcal G_\theta$, $\theta\in\mathcal S$, 
\[
\mathcal G^{\rm f}:=\sqcup_{\theta\in\mathcal S}\mathcal G_\theta
=\big\{(\theta,\Phi,\theta)\mid\theta\in\mathcal S,\,\Phi\in G_\theta\big\}.
\]
Since it has the same object set~$\mathcal S$ and the same vertex groups as~$\mathcal G^\sim$
and $\mathcal T^{-1}\mathcal G_{\tilde\theta}\mathcal T=\mathcal G_\theta$ 
for any $\mathcal T\in\mathcal G(\theta,\tilde\theta)$,
it is a normal subgroupoid of~$\mathcal G^\sim$, 
which is also called the \emph{fundamental subgroupoid} of~$\mathcal G^\sim$. 

Denote by $\pi^k$, $k\in\{0,\dots,\EqOrd\}$, the natural projection from the space with the coordinates $(x,u_{(\EqOrd)},\theta)$ 
to the spaces with the coordinates $(x,u_{(k)})$. In this notation, $\pi^0=\pi$.

The \emph{(usual) equivalence group}~$G^\sim$ of the class~$\mathcal L|_{\mathcal S}$ 
is the (pseudo)group of point transformations $\mathscr T$ in the space with the coordinates $(x,u_{(\EqOrd)},\theta)$ 
that 
\begin{itemize}\itemsep=0ex
\item
are projectable to the spaces with the coordinates $(x,u_{(k)})$ for any $k\in\{0,\dots,\EqOrd\}$,
\item 
are consistent with the contact structure on the space with the coordinates $(x,u_{(\EqOrd)})$, 
i.e., $\pi^\EqOrd_*\mathscr T$ is the standard prolongation of~$\pi^k_*\mathscr T$ to $\EqOrd$th order jets $(x,u_{(\EqOrd)})$, and 
\item 
map the class~$\mathcal L|_{\mathcal S}$ onto itself, 
$\mathscr T_*\mathcal L_\theta\in\mathcal L|_{\mathcal S}$ for any $\theta\in\mathcal S$.
\end{itemize}
If the arbitrary-element tuple~$\theta$ depends at most on $(x,u_{(\EqOrd')})$ with $\EqOrd'<\EqOrd$, then 
the group~$G^\sim$ can be considered to act in the space with the coordinates $(x,u_{(\EqOrd')},\theta)$. 
We call the group $\pi_*G^\sim$ the \emph{shadow} of~$G^\sim$.
Two specific restrictions on equivalence transformations, 
their projectability and their locality with respect to arbitrary elements, 
can be weakened to generalize the notion of usual equivalence group in several ways.
This results in the notions of generalized equivalence group and extended equivalence group, respectively, 
or the notions of extended generalized equivalence group if both restrictions are weakened simultaneously 
\cite{ivan2010a,mele1996a,opan2017a,opan2020b,popo2006b,popo2010a}.

The \emph{action groupoid} $\mathcal G^H$a subgroup~$H$ of the usual equivalence group~$G^\sim$ of the class~$\mathcal L|_{\mathcal S}$, 
\[
\mathcal G^H:=\big\{(\theta,\pi_*\mathscr T,\mathscr T_*\theta)\mid\theta\in\mathcal S,\,\mathscr T\in H\big\},
\] 
is a wide subgroupoid\footnote{%
A subgroupoid~$\mathcal U$ of a groupoid~$\mathcal G$ is called \emph{wide}
if the sets of objects of~$\mathcal U$ and~$\mathcal G$ coincide, ${\rm s}(\mathcal U)={\rm s}(\mathcal G)$.
} 
of the equivalence groupoid~$\mathcal G^\sim$ of this class, 
$\mathcal G^H\subseteq\mathcal G^\sim$ and ${\rm s}(\mathcal G^H)=\mathcal S$. 
We say that an admissible transformation~$\mathcal T$ in the class~$\mathcal L|_{\mathcal S}$ 
is generated by an equivalence transformation from~$H$ if $\mathcal T\in\mathcal G^H$.

The \emph{kernel point symmetry group} $G^\cap:=\cap_{\theta\in\mathcal S}G_\theta$ 
of systems from the class~$\mathcal L|_{\mathcal S}$ 
is related to the unfaithful subgroup~\smash{$\tilde G^\cap$} of~$G^\sim$ 
under the action on~$\mathcal S$, $G^\cap=\pi_*\tilde G^\cap$. 
Thus, \smash{$\tilde G^\cap$} is a normal subgroup of~$G^\sim$ 
that is isomorphic to~$G^\cap$, and the isomorphism is established by~$\pi_*$.  

Given a transformation group, 
it is often required to consider the associated Lie algebra of vector fields, 
which consists of the generators of one-parameter subgroups of this group.
The infinitesimal counterparts 
of the {\it equivalence group}~$G^\sim$ and the kernel point symmetry group~$G^\cap$ of the class $\mathcal L|_{\mathcal S}$
and the point symmetry group~$G_\theta$ of a system~$\mathcal L_\theta$ 
are called 
the \emph{equivalence algebra}~$\mathfrak g^\sim$ and the \emph{kernel invariance algebra} 
of the class $\mathcal L|_{\mathcal S}$ 
and the \emph{maximal Lie invariance algebra} of the system $\mathcal L_\theta$, respectively. 

If the equivalence groupoid~$\mathcal G^\sim$ has specific properties, 
then the class~$\mathcal L|_{\mathcal S}$ is more convenient 
for studying within the framework of symmetry analysis of differential equations 
\cite{bihl2012b,kuru2018a,popo2006b,popo2010a,vane2020b}. 

\begin{definition}
The class~$\mathcal L|_{\mathcal S}$ is called
\emph{normalized} if $\mathcal G^{G^\sim}\!\!=\mathcal G^\sim$. 
It is called \emph{semi-normalized} if $\mathcal G^{\rm f}\star\mathcal G^{G^\sim}\!\!=\mathcal G^\sim$. 
\end{definition}

The Frobenius product 
$
\mathcal G^{\rm f}\star\mathcal G^{G^\sim}\!\!=\big\{\mathcal T\star\mathcal T'\mid
\mathcal T\in\mathcal G^{\rm f},\,\mathcal T'\in\mathcal G^{G^\sim}\!\!,\,{\rm t}(\mathcal T)={\rm s}(\mathcal T')\big\}
$
is a well-defined subgroupoid of~$\mathcal G^\sim$ 
since the fundamental groupoid~$\mathcal G^{\rm f}$ is a normal subgroupoid of~$\mathcal G^\sim$. 
It is obvious that any normalized class of differential equations is semi-normalized.

When considering different kinds of equivalence groups for the class~$\mathcal L|_{\mathcal S}$, 
the usual, the generalized, the extended or the extended generalized equivalence groups of~$\mathcal L|_{\mathcal S}$,
we should specify the corresponding kinds of (semi-)normalization, 
(semi-)normalization in the usual, the generalized, the extended or the extended generalized sense. 
By default, \mbox{(semi-)}normalization means (semi-)normalization in the usual sense.

More general notion of semi-normalization is presented in the next section.

\section{Semi-normalized classes of differential equations}\label{sec:Semi-normalizedClasses}

By the definition introduced in~\cite[Section~3]{kuru2018a},
a class of differential equations~$\mathcal L|_\mathcal S$ is uniformly semi-normalized
if there exists a family of point-symmetry groups of systems from this class
with certain properties with respect to its equivalence groupoid and its equivalence group,
and these point-symmetry groups can be proper subgroups of the corresponding complete point-symmetry groups.
At the same time, it was noted therein that a proper subgroup of the equivalence group
can be used instead of the entire equivalence group as well.
The note was important since already in~\cite[Section~7]{kuru2018a}
we needed to classify Lie symmetries of a class
that can be assumed as uniformly semi-normalized only
with respect to a proper subgroup of its equivalence group.
Carefully analyzing the above results,
we also see that the group classification problem can be solved by the algebraic method
for classes that possess at least a part of properties required for uniformly semi-normalized classes
in~\cite{kuru2018a}.
This is why, following terminology and notation of~\cite{opan2022a,vane2020b},
in the present paper we change the terminology and extend the algebraic method of group classification
to more general classes of differential equations with weaker normalization properties.

\begin{definition}\label{def:Semi-normClass} 
Let $\mathcal G^H$ be the action groupoid of a subgroup $H$ of~$G^\sim$.
Suppose that a family
$N_{\mathcal S}:=\{N_\theta<G_\theta\mid\theta\in\mathcal S\}$
of subgroups of the point symmetry groups~$G_\theta$
with the associated subgroups $\mathcal N_\theta:=\{(\theta,\Phi,\theta)\mid\Phi\in N_\theta\}$
of the vertex groups~$\mathcal G_\theta$ for each $\theta\in\mathcal S$
satisfies the property $\mathcal N_{{\rm s}(\mathcal T)}\star\mathcal T=\mathcal T\star\mathcal N_{{\rm t}(\mathcal T)}$
for any $\mathcal T\in\mathcal G^H$.
In other words, the family~$N_{\mathcal S}$ is \emph{uniform} with respect to the action of~$H$ on~$\mathcal G^\sim$,
or \emph{$H$-uniform} for short. 
Then the Frobenius product
\[
\mathcal N^{\rm f}\star\mathcal G^H=\big\{\mathcal T\star\mathcal T'\mid
\mathcal T\in\mathcal N^{\rm f},\,\mathcal T'\in\mathcal G^H,\,{\rm t}(\mathcal T)={\rm s}(\mathcal T')\big\}
\]
with $\mathcal N^{\rm f}:=\sqcup_{\theta\in\mathcal S}\mathcal N_\theta$
is a subgroupoid of~$\mathcal G^\sim$,
which coincides with the image of~$\mathcal N^{\rm f}$
under the ${\rm s}$-action (resp.\ the ${\rm t}$-action) of~$H$ on~$\mathcal G^\sim$.
If $\mathcal N^{\rm f}\star\mathcal G^H=\mathcal G^\sim$,
we call the class~$\mathcal L|_{\mathcal S}$
\emph{semi-normalized with respect to the subgroup $H$ of~$G^\sim$
and the family $N_{\mathcal S}$ of subgroups of the point symmetry groups}, 
or \emph{$(H,N_{\mathcal S})$-semi-normalized} for brevity. 
\end{definition}

\begin{remark}%\label{rem:SeminormInGenExtSense}
Since $G^\sim$ is the usual equivalence group of the class~$\mathcal L|_{\mathcal S}$, 
we can say more precisely that this class is $(H,N_{\mathcal S})$-semi-normalized
\emph{in the usual sense}.
The other kinds of semi-normalization, 
in the generalized sense, in the extended sense or in the generalized extended sense, 
which are associated with equivalence groups of the respective kinds, 
will be introduced elsewhere. 
\end{remark}

The case $\mathcal G^H=\mathcal S\rightrightarrows\mathcal S$ is degenerate;
then necessarily $N_\theta=G_\theta$ for any $\theta\in\mathcal S$, 
and thus $\mathcal G^\sim=\mathcal G^{\rm f}$ and $\pi_*G^\sim=G^\cap$,
which is not of interest in the sense of group classification.
At the same time, the other improper case,
where the corresponding subgroup~$H$ of~$G^\sim$ coincides with the entire equivalence group~$G^\sim$,
is important.
The class~$\mathcal L|_{\mathcal S}$ that is semi-normalized with respect to the group $H=G^\sim$
and the family $N_{\mathcal S}=\big\{N_\theta=\{{\rm id}\}\mid\theta\in\mathcal S\big\}$ is genuinely normalized 
since then $\mathcal G^\sim=\mathcal G^{G^\sim}$.
Here and in what follows ${\rm id}$ denotes the identity transformation in the space with the coordinates $(x,u)$.
If the class~$\mathcal L|_{\mathcal S}$ is semi-normalized with respect to the group $H=G^\sim$
and the family $N_{\mathcal S}=\{N_\theta=G_\theta\mid\theta\in\mathcal S\}$,
then it is literally \emph{semi-normalized}.
It is obvious that a normalized class 
or a class that is semi-normalized with respect to a subgroup of its equivalence group 
and a family of subgroups of the point symmetry groups of systems from this class
is semi-normalized in the above specific sense.
More generally, if $H=G^\sim$, then the class~$\mathcal L|_{\mathcal S}$ is called
semi-normalized with respect to the family $N_{\mathcal S}$ of subgroups of the point symmetry groups,
i.e., we avoid denominating $H$ in this specific situation. 
Moreover, we contract the term 
``a class that is semi-normalized with respect to a subgroup of its equivalence group 
and a family of subgroups of the point symmetry groups of systems from this class'', 
when talking about such classes in general, 
and just call them semi-normalized.

\begin{theorem}\label{thm:EquivGroupoidOfSemi-normClasses}
Suppose that the class~$\mathcal L|_{\mathcal S}$ is semi-normalized with respect to the subgroup $H$ of~$G^\sim$
and the symmetry-group family $N_{\mathcal S}=\{N_\theta\mid\theta\in\mathcal S\}$.

(i) Then $\mathcal G^\sim=\mathcal N^{\rm f}\star\mathcal G^H=\mathcal G^H\star\mathcal N^{\rm f}$,
$\mathcal N^{\rm f}$ is a normal subgroupoid of $\mathcal G^\sim$
and $\mathcal N_{{\rm s}(\mathcal T)}\star\mathcal T=\mathcal T\star\mathcal N_{{\rm t}(\mathcal T)}$
for any $\mathcal T\in\mathcal G^\sim$.

(ii) Furthermore, for each $\theta\in\mathcal S$, $N_\theta$ is a normal subgroup of~$G_\theta$,
$G^{\rm ess}_\theta:=G_\theta\cap\pi_*H$ is a subgroup of~$G_\theta$,
and the group~$G_\theta$ is the Frobenius product %https://en.wikipedia.org/wiki/Product_of_group_subsets
of~$G^{\rm ess}_\theta$ and~$N_\theta$, $G_\theta=G^{\rm ess}_\theta N_\theta$.

(iii) Systems from the class~$\mathcal L|_{\mathcal S}$ are similar (i.e., they are related by point transformations)
if and only if they are $H$-equivalent.
\end{theorem}

\begin{proof}
(i) The $(H,N_{\mathcal S})$-semi-normalization of the class~$\mathcal L|_{\mathcal S}$ means that
any element $\mathcal T\in\mathcal G^\sim$
can be represented in the form $\mathcal T=\mathcal T_1\star\mathcal T_2$
for some $\mathcal T_1\in\mathcal N_{{\rm s}(\mathcal T)}$
and some \smash{$\mathcal T_2\in\mathcal G^H\cap\mathcal G\big({\rm s}(\mathcal T),{\rm t}(\mathcal T)\big)$}.
Due to the uniformity of~$N_{\mathcal S}$, there exists \smash{$\tilde{\mathcal T}_1\in\mathcal N_{{\rm t}(\mathcal T)}$}
such that \smash{$\mathcal T_1\star\mathcal T_2=\mathcal T_2\star\tilde{\mathcal T}_1$}.
In other words, any admissible transformation $\mathcal T\in\mathcal G^\sim$
can also be represented in the form $\mathcal T=\mathcal T_2\star\tilde{\mathcal T}_1$
for some \smash{$\mathcal T_2\in\mathcal G^H\cap\mathcal G\big({\rm s}(\mathcal T),{\rm t}(\mathcal T)\big)$}
and some \smash{$\tilde{\mathcal T}_1\in\mathcal N_{{\rm t}(\mathcal T)}$}.

It is obvious that $\mathcal N^{\rm f}$ is a wide subgroupoid of $\mathcal G^\sim$ (as well as $\mathcal G^{\rm f}$),
and it coincides with its fundamental subgroupoid.
In the notation of the previous paragraph, for any $\mathcal T\in\mathcal G^\sim$ we have
\[
\mathcal N_{{\rm s}(\mathcal T)}\star\mathcal T
=\mathcal N_{{\rm s}(\mathcal T)}\star\mathcal T_2\star\tilde{\mathcal T}_1
=\mathcal T_2\star\mathcal N_{{\rm t}(\mathcal T)}\star\tilde{\mathcal T}_1
=\mathcal T_2\star\tilde{\mathcal T}_1\star\mathcal N_{{\rm t}(\mathcal T)}
=\mathcal T\star\mathcal N_{{\rm t}(\mathcal T)},
\]
which directly implies that $\mathcal N^{\rm f}$ is a normal subgroupoid of $\mathcal G^\sim$.
In other words, under the condition $\mathcal G^\sim=\mathcal N^{\rm f}\star\mathcal G^H$,
the uniformity of~$N_{\mathcal S}$ with respect to the actions of~$H$ on~$\mathcal G^\sim$
implies the uniformity of~$N_{\mathcal S}$ with respect to the similar actions of the entire group~$G^\sim$
or even of the entire groupoid~$\mathcal G^\sim$.

(ii) We fix an arbitrary $\theta\in\mathcal S$ and take an arbitrary~$\Phi\in G_\theta$.
Then $(\theta,\Phi,\theta)\in\mathcal G^\sim$ and,
since $\mathcal G^H\star\mathcal N^{\rm f}=\mathcal G^\sim$,
the point transformation~$\Phi$ admits the factorization $\Phi=(\pi_*\mathscr T)\circ\Phi_1$
for some $\mathscr T\in H$ and some $\Phi_1\in N_\theta$.
The element~$N_\theta$ of the family~$\mathcal N_{\mathcal S}$ is a subgroup of~$G_\theta$, $N_\theta<G_\theta$,
and hence the point transformation $\Phi_0:=\pi_*\mathscr T=\Phi\circ\Phi_1^{-1}$
also belongs to~$G_\theta$, and consequently to $G_\theta\cap\pi_*H=:G^{\rm ess}_\theta$,
which is a group as intersection of two groups and is thus a subgroup of~$G_\theta$.
This implies that for any~$\Phi\in G_\theta$ we have the representation $\Phi=\Phi_0\circ\Phi_1$,
where $\Phi_0\in G^{\rm ess}_\theta$ and $\Phi_1\in N_\theta$, i.e.,
$G_\theta\subseteq G^{\rm ess}_\theta N_\theta$.
It is obvious that $G_\theta\supseteq G^{\rm ess}_\theta N_\theta$ as well,
which implies $G_\theta= G^{\rm ess}_\theta N_\theta$.

The fact that $\mathcal N^{\rm f}$ is a normal subgroupoid of $\mathcal G^\sim$
directly implies that $\mathcal N_\theta$ is a normal subgroup of~$\mathcal G_\theta$, i.e.,
$N_\theta$ is a normal subgroup of~$G_\theta$.
We can also immediately prove the last claim in terms of groups.
For an arbitrary~$\Phi\in G_\theta$ and an arbitrary~$\tilde\Phi\in N_\theta$,
we consider the composition $\Phi\circ\tilde\Phi\circ\Phi^{-1}$.
As an element of~$G_\theta$, the transformation~$\Phi$ admits
the factorization $\Phi=\Phi_0\circ\Phi_1$
with some $\Phi_0\in G^{\rm ess}_\theta$ and some $\Phi_1\in N_\theta$.
Since $G^{\rm ess}_\theta<\pi_*H$, we have $(\theta,\Phi_0,\theta)\in\mathcal G^H$.
Due to the uniformity of the family~$N_{\mathcal S}$ with respect to the action of~$H$ on~$\mathcal G^\sim$,
we obtain $N_\theta=\Phi_0\circ N_\theta\circ\Phi_0^{-1}$.
Hence the composition
$\Phi\circ\tilde\Phi\circ\Phi^{-1}=
\Phi_0\circ\Phi_1\circ\tilde\Phi\circ\Phi_1^{-1}\circ\Phi_0^{-1}$ belongs to~$N_\theta$.
Thus we have that $N_\theta$ is a normal subgroup of~$G_\theta$, $N_\theta\unlhd G_\theta$.

(iii) If the systems $\mathcal L_\theta$ and~$\mathcal L_{\theta'}$ from the class~$\mathcal L|_{\mathcal S}$
are related  by a point transformation~$\Phi$ in the space with coordinates~$(x,u)$,
$\mathcal L_{\theta'}=\Phi_*\mathcal L_\theta$,
then $\mathcal T=(\theta,\Phi,\theta')$ is an admissible transformation of this class, $\mathcal T\in\mathcal G^\sim$.
Hence there exist $\mathcal T_1\in\mathcal N_{\theta}$ and $\mathcal T_2\in\mathcal G^H\cap\mathcal G(\theta,\theta')$
such that $\mathcal T=\mathcal T_1\star\mathcal T_2$.
The existence of $\mathcal T_2\in\mathcal G^H\cap\mathcal G(\theta,\theta')$ means that
the systems~$\mathcal L_\theta$ and~$\mathcal L_{\theta'}$ are $H$-equivalent.
\end{proof}

The member~$N_\theta$ of the family $\mathcal N_{\mathcal S}$ 
and the subgroup~$G^{\rm ess}_\theta:=G_\theta\cap\pi_*H$ of~$G_\theta$ 
are called the \emph{uniform point-symmetry group} 
and the \emph{essential point-symmetry group} of the system~$\mathcal L_\theta$ 
that are associated with the $(H,N_{\mathcal S})$-semi-normalization of the class~$\mathcal L|_{\mathcal S}$, 
respectively.
The knowledge of a family of uniform point-symmetry groups trivializes them in the following sense:
since in view of Theorem~\ref{thm:EquivGroupoidOfSemi-normClasses}(ii)
the group~$G_\theta$ is factorized over~$N_\theta$ for each~$\theta$,
then we only need to find the subgroup~$G^{\rm ess}_\theta$ in order to construct~$G_\theta$.
If either $N_\theta=G_\theta$ for any $\theta\in\mathcal S$ or $N_\theta=\{{\rm id}\}$ for any $\theta\in\mathcal S$,
then the factorization is trivial.
Theorem~\ref{thm:EquivGroupoidOfSemi-normClasses}(i) implies that 
\emph{the family~$N_{\mathcal S}$ is not only $H$-uniform but also $G^\sim$-uniform and even $\mathcal G^\sim$-uniform}.
We can interpret each of these uniformities
as the equivariance of the map $\mathcal S\ni\theta\mapsto (\theta,N_\theta)$
under the analogous action on~$\mathcal S$
and on $\{(\theta,K_\theta)\mid \theta\in\mathcal S,\,K_\theta\leq G_\theta\}$. 

It also follows from Theorem~\ref{thm:EquivGroupoidOfSemi-normClasses}(i) 
that the condition $\mathcal G^\sim=\mathcal N^{\rm f}\star\mathcal G^H$ in Definition~\ref{def:Semi-normClass}
can be replaced by the condition $\mathcal G^\sim=\mathcal G^H\star\mathcal N^{\rm f}$
or even by the weaker condition $\mathcal G^\sim=\mathcal N^{\rm f}\star\mathcal G^H\star\mathcal N^{\rm f}$
since these conditions are equivalent if the family~$N_{\mathcal S}$ is uniform with respect to the action of~$H$ on~$\mathcal G^\sim$.
Roughly speaking, each of these conditions means that the entire equivalence groupoid~$\mathcal G^\sim$
is generated by distinguished equivalence transformations and transformations from uniform point symmetry groups.

Given a $(H,N_{\mathcal S})$-semi-normalized class~$\mathcal L|_{\mathcal S}$, 
one can extend the subgroup~$H$ of~$G^\sim$ and elements of~$N_{\mathcal S}$ 
in such a way that this class will be semi-normalized with respect to the extended objects.  

\begin{theorem}\label{thm:ExdendingSemi-normalization}
Let the class~$\mathcal L|_{\mathcal S}$ be semi-normalized with respect to the subgroup $H$ of~$G^\sim$
and the symmetry-group family $N_{\mathcal S}=\{N_\theta\mid\theta\in\mathcal S\}$, 
and let $G^\cap$ denote the kernel group of~$\mathcal L|_\mathcal S$. 
Then the class~$\mathcal L|_\mathcal S$ is 
$(\bar H,N_{\mathcal S})$- and $(\bar H,\bar N_{\mathcal S})$-semi-normalized, 
in particular
$(G^\sim,N_{\mathcal S})$- and $(G^\sim,\bar N_{\mathcal S})$-semi-normalized, 
where $\bar N_{\mathcal S}:=\{N_\theta G^\cap\mid N_\theta\in N_{\mathcal S}\}$ 
with $N_\theta G^\cap$ denoting the Frobenius product of $N_\theta\unlhd G\theta$ and $G^\cap\leq G\theta$
and $\bar H$ is any subgroup of~$G^\sim$ that contains~$H$, $H\leq\bar H\leq G^\sim$.
\end{theorem}

\begin{proof}
For each $\theta\in\mathcal S$,
the uniform symmetry group~$N_\theta$ and the kernel symmetry group~$G^\cap$
are a normal subgroup and a subgroup of~$G_\theta$
in view of Theorem~\ref{thm:EquivGroupoidOfSemi-normClasses}
and by the definition of~$G^\cap$, respectively.
Hence the Frobenius product $\bar N_\theta:=G^\cap N_\theta$ is a subgroup~$G_\theta$ for any $\theta\in\mathcal S$.

In view of Theorem~\ref{thm:EquivGroupoidOfSemi-normClasses}(i), 
the symmetry-subgroup family~$\mathcal N_{\mathcal S}$ 
is uniform with respect to any subgroup~$\bar H$ of~$G^\sim$. 
We show that this is also the case for the symmetry-subgroup family~$\bar N_{\mathcal S}$.

We take arbitrary $\mathcal T=(\theta,\bar\Psi,\theta')\in\mathcal G^{\bar H}$
and $\bar\Phi\in\bar N_\theta$.
By the definition of~$\bar N_\theta$,
there exists $\Xi\in G^\cap$ and $\Phi\in N_\theta$ such that $\bar\Phi=\Phi\circ\Xi$.
The uniformity of~$N_{\mathcal S}$ with respect to the action of~$\bar H$ on~$\mathcal G^\sim$
implies that there exists $\Phi'\in N_{\theta'}$ such that $\bar\Psi\circ\Phi=\Phi'\circ\bar\Psi$.
Since $G^\cap\unlhd\pi_*G^\sim$ and $\bar H\leq G^\sim$,
we have $\bar\Psi\circ\Xi=\Xi'\circ\bar\Psi$ for some $\Xi'\in G^\cap$.
As a result,
\begin{gather*}
\bar\Psi\circ\bar\Phi
=\bar\Psi\circ\Phi\circ\Xi
=\Phi'\circ\bar\Psi\circ\Xi
=\Phi'\circ\Xi'\circ\bar\Psi
=\bar\Phi'\circ\bar\Psi,
\end{gather*}
where $\bar\Phi':=\Phi'\circ\Xi'\in\bar N_{\theta'}$.

It is obvious that 
$\mathcal G^\sim=\mathcal N^{\rm f}\star\mathcal G^H
\subseteq\mathcal N^{\rm f}\star\mathcal G^{\bar H}
\subseteq\bar{\mathcal N}^{\rm f}\star\mathcal G^{\bar H}
\subseteq\mathcal G^\sim$ if $H\leq\bar H$.
As a result, $\mathcal G^\sim=\mathcal N^{\rm f}\star\mathcal G^{\bar H}=\bar{\mathcal N}^{\rm f}\star\mathcal G^{\bar H}$.
The two limit cases are $\bar H=H$ and $\bar H=G^\sim$.
\end{proof}

We call the claim from Theorem~\ref{thm:EquivGroupoidOfSemi-normClasses}(ii)
on the factorization $G_\theta=G^{\rm ess}_\theta N_\theta$
the \emph{theorem on factoring out point-symmetry groups within semi-normalized classes}.
The infinitesimal version of this claim
can be called the \emph{theorem on factoring out maximal Lie invariance algebras within semi-normalized classes}.
This version follows immediately from Theorem~\ref{thm:EquivGroupoidOfSemi-normClasses}
if we replace the groups by the corresponding algebras of generators of the one-parameter subgroups of these groups.

\begin{theorem}\label{thm:OnInvAlgsInSemi-normClasses}
Suppose that a class of differential equations~$\mathcal L|_{\mathcal S}$ is semi-normalized
with respect to a subgroup~$H$ of~$G^\sim$
and a family of symmetry subgroups~$\mathcal N_{\mathcal S}=\{N_\theta\mid\theta\in\mathcal S\}$,
and $\mathfrak h$ is a subalgebra of~$\mathfrak g^\sim$ that corresponds to the subgroup~$H$.
Then, for each $\theta\in\mathcal S$, the Lie algebras~$\mathfrak g^{\rm ess}_\theta$ and~$\mathfrak n_\theta$
that are associated with the groups~$G^{\rm ess}_\theta$ and~$N_\theta$
are a subalgebra and an ideal of the maximal Lie invariance algebra~$\mathfrak g_\theta$
of the system~$\mathcal L_\theta\in\mathcal L|_{\mathcal S}$, respectively.
Moreover, the algebra~$\mathfrak g_\theta$ is the sum of~$\mathfrak g^{\rm ess}_\theta$ and~$\mathfrak n_\theta$,
$\mathfrak g_\theta=\mathfrak g^{\rm ess}_\theta+\mathfrak n_\theta$,
and $\mathfrak g^{\rm ess}_\theta=\mathfrak g_\theta\cap\pi_*\mathfrak h$.
\end{theorem}

\section{Disjoint and non-disjoint semi-normalization}\label{sec:DisjointSemi-normalization}

\begin{definition}\label{def:DisjointedlySemi-normClass}
Under the conditions of Definition~\ref{def:Semi-normClass},
we call the class~$\mathcal L|_{\mathcal S}$
\emph{disjointedly semi-normalized with respect to the subgroup $H$ of~$G^\sim$
and the family $N_{\mathcal S}$ of subgroups of the point symmetry groups},
or \emph{disjointedly $(H,N_{\mathcal S})$-semi-normalized} for short, 
if in addition ${N_\theta\cap\pi_*H=\{{\rm id}\}}$ for any $\theta\in\mathcal S$. 
\end{definition}

If the class~$\mathcal L|_{\mathcal S}$ is disjointedly $(H,N_{\mathcal S})$-semi-normalized, 
then $\mathcal N^{\rm f}\cap\mathcal G^H=\mathcal S\rightrightarrows\mathcal S$.%
\footnote{% 
Indeed, let $\mathcal T\in\mathcal N^{\rm f}\cap\mathcal G^H$. 
Then $\mathcal T=(\theta,\Phi,\theta)=(\theta,\pi_*\mathscr T,\mathscr T_*\theta)$ 
for some $\theta\in\mathcal S$, $\Phi\in G_\theta$ and $\mathscr T\in H$. 
Hence $\Phi=\pi_*\mathscr T\in N_\theta\cap\pi_*H$, and thus $\Phi={\rm id}$, 
i.e., $\mathcal T\in\mathcal S\rightrightarrows\mathcal S$. 
Therefore, 
$\mathcal S\rightrightarrows\mathcal S\subseteq\mathcal N^{\rm f}\cap\mathcal G^H
\subseteq\mathcal S\rightrightarrows\mathcal S$, 
which means that $\mathcal N^{\rm f}\cap\mathcal G^H=\mathcal S\rightrightarrows\mathcal S$.
}
The inverse implication does not hold in general%
\footnote{%
Let $\mathcal E$ be the class of the (1+1)-dimensional linear evolution equations %of the form 
${\rm e}^\beta u_t-{\rm e}^\beta u_{xx}+(x+\alpha){\rm e}^{x^3+\beta}u=0$, 
where all variables and parameters are real, 
the arbitrary-element tuple is $\theta=(\alpha,\beta)$ with arbitrary constants~$\alpha$ and~$\beta$. 
The arbitrary element~$\beta$ is artificial and inessential and can be set to be equal to zero 
by gauge equivalence transformations. 
For the class $\mathcal L|_{\mathcal S}=\mathcal E$, we have $\mathcal S=\mathbb R^2$ and
$\mathcal G^\sim=\big\{(\theta,\Phi,\tilde\theta)\mid\tilde\alpha=\alpha,\,\Phi\in G_\theta\big\} 
=\mathcal G^{\rm f}\star\mathcal G^{\rm g}$ 
with the action groupoid 
\smash{$\mathcal G^{\rm g}=\big\{(\theta,{\rm id},\tilde\theta)\mid\tilde\alpha=\alpha\big\}$} 
of the gauge equivalence group~$G^{\rm g}$ of the class~$\mathcal E$. 
The equivalence group~$G^\sim$ of this class
is constituted by the point transformations
$\tilde t=t+\delta_0$, $\tilde x=x$, $\tilde u=\delta_1u$, 
$\tilde\alpha=\alpha$, $\tilde\beta=F(\alpha,\beta)$, 
where $\delta_0$ and~$\delta_1$ are arbitrary constants with $\delta_1\ne0$, 
and $F$ is an arbitrary sufficiently smooth function of $(\alpha,\beta)$ with $F_\beta\ne0$.
The group~$G^{\rm g}$ is the normal subgroup of~$G^\sim$ singled out by the constraints $\delta_0=0$ and $\delta_1=1$.
For any $\theta\in\mathcal S$, the group~$G_\theta$ consists of the point transformations
$\tilde t=t+\delta_0$, $\tilde x=x$, $\tilde u=\delta_1u+h(t,x)$, 
where $\delta_0$ and~$\delta_1$ are arbitrary constants with $\delta_1\ne0$, 
and $h$ is an arbitrary solution of~$\mathcal E_\theta$.
We take the symmetry-subgroup family~$\mathcal N_{\mathcal S}=\{N_\theta=G_\theta\mid\theta\in\mathcal S\}$ 
and the subgroup~$H$ of~$G^\sim$ constituted by the transformations 
$\tilde t=t+\delta_0$, $\tilde x=x$, $\tilde u=u$, $\tilde\alpha=\alpha$, $\tilde\beta=\beta+\delta_0$.
The class~$\mathcal L|_{\mathcal S}$ is $(H,N_{\mathcal S})$-semi-normalized, 
$\mathcal N^{\rm f}\cap\mathcal G^H=\mathcal S\rightrightarrows\mathcal S$ and 
$N_\theta\cap\pi_*H=\pi_*H\ne\{{\rm id}\}$ for any $\theta\in\mathcal S$. 
} 
but is obvious if the class~$\mathcal L|_{\mathcal S}$ does not admit gauge admissible transformations.%
\footnote{% 
Fix an arbitrary $\theta\in\mathcal S$.
Let $\Phi\in N_\theta\cap\pi_*H$, i.e., $\Phi\in N_\theta$ and $\Phi=\pi_*\mathscr T$ for some $\mathscr T\in H$.
Then 
$\mathcal T_1:=(\theta,\Phi,\theta)\in\mathcal N^{\rm f}$,
$\mathcal T_2:=(\theta,\Phi,\mathscr T_*\theta)\in\mathcal G^H$,
$\mathcal T_1{}^{-1}\star\mathcal T_2=(\theta,{\rm id},\mathscr T_*\theta)\in\mathcal G^\sim$
and thus $\mathscr T_*\theta=\theta$ 
since the class~$\mathcal L|_{\mathcal S}$ does not admit gauge admissible transformations.
Therefore, $\mathcal T_1=\mathcal T_2\in\mathcal N^{\rm f}\cap\mathcal G^H=\mathcal S\rightrightarrows\mathcal S$, 
which means that $\Phi={\rm id}$. 
As a result, ${N_\theta\cap\pi_*H=\{{\rm id}\}}$ for any $\theta\in\mathcal S$. 
} 

In particular, the class~$\mathcal L|_{\mathcal S}$ that is literally normalized
is in fact disjointedly semi-normalized with respect to the group $H=G^\sim$
and the family $N_{\mathcal S}=\big\{N_\theta=\{{\rm id}\}\mid\theta\in\mathcal S\big\}$. 
It is also clear that each disjointedly semi-normalized class is literally semi-normalized.

Disjointedly semi-normalized classes have nicer properties 
than those presented in Theorem~\ref{thm:EquivGroupoidOfSemi-normClasses}
for general semi-normalized classes.

\begin{theorem}\label{thm:DisjontedlylySemi-normClasses}
If the class~$\mathcal L|_{\mathcal S}$ is disjointedly semi-normalized
with respect to a subgroup~$H$ of~$G^\sim$
and a symmetry-subgroup family~$\mathcal N_{\mathcal S}=\{N_\theta\mid\theta\in\mathcal S\}$,
then

(i)
any admissible transformation $\mathcal T\in\mathcal G^\sim$
possesses a \emph{unique} decomposition of the form $\mathcal T=\mathcal T_1\star\mathcal T_2$
with $\mathcal T_1\in\mathcal N_{{\rm s}(\mathcal T)}$
and \smash{$\mathcal T_2\in\mathcal G^H\cap\mathcal G\big({\rm s}(\mathcal T),{\rm t}(\mathcal T)\big)$}
as well as one of the form $\mathcal T=\mathcal T_2\star\tilde{\mathcal T}_1$ with the same~$\mathcal T_2$
and with \smash{$\tilde{\mathcal T}_1\in\mathcal N_{{\rm t}(\mathcal T)}$};

(ii)
for each $\theta\in\mathcal S$, the point symmetry group~$G_\theta$ of the system~$\mathcal L_\theta\in\mathcal L|_{\mathcal S}$
splits over~$N_\theta$;
more specifically, $N_\theta$ is a normal subgroup of~$G_\theta$,
$G^{\rm ess}_\theta:=G_\theta\cap\pi_*H$ is a subgroup of~$G_\theta$,
and the group~$G_\theta$ is the semidirect product of~$G^{\rm ess}_\theta$ acting on~$N_\theta$,
$G_\theta=G^{\rm ess}_\theta\ltimes N_\theta$;

(iii)
for each $\theta\in\mathcal S$, the Lie algebras~$\mathfrak g^{\rm ess}_\theta$ and~$\mathfrak n_\theta$
that are associated with the groups~$G^{\rm ess}_\theta$ and~$N_\theta$
are a subalgebra and an ideal of the maximal Lie invariance algebra~$\mathfrak g_\theta$
of the system~$\mathcal L_\theta\in\mathcal L|_{\mathcal S}$, respectively;
moreover, the algebra~$\mathfrak g_\theta$ is the semidirect sum
$\mathfrak g_\theta=\mathfrak g^{\rm ess}_\theta\lsemioplus\mathfrak n_\theta$,
and $\mathfrak g^{\rm ess}_\theta=\mathfrak g_\theta\cap\pi_*\mathfrak h$,
where $\mathfrak h$ is a subalgebra of the equivalence algebra~$\mathfrak g^\sim$ of~$\mathcal L|_{\mathcal S}$
that corresponds to the subgroup~$H$.
\end{theorem}

\begin{proof}
(i)
Taking into account the proof of Theorem~\ref{thm:EquivGroupoidOfSemi-normClasses},
it is only required to show the uniqueness of the first decomposition.
The uniqueness of the second decomposition is proved in the same way.
We fix an arbitrary admissible transformation $\mathcal T$ of the class~$\mathcal L|_{\mathcal S}$
and consider two decompositions of the first kind for~$\mathcal T$,
$\mathcal T=\mathcal T_1\star\mathcal T_2=\mathcal T_3\star\mathcal T_4$,
where $\mathcal T_1,\mathcal T_3\in\mathcal N^{\rm f}$ and $\mathcal T_2,\mathcal T_4\in\mathcal G^H$.
Then $\mathcal T_3^{-1}\star\mathcal T_1=\mathcal T_4\star\mathcal T_2^{-1}$.
Since $\mathcal N^{\rm f}$ and~$\mathcal G^H$ are subgroupoids of~$\mathcal G^\sim$,
the left- and right-hand sides of the last equality belong to $\mathcal N^{\rm f}$ and~$\mathcal G^H$,
respectively.
This means that they both belong to $\mathcal N^{\rm f}\cap\mathcal G^H=\mathcal S\rightrightarrows\mathcal S$.
Hence $\mathcal T_1=\mathcal T_3$ and $\mathcal T_2=\mathcal T_4$.

(ii)
In view of the claim on factoring out the groups~$G_\theta$ from Theorem~\ref{thm:EquivGroupoidOfSemi-normClasses},
it suffices to show that $G^{\rm ess}_\theta\cap N_\theta=\{{\rm id}\}$,
but this is a direct consequence of the disjointness property ${N_\theta\cap\pi_*H=\{{\rm id}\}}$
since
\[
G^{\rm ess}_\theta\cap N_\theta
=(G_\theta\cap\pi_*H)\cap N_\theta
=G_\theta\cap(\pi_*H\cap N_\theta)
=G_\theta\cap\{{\rm id}\}
=\{{\rm id}\}.
\]
Thus, $G_\theta$ splits over~$N_\theta$ for each~$\theta\in\mathcal S$.

(iii)
This item is the infinitesimal version of item (ii)
and follows immediately from (ii)
via replacing the groups with the corresponding algebras
of generators of the one-parameter subgroups of these groups.
It can also be considered as a consequence of Theorem~\ref{thm:OnInvAlgsInSemi-normClasses}.
\end{proof}

We would like to emphasize that augmenting the semi-normalization with the disjointedness 
in Theorem~\ref{thm:DisjontedlylySemi-normClasses} additionally gives
the uniqueness of the decompositions in item (i),
replacing the Frobenius product by the semidirect product in item (ii)
and replacing the sum of vector spaces by the semidirect sum of algebras in item (iii).
We call item (ii) the \emph{theorem on splitting point-symmetry groups in disjointedly semi-normalized classes}
and  item (iii) the \emph{theorem on splitting of maximal Lie invariance algebras in disjointedly semi-normalized classes}.
In contrast to Theorem~\ref{thm:DisjontedlylySemi-normClasses}(ii), 
the product of~$G^{\rm ess}_\theta$ and~$N_\theta$ in Theorem~\ref{thm:EquivGroupoidOfSemi-normClasses}(ii)
is in general a product of subgroups with nontrivial intersection as a product of group subsets
but not a semidirect product of subgroups
since the condition $G^{\rm ess}_\theta\cap N_\theta=\{{\rm id}\}$ necessarily holds for any $\theta\in\mathcal S$
only in disjointedly semi-normalized classes.

\begin{remark}\label{rem:igammax subclass}
\emph{There are classes of differential equations that are disjointedly semi-normalized
only with respect to proper subgroups of the corresponding equivalence groups.
}%
The first example of such a class was constructed in~\cite[Section~7]{kuru2018a}
in the course of group classification of (1+1)-dimensional linear Schr\"odinger equations
with complex-valued potentials depending on~$(t,x)$. 
Here we reinterpret it in the context of the above claim. 
More specifically, let $n=1$ in~\eqref{MLinSchEqs}.
Consider the subclass of the class~$\mathcal F$ that consists of the equations of the form~\eqref{MLinSchEqs},
where the potential~$V$ is purely imaginary and (homogeneously) linear with respect to~$x$.
This subclass can be reparameterized to the class~$\mathcal E$ of equations of the form
\begin{gather*}%\label{igammax subclass}
\mathcal E_\gamma\colon\quad i\psi_t+\psi_{xx}+i\gamma(t)x\psi=0,
\end{gather*}
where $\gamma$ is an arbitrary smooth real-valued function of $t$,
which is the only arbitrary element.
\begin{subequations}\label{equgroupoidigamamx subclass}
The equivalence groupoid~$\mathcal G^\sim_{\mathcal E}$ of the class~$\mathcal E$
consists of triples of the form $(\gamma,\Phi,\tilde\gamma)$,
where $\Phi$ is a point transformation in the space of variables with components
\begin{gather}\label{equgroupoidigamamx subclass_a}
\tilde t=T:=\frac{a_1t+a_0}{a_3t+a_2}, \quad
\tilde x=\varepsilon|T_t|^{1/2}x+b_1T+b_0,
\\\label{equgroupoidigamamx subclass_b}
\tilde \psi=\frac c{|T_t|^{1/4}}\exp\left(
\frac i8\frac{T_{tt}}{|T_t|}x^2+\frac i2\varepsilon b_1|T_t|^{1/2}x-\varepsilon\int\gamma\frac{b_1T+b_0}{|T_t|^{1/2}}\,{\rm d}t+i\frac{b_1^2}4 T
\right)(\hat\psi+\hat\Theta),
%\tilde \psi=c\exp\left(-\frac{ia_3}{4(a_3t+a_2)}x^2+\frac i2\varepsilon b_1|T_t|^{1/2}x-\varepsilon\int \gamma\frac{b_1T+b_0}{|T_t|^{1/2}}\,{\rm d}t+i\frac{b_1^2}4 T\right)\times{}
%\nonumber\\[.5ex]\phantom{\tilde \psi={}}{}\times|a_3t+a_2|^{1/2}(\hat\psi+\hat\Theta),
\end{gather}
the relation between~$\gamma$ and~$\tilde\gamma$ is given by
\begin{gather}\label{equgroupoidigamamx subclass_c}
\tilde \gamma=\frac{\varepsilon\varepsilon'}{|T_t|^{3/2}}\gamma,
\end{gather}
\end{subequations}
$a_0$, $a_1$, $a_2$, $a_3$, $b_0$ and $b_1$ are arbitrary real constants with
$a_1a_2-a_0a_3=:\varepsilon'=\pm 1$,
$c$ is a nonzero complex constant,
$\Theta=\Theta(t,x)$ is an arbitrary solution of the initial equation~$\mathcal E_\gamma$,
and $\varepsilon=\pm 1$.
The (usual) equivalence group~$G^\sim_{\mathcal E}$ of the class~$\mathcal E$
is constituted by point transformations in the extended space of $(t,x,\psi,\psi^*,\gamma)$
with components of the form~\eqref{equgroupoidigamamx subclass},
where $b_0=b_1=0$ and $\Theta=0$.
For each $\gamma$, the equation~$\mathcal E_\gamma$ admits
the group~$G^{\rm unf}_\gamma$ of point symmetry transformations of the form~\eqref{equgroupoidigamamx subclass_a}--\eqref{equgroupoidigamamx subclass_b}
with $T=t$ and $\varepsilon=1$.
Elements of the kernel symmetry group~$G^\cap_{\mathcal E}$ of the class~$\mathcal E$
additionally satisfy the constrains $b_0=b_1=0$ and $\Theta=0$,
i.e., it consists of the `scalings' of~$\psi$ with nonzero complex numbers,
$\mathrm S_c$: $\tilde t=t$, $\tilde x=x$, $\tilde\psi=c\psi$.
The class~$\mathcal E$ is disjointedly semi-normalized
with respect to the family $\mathcal N=\{G^{\rm unf}_\gamma\}$
and the subgroup~$H$ of~$G^\sim_{\mathcal E}$
singled out from~$G^\sim_{\mathcal E}$ by the constraint $c=1$.
Admissible transformations with $T(t)=t$, $c=1$ and $\varepsilon=1$
have no counterparts among equivalence transformations
and thus, for a proper interpretation within the framework of uniform semi-normalization,
their transformational parts must and can be included in uniform symmetry groups
for the corresponding values of~$\gamma$.
Each transformation $\mathrm S_c$ with $|c|=1$
is the commutator of such transformational parts with zero $\Theta$ for any value of~$\gamma$
and thus necessarily belong to any possible uniform symmetry groups of the class~$\mathcal E$. %%% CHECKED !!!
Moreover, the above transformational parts
involve antiderivatives of~$\gamma$ and of~$t\gamma$,
which are defined up to constant summands,
and thus it is also natural to include the transformations $\mathrm S_c$ with positive real~$c$
into possible uniform symmetry groups of the class~$\mathcal E$.
Therefore, the family $\mathcal N=\{G^{\rm unf}_\gamma\}$ is the only natural choice
for a family of uniform symmetry groups of the class~$\mathcal E$.
Since each of the transformations $\mathrm S_c$ belongs to the kernel symmetry group
of the class~$\mathcal E$,
which can be embedded in the equivalence group~$G^\sim_{\mathcal E}$ via trivial prolongation
of its transformations to the arbitrary element~$\gamma$ \cite[Proposition~3]{card2011a},
the intersection of~$\pi_*G^\sim_{\mathcal E}$
with uniform symmetry groups cannot coincide with $\{\rm id\}$.
This means that the class~$\mathcal E$ is not disjointedly semi-normalized
with respect to the entire equivalence group~\smash{$G^\sim_{\mathcal E}$}
although it is semi-normalized with respect to~\smash{$G^\sim_{\mathcal E}$}
and the family~$\mathcal N=\{G^{\rm unf}_\gamma\}$.
For any $\gamma$ we have
$G^{\rm unf}_\gamma\cap\pi_*G^\sim_{\mathcal E}=G^\cap_{\mathcal E}\ne\{{\rm id}\}$.
\end{remark}

\begin{remark}\label{rem:OnNondejointedlySemiNormClasses}
It is also clear that each disjointedly semi-normalized class is semi-normalized.
At the same time, \emph{there are semi-normalized classes that are not disjointedly semi-normalized},
as can be seen from the example of the class~$\mathcal H$ 
of nonlinear diffusion equations~$\mathcal H_f$ of the form $u_t=(f(u)u_x)_x$ with $f_u\ne0$ 
briefly discussed in a different terminology in~\cite[Section~3]{kuru2018a}.
Let us properly present this example with extended arguments. 
The equivalence group~$G^\sim$ of the class~$\mathcal H$ consists of the transformations
\[
\tilde t=\alpha_1t+\alpha_0,\quad
\tilde x=\beta_1x+\beta_0,\quad
\tilde u=\gamma_1u+\gamma_0,\quad
\tilde f=\alpha_1^{-1}\beta_1^{\,2}f,
\]
where $\alpha_0$, $\alpha_1$, $\beta_0$, $\beta_1$, $\gamma_0$ and~$\gamma_1$ 
are arbitrary real constants with $\alpha_1\beta_1\gamma_1\ne0$, 
see, e.g., \mbox{\cite[Eq.~(3.15)]{akha1991a}}. 
For our purpose, it is more convenient to introduce $\delta_1:=\alpha_1^{-1}\beta_1^{\,2}$ 
and reparameterize the group~$G^\sim$ as  
\begin{gather}\label{eq:NonlinDiffusionEqsEquivGroup}
\tilde t=\delta_1^{-1}\beta_1^{\,2}t+\alpha_0,\quad
\tilde x=\beta_1x+\beta_0,\quad
\tilde u=\gamma_1u+\gamma_0,\quad
\tilde f=\delta_1f,
\end{gather}
where $\alpha_0$, $\beta_0$, $\beta_1$, $\gamma_0$, $\gamma_1$ and~$\alpha_1$
are arbitrary real constants with $\beta_1\gamma_1\delta_1\ne0$. 
The kernel point symmetry group~$G^\cap:=\cap_{f_u\ne0}G_f$ of the class~$\mathcal H$, 
where $G_f$ denotes the point symmetry group of the equation~$\mathcal H_f$,
is constituted by the point transformations
\[
\tilde t=\beta_1^2t+\alpha_0,\quad
\tilde x=\beta_1x+\beta_0,\quad
\tilde u=u,
\]
where $\alpha_0$, $\beta_0$ and~$\beta_1$ are arbitrary constants with $\beta_1\ne0$.
The point symmetry group~$G_{f^0}$ of the nonlinear diffusion equation~$\mathcal H_{f^0}$ with $f^0(u)=u^{-4/3}$
consists of by the point transformations
\begin{gather}\label{eq:u-4/3DiffusionEqSymGroup}
\tilde t=\lambda_1t+\lambda_0,\quad
\tilde x=\frac{\alpha x+\beta}{\gamma x+\delta},\quad
\tilde u=\lambda_1^{3/4}(\gamma x+\delta)^3u,
\end{gather}
where $\alpha$, $\beta$, $\gamma$, $\delta$, $\lambda_1$ and $\lambda_0$ are arbitrary constants 
with $\alpha\delta-\beta\gamma=\pm1$ and $\lambda_1>0$, 
and, according to \cite{kova2023a,kova2023b},
we take the natural domain for each transformation of the form~\eqref{eq:u-4/3DiffusionEqSymGroup} 
and use the modified composition of such transformations that is based on completing by continuity. 
The group $G_{f^0}$ is isomorphic to ${\rm Aff}^+(\mathbb R)\times{\rm SL}^{\pm}(2,\mathbb R)$,
where ${\rm Aff}^+(\mathbb R)$ is the subgroup of~${\rm Aff}(\mathbb R)$ 
constituted by all invertible orientation-preserving affine transformations of the affine space~$\mathbb R$, 
${\rm Aff}^+(\mathbb R)\simeq\mathbb R\rtimes {\rm GL}^+(\mathbb R)=\mathbb R\rtimes\mathbb R_{>0}$.

Suppose that the class~$\mathcal H$ is disjointedly semi-normalized 
with respect to a subgroup~$H$ of~$G^\sim$ and 
a symmetry-subgroup family~$\mathcal N_{\mathcal S}=\{N_f\lhd G_f\mid f\in\mathcal S\}$.
The subgroup~$N_{f^0}$ of~$G_{f^0}$ should contain a transformation with a nonzero value of the parameter~$\gamma$.
Since in addition $N_{f^0}$ is a normal subgroup of~$G_{f^0}$, 
it then contains the subgroup of elements of~$G_{f^0}$ with $\lambda_0=0$, $\lambda_1=1$ and $\alpha\delta-\beta\gamma=1$.
If $N_{f^0}$ contains an element of~$G_{f^0}$ with $\lambda_0\ne0$, 
then it contains the subgroup of shifts with respect to~$t$, 
$(\tilde t,\tilde x,\tilde u)=(t+\lambda_0,x,u)$ with $\lambda_0$ running through~$\mathbb R$.
Then the subgroup~$K_{f^0}$ of $N_{f^0}$ singled out be the constraint $\gamma=0$ 
consists of the transformations 
$\tilde t=\lambda_1t+\lambda_0$, 
$\tilde x=\alpha^2x+\beta$, 
\smash{$\tilde u=\lambda_1^{3/4}\alpha^{-3}u$},
where the parameters~$\alpha$ and~$\beta$ freely run through~$\mathbb R$, 
the parameter~$\lambda_0$ either freely runs through~$\mathbb R$ or is equal to zero
and the parameter~$\lambda_1$ runs through a subgroup of the multiplicative group~$\mathbb R^+$ of positive real numbers. 
Denoting \smash{$\hat\lambda_1=\lambda_1^{1/2}$} and \smash{$\hat\alpha=\lambda_1^{3/4}\alpha^{-3}$}, 
we reparameterize elements of the subgroup~$K_{f^0}$ as 
\smash{$\tilde t=\hat\lambda_1^2t+\lambda_0$}, 
\smash{$\tilde x=\hat\lambda_1\hat\alpha^{-2/3}x+\beta$},
$\tilde u=\hat\alpha u$,
where the parameter tuple~\smash{$(\hat\alpha,\beta,\lambda_0,\hat\lambda_1)$} satisfies 
the same conditions as the parameter tuple~$(\alpha,\beta,\lambda_0,\lambda_1)$ does. 
In particular, the parameter~$\hat\lambda_1$ runs through a subgroup~$\hat P$ of~$\mathbb R^+$.

In the parameterization~\eqref{eq:NonlinDiffusionEqsEquivGroup} of~$G^\sim$, 
it is clear that varying the parameters~$\gamma_0$, $\gamma_1$ and~$\delta_1$
definitely leads to changing general values of the arbitrary element~$f$, 
and varying the other parameters does not affect~$f$. 
This is why the subgroup~$H$ necessarily contain a subgroup of~$G^\sim$ 
parameterized by $(\gamma_0,\gamma_1,\delta_1)\in\mathbb R^3$, 
where the other parameters~$\alpha_0$, $\beta_0$ and~$\beta_1$ 
are assumed to be functions of $(\gamma_0,\gamma_1,\delta_1)$, 
which are respectively equal to 0, 0, 1 at $(0,1,1)$. 
Since $(f^0,\Phi^\nu,f^0)\in\mathcal G^\sim$ for any $\nu\in\mathbb R$, 
where $\Phi^\nu$ denotes the scaling transformation $(\tilde t,\tilde x,\tilde u)=(\nu^2t,\nu x,u)$, 
the elements of~$H$ can be represented as 
\[
\tilde t=\delta_1^{-1}\check\lambda_1^{\,2}\check\beta_1^{\,2}t+\alpha_0,\quad
\tilde x=\check\lambda_1\check\beta_1x+\beta_0,\quad
\tilde u=\gamma_1u+\gamma_0,\quad
\tilde f=\delta_1f,
\]
the parameter tuple~$(\gamma_0,\gamma_1,\delta_1)$ freely runs through~$\mathbb R^3$, 
the parameter~$\check\lambda_1$ runs through a subgroup~$\check P$ 
of the multiplicative group~$\mathbb R_*$ of real numbers that is a complement to~$\hat P$, 
$\mathbb R_*=\hat P\times\check P$,
$\beta_0$, $\check\beta_1$ and, if $\lambda_0$ is free in~$K_{f^0}$, $\alpha_0$ 
are functions of $(\gamma_0,\gamma_1,\delta_1,\check\lambda_1)$, which are respectively equal to 0, 0, 1 at $(0,1,1,1)$, 
and $\alpha_0$ is a free real parameter if $\lambda_0$ takes only the zero values for elements of~\smash{$K_{f^0}$}.

Thus, the ratio $\hat\lambda_1/\check\lambda_1$ runs through the entire set~$\mathbb R$ of real numbers 
if $\hat\lambda_1$ and~$\check\lambda_1$ run through~$\hat P$ and~$\check P$, respectively.  
Setting $\gamma_0=0$, $\gamma_1=\hat\alpha$, $\delta_1>0$, 
$\hat\lambda_1=\check\lambda_1\delta^{-1/2}\check\beta_1$,
$\beta=\beta_0(0,\hat\alpha,\delta_1,\check\lambda_1)$,
$\lambda_0=\alpha_0(0,\hat\alpha,\delta_1,\check\lambda_1)$ if $\lambda_0$ is a free real parameter 
and $\alpha_0=\lambda_0=0$ otherwise, 
we obtain nonidentity transformations that are common for~$K_{f^0}$ and $\pi_*H$, 
i.e., $N_{f^0}\cap\pi_*H\ne\{{\rm id}\}$.
\end{remark}

Theorem~\ref{thm:ExdendingSemi-normalization} implies the following assertions. 

\begin{corollary}%\label{cor:ExtEquivSubgroup}
If the class~$\mathcal L|_{\mathcal S}$ is $(H,N_{\mathcal S})$-semi-normalized 
and does not possess inessential equivalence transformations~\cite[Appendix~A]{boyk2024a}
and $\bar H$ is a subgroup of~$G^\sim$ that properly contains~$H$, 
then this class is non-disjointedly $(\bar H,N_{\mathcal S})$-semi-normalized.
\end{corollary}

\begin{proof}
If the class~$\mathcal L|_{\mathcal S}$ is non-disjointedly $(H,N_{\mathcal S})$-semi-normalized,
then the corollary statement is trivial.  

Suppose that the class~$\mathcal L|_{\mathcal S}$ 
is disjointedly $(H,N_{\mathcal S})$- and $(\bar H,N_{\mathcal S})$-semi-normalized. 
Let $\bar{\mathscr T}$ be an element of~$\bar H\setminus H$. 
For any $\theta\in\mathcal S$, consider admissible transformation 
$\bar{\mathcal T}:=(\theta,\pi_*\bar{\mathscr T},\bar{\mathscr T}\theta)$. 
Since $\mathcal G^\sim=\mathcal N^{\rm f}\star\mathcal G^H$, we have 
$\bar{\mathcal T}=\mathcal T_0\star\mathcal T$ 
for some $\mathcal T_0\in\mathcal N_\theta$ and some $\mathcal T\in\mathcal G^H$. 
Thus, $\mathcal T_0=(\theta,\Phi,\theta)$ for some $\Phi\in N_\theta$, 
$\mathcal T=(\theta,\pi_*\mathscr T,\mathscr T\theta)$ for some $\mathscr T\in H$, 
$\mathscr T\theta=\bar{\mathscr T}\theta$ and $\pi_*\bar{\mathscr T}=\pi_*\mathscr T\circ\Phi$.
Then $\Phi=\pi_*(\mathscr T^{-1}\circ\bar{\mathscr T})\in{N_\theta\cap\pi_*\bar H=\{{\rm id}\}}$, 
i.e., $\Phi={\rm id}$ and $\pi_*\bar{\mathscr T}=\pi_*\mathscr T$. 
In other words, the equivalence transformations~$\mathscr T$ and~$\bar{\mathscr T}$ 
generate the same subset of~$\mathcal G^\sim$ and, therefore, should coincide 
in view of the absence of inessential equivalence transformation for the class~$\mathcal L|_{\mathcal S}$, 
which contradicts the fact that $\mathscr T\in H$ and~$\bar{\mathscr T}\in\bar H\setminus H$. 
\end{proof}

A typical property for the intersections of~$\pi_*H$ with elements~$N_\theta$ of the symmetry-subgroup family~$\mathcal N_{\mathcal S}$
is given by $N_\theta\cap\pi_*H\supseteq G^\cap$. 
Moreover, any semi-normalized class, even a disjointedly semi-normalized one, 
can be considered as semi-normalized with respect to 
a wider subgroup of the corresponding equivalence group 
and a family of wider uniform point symmetry subgroups 
that satisfy the above property.

\begin{corollary}\label{cor:IncludingKernelInUniformSymSubgroups}
If a class~$\mathcal L|_{\mathcal S}$
is semi-normalized with respect to a subgroup~$H$ of its equivalence group~$G^\sim$
and a symmetry-subgroup family $N_{\mathcal S}=\{N_\theta\mid\theta\in\mathcal S\}$,
then it is also semi-normalized with respect to the subgroup~$\bar H:=\hat G^\cap H$ of~$G^\sim$
and the symmetry-subgroup family
$\bar N_{\mathcal S}:=\{\bar N_\theta:=G^\cap N_\theta\mid\theta\in\mathcal S\}$,
where $\bar N_\theta\cap\pi_*\bar H\supseteq G^\cap$ for any $\theta\in\mathcal S$.
\end{corollary}

\begin{proof}
Denote by $\hat G^\cap$ the subgroup of~$G^\sim$ associated with~$G^\cap$, 
which is obtained from the group~$G^\cap$ 
by the trivial (identity) extension of its elements to the arbitrary elements of the class~$\mathcal L|_{\mathcal S}$.
Since $\hat G^\cap$ is the normal subgroup of~$G^\sim$ \cite[Proposition~3]{card2011a}, 
the Frobenius product $\bar H:=\hat G^\cap H$ is a subgroup of~$G^\sim$, 
and then Theorem~\ref{thm:ExdendingSemi-normalization} directly implies the required claim.  
\end{proof}

In the notation of Corollary~\ref{cor:IncludingKernelInUniformSymSubgroups},
If $G^\cap\ne\{{\rm id}\}$, then the class~$\mathcal L|_{\mathcal S}$ is 
non-disjointedly $(\bar H,\bar N_{\mathcal S})$-semi-normalized.

Recall that a normalized class~$\mathcal L|_{\mathcal S}$
is semi-normalized with respect to its entire equivalence group~$G^\sim$
and the trivial symmetry-subgroup family $\mathcal N_{\mathcal S}=\big\{N_\theta=\{{\rm id}\}\mid\theta\in\mathcal S\big\}$.
Hence Corollary~\ref{cor:IncludingKernelInUniformSymSubgroups} implies,
in view of Theorem~\ref{thm:EquivGroupoidOfSemi-normClasses}(ii),
the following assertion, whose infinitesimal counterpart is Corollary~2 in~\cite{card2011a}.

\begin{corollary}
For any system~$\mathcal L_\theta$ from a normalized class~$\mathcal L|_{\mathcal S}$,
the kernel symmetry group~$G^\cap$ of the systems in~$\mathcal L|_{\mathcal S}$
is a normal subgroup of the point symmetry group~$G_\theta$ of~$\mathcal L_\theta$.
\end{corollary}

\section{Uniform semi-normalization and\\ semi-normalization with respect to linear superposition}\label{sec:UniformSemi-norm}

In addition to the basic property 
$\mathcal N_{{\rm s}(\mathcal T)}\star\mathcal T=\mathcal T\star\mathcal N_{{\rm t}(\mathcal T)}$
for any $\mathcal T\in\mathcal G^H$ of Definition~\ref{def:Semi-normClass},
the attribute ``uniform'' for the subgroups~$N_\theta$ is also justified by the fact 
that in most of practically relevant examples of semi-normalized classes,
all the subgroups~$N_\theta$ are isomorphic or are at least of similar structure (in particular, of the same dimension).
To distinguish such examples, where in addition the subgroups~$N_\theta$ are known a priori,
we call the corresponding class~$\mathcal L|_{\mathcal S}$
\emph{uniformly $(H,N_{\mathcal S})$-semi-normalized}. 
If in addition the class~$\mathcal L|_{\mathcal S}$ is disjointedly $(H,N_{\mathcal S})$-semi-normalized, 
then we can call it \emph{uniformly disjointedly $(H,N_{\mathcal S})$-semi-normalized}.  

\begin{example}
The class~$\mathcal H$ defined in Remark~\ref{rem:OnNondejointedlySemiNormClasses} 
is semi-normalized but not uniformly semi-normalized.
Indeed, the point symmetry group~$G_{f^0}$ of the nonlinear diffusion equation with $f^0(u)=u^{-4/3}$
in particular contains the transformations of form~\eqref{eq:u-4/3DiffusionEqSymGroup}, 
which do not belong to~$\pi_*G^\sim_{\mathcal H}$, 
and hence for any choice of a family~$N_{\mathcal S}$ of uniform symmetry subgroups,
such a transformation belongs to~$N_{f^0}$. 
At the same time, the point symmetry group 
of any equation from the complement of the $G^\sim_{\mathcal H}$-orbit of~$\mathcal H_{f^0}$ in~$\mathcal H$
is contained in~$\pi_*G^\sim_{\mathcal H}$. 
Therefore, the structure of the corresponding element of~$N_{\mathcal S}$ definitely differs from that of~$N_{f^0}$. 
\end{example}

A particular but important case of uniform semi-normalized classes, 
which is also relevant to the linear Schr\"odinger equations studied in this paper,  
is given by classes of homogeneous linear systems of differential equations. 
Below, we follow the presentation in the second part of Section~3 in~\cite{kuru2018a} 
and modify it according to the framework developed 
in Sections~\ref{sec:Semi-normalizedClasses} and~\ref{sec:DisjointSemi-normalization}.

We start with a normalized superclass~$\mathcal L|_{\mathcal S}$ 
of (in general) inhomogeneous linear systems of differential equations of the form 
$\mathcal L_{\theta\zeta}$: $L(x,u_{(p)},\theta(x))=\zeta(x)$ with $(\theta,\zeta)\in\mathcal S$.%
\footnote{%
We avoid the parameterization of~$L$ by derivatives of arbitrary elements, 
considering such involved derivatives as additional arbitrary elements 
that are related to the original arbitrary elements~$\theta$ 
via additional auxiliary equations as their derivatives. 
} 
Here we consider the arbitrary-element tuple as consisting of two parts and 
change its notation from $\theta$ to $(\theta,\zeta)$. 
The subtuple $\theta$ parameterizes coefficients 
of the homogeneous linear left-hand side~$L$ and depends at most on~$x$. 
Each component of the right-hand side~$\zeta$ runs through the set of sufficiently smooth functions of~$x$.
Suppose that the class $\mathcal L|_{\mathcal S}$ also satisfies the following conditions: 
\begin{enumerate}\itemsep=0ex
\item 
Each system from this class is locally solvable. 
\item 
Elements of the pushforward~$\pi_*G^\sim$ of the equivalence group~$G^\sim$ 
of the class $\mathcal L|_{\mathcal S}$ by~$\pi$ 
are fibre-preserving point transformations of $(x,u)$
whose components for~$u$ are affine in~$u$. 
\end{enumerate}
The second condition means that elements of~$\pi_*G^\sim$ are of the form 
\begin{gather}\label{eq:EquivTransOfLinSystems}
\tilde x_j=X^j(x),\quad 
\tilde u^a=M^{ab}(x)(u^b+h^b(x)) 
\quad\mbox{with}\quad\det(X^j_{x_{j'}})\ne0,\quad \det(M^{ab})\ne0.
\end{gather}
The functions~$X^j$ and~$M^{ab}$ may satisfy additional constraints or even be quite specific
but the components of the tuple $h=(h^1,\dots,h^m)$ are arbitrary smooth functions of~$x$.
The transformations 
\[
\mathscr T_h\colon\quad
\tilde x_j=x_j,\quad 
\tilde u^a=u^a+h^a(x),\quad 
\tilde\theta=\theta,\quad 
\tilde\zeta=\zeta-L(x,h_{(p)}(x),\theta(x)),
\] 
constitute a normal subgroup~$N^\sim$ of~$G^\sim$. 
Furthermore, $G^\sim$ splits over~$N^\sim$ 
since $G^\sim=H\ltimes N^\sim$, 
where $H$ is the subgroup of the group~$G^\sim$ defined by the constraint~$h=0$. 
The normalization of the class~$\mathcal L|_{\mathcal S}$ means that 
its equivalence groupoid~$\mathcal G^\sim$ coincides with the action groupoid of~$G^\sim$. 

On the infinitesimal level, the counterpart of the second condition is that 
the pushforward~$\pi_*\mathfrak g^\sim$ of the equivalence algebra~$\mathfrak g^\sim$ 
of the class $\mathcal L|_{\mathcal S}$ by~$\pi$ 
consists of the vector fields of the form 
\begin{gather*}%\label{eq:VFsFromGsimLinCase}
Q=\xi^j(x)\p_{x_j}+\big(\eta^{ab}(x)u^b+\eta^{a0}(x)\big)\p_{u^a},
\end{gather*}
where coefficients~$\xi^j$ and~$\eta^{ab}$ may satisfy additional constraints 
but the $\eta^{a0}$ are arbitrary smooth functions of~$x$. 
The equivalence algebra~$\mathfrak g^\sim$ is the semidirect sum of the subalgebra~$\mathfrak h$
and the ideal~$\mathfrak n$, 
which are singled out from~$\mathfrak g^\sim$ by the constraints $\eta^{a0}=0$ and $\xi^j=\eta^{ab}=0$, respectively.

Any system~$\mathcal L_{\theta\zeta}$ from the class~$\mathcal L|_{\mathcal S}$ is mapped 
to the associated homogeneous system $\mathcal L_{\theta0}$ 
by the transformation~$\mathscr T_h$ from~$N^\sim$, where $h$ is a solution of $\mathcal L_{\theta\zeta}$. 
In other words, the class $\mathcal L|_{\mathcal S}$ is mapped 
onto the corresponding class~$\mathcal L^0|_{\mathcal S^0}$ of homogeneous systems
by a wide family of admissible transformations from the action groupoid of~$N^\sim$. 

We will use a double interpretation of the class~$\mathcal L^0|_{\mathcal S^0}$. 
On the one hand, 
it can be embedded in the class~$\mathcal L|_{\mathcal S}$ as the subclass 
singled out by the constraint $\zeta=0$, 
and then all the transformational structures related to~$\mathcal L^0|_{\mathcal S^0}$ 
are substructures of their counterparts for~$\mathcal L|_{\mathcal S}$. 
In particular, under this interpretation, 
the equivalence groupoid~$\mathcal G^\sim_0$ of~$\mathcal L^0|_{\mathcal S^0}$ 
is the subgroupoid of~$\mathcal G^\sim$ 
that is singled out by the constraints $\zeta=0$, $\tilde\zeta=0$ and~$L(x,h_{(p)}(x),\theta(x))=0$ 
for the source and target arbitrary elements and the parameter function tuple~$h$.
Up to factoring out the insignificant equivalence transformations related to the constraint $\zeta=0$
\cite[Appendix~A]{boyk2024a},
we identify the equivalence group~$G^\sim_0$ of the class~$\mathcal L^0|_{\mathcal S^0}$ 
with its subgroup consisting of the elements of~$G^\sim$ 
for which the corresponding values of the parameter function tuple~$h$ 
belong to the set~$\mathcal M^0$ of common solutions of the systems from~$\mathcal L^0|_{\mathcal S^0}$. 
In other words, $G^\sim_0=H_0\ltimes N^\sim_0$, where $H_0=H$ 
and the normal subgroup~$N^\sim_0$ of~$G^\sim_0$ is constituted by the transformations~$\mathscr T_h$  
with $h\in\mathcal M^0$.
Analogously, $\mathfrak g^\sim_0=\mathfrak h_0\lsemioplus\mathfrak n_0$,
where 
$\mathfrak g^\sim_0$ is the equivalence algebra of~$\mathcal L^0|_{\mathcal S^0}$,
$\mathfrak h_0=\mathfrak h$ is a subalgebra of~$\mathfrak g^\sim_0$ and
$\mathfrak n_0$ is the ideal of~$\mathfrak g^\sim_0$ 
consisting of the vector fields $\eta^{a0}(x)\p_{u^a}$ with $(\eta^{10},\dots,\eta^{m0})\in\mathcal M^0$.
On the other hand, 
we can consider~$\mathcal L^0|_{\mathcal S^0}$ as the class of systems of the form $L(x,u_{(p)},\theta(x))=0$, 
where $\theta$ is the complete tuple of its arbitrary elements that runs through the set~$\mathcal S^0=\varpi_*\mathcal S$.
By $\varpi$ we denote the natural projection from the space with the coordinates $(x,u,\theta,\zeta)$
onto the space with the coordinates $(x,u,\theta)$.
Then the corresponding transformational structures are 
the pushforwards, by~$\varpi$, of their counterparts under the former interpretation.
An advantage of the latter interpretation is that it avoids 
the insignificant equivalence transformations in $G^\sim_0$ related to the constraint $\zeta=0$.

The systems  $\mathcal L_{\theta\zeta}$ and $\mathcal L_{\tilde\theta\tilde\zeta}$ 
are $G^\sim$-equivalent if and only if 
their homogeneous counterparts $\mathcal L_{\theta0}$ and $\mathcal L_{\tilde\theta0}$
are $G^\sim_0$-equivalent.
Here the $G^\sim$- and $G^\sim_0$-equivalences 
can be replaced by the $\mathcal G^\sim$- and $\mathcal G^\sim_0$-equivalences, respectively.
Thus, the group classification of systems from the class $\mathcal L|_{\mathcal S}$ 
reduces to the group classification of systems from the class $\mathcal L^0|_{\mathcal S^0}$. 

For each~$\theta\in\mathcal S^0$ 
we denote by~$G^{\rm lin}_{\theta0}$ the subgroup 
of the point symmetry group~$G_{\theta0}$ of~$\mathcal L_{\theta0}$ 
that consists of the linear superposition transformations $\pi_*\mathscr T_h$:
$\tilde x_j=x_j$, $\tilde u^a=u^a+h^a(x)$,
where the tuple~$h$ runs through the solution set of~$\mathcal L_{\theta0}$. 
The subgroup~$H_0$ of~$G^\sim_0$ 
and the family $N^{\rm lin}=\{G^{\rm lin}_{\theta0}\mid\theta\in\mathcal S^0\}$
satisfy the conditions in Definition~\ref{def:DisjointedlySemi-normClass}. 
Therefore, the class $\mathcal L^0|_{\mathcal S^0}$ is disjointedly uniformly $(H_0,N^{\rm lin})$-semi-normalized. 
We call this kind of semi-normalization, 
which is characteristic for classes of homogeneous linear systems of differential equations, 
\emph{uniform semi-normalization with respect to the linear superposition of solutions}. 
If the systems from the class $\mathcal L^0|_{\mathcal S^0}$ have no common solutions, 
which is a regular situation, then $H_0=G^\sim_0$, and we will not indicate~$H_0$ in this case. 
By Theorem~\ref{thm:DisjontedlylySemi-normClasses}(ii),
for each~$\theta\in\mathcal S^0$ the group~$G_{\theta0}$ splits over~$G^{\rm lin}_{\theta0}$, 
and $G_{\theta0}=G^{\rm ess}_{\theta0}\ltimes G^{\rm lin}_{\theta0}$, 
where $G^{\rm ess}_{\theta0}=G_{\theta0}\cap\pi_*H_0$.
By Theorem~\ref{thm:DisjontedlylySemi-normClasses}(iii), 
the splitting of the point symmetry group induces a splitting of the corresponding maximal Lie invariance algebra, 
$\mathfrak g_{\theta0}=\mathfrak g^{\rm ess}_{\theta0}\lsemioplus\mathfrak g^{\rm lin}_{\theta0}$. 
Here $\mathfrak g^{\rm ess}_{\theta0}$ is the essential Lie invariance algebra of~$\mathcal L_{\theta0}$, 
$\mathfrak g^{\rm ess}_{\theta0}=\mathfrak g_{\theta0}\cap\pi_*\mathfrak h_0$, 
and the ideal $\mathfrak g^{\rm lin}_{\theta0}$, being the trivial part of~$\mathfrak g_{\theta0}$, 
consists of vector fields generating one-parameter symmetry groups of linear superposition of solutions. 
Thus, the group classification problem for the class $\mathcal L^0|_{\mathcal S^0}$ 
reduces to the classification of \emph{appropriate subalgebras} 
of the equivalence algebra~$\mathfrak g^\sim_0$ of this class. 
The qualification ``appropriate'' means that the pushforward of such a subalgebra by~$\pi$
is the essential Lie invariance algebra of a system from $\mathcal L^0|_{\mathcal S^0}$. 
We can classify appropriate subalgebras of~$\pi_*\mathfrak g^\sim_0$ instead of~$\mathfrak g^\sim_0$, 
see Definition~\ref{def:AppropriateSubalg} below. 

Given a class $\mathcal L^0|_{\mathcal S^0}$ 
of linear homogeneous systems of differential equations
that is uniformly semi-normalized with respect to the linear superposition of solutions, 
it is not necessary to start by considering the associated normalized superclass~$\mathcal L|_{\mathcal S}$ 
of generally inhomogeneous linear systems.
The class~$\mathcal L^0|_{\mathcal S^0}$ itself can be the starting point of the analysis. 
In order to get directly its specific semi-normalization, 
we need to suppose the following properties of~$\mathcal L^0|_{\mathcal S^0}$:
\begin{enumerate}\itemsep=0ex
\item 
The transformational part of any admissible transformation in the class~$\mathcal L^0|_{\mathcal S^0}$ 
is of the form~\eqref{eq:EquivTransOfLinSystems}. 
(Then the parameter function tuple~$h$ necessarily runs through the solution set of the source system.)
\item 
Any admissible transformation in the class~$\mathcal L^0|_{\mathcal S^0}$ 
with zero value of~$h$ belongs to the action groupoid of the equivalence group~$G^\sim_0$ of this class. 
\end{enumerate}

The normal subgroup~$N^\sim_0$ and the subgroup~$H_0$ of~$G^\sim_0$ and the family~$N^{\rm lin}$ 
are defined as above, $G^\sim_0=H_0\ltimes N^\sim_0$ and 
the class $\mathcal L^0|_{\mathcal S^0}$ is disjointedly uniformly $(H_0,N^{\rm lin})$-semi-normalized, 
i.e., it is uniformly semi-normalized with respect to the linear superposition of solutions.
The action groupoid of~$H_0$ can be called the \emph{essential equivalence groupoid} of the class $\mathcal L^0|_{\mathcal S^0}$
\cite[Remark 16]{vane2020b}.

The class~$\mathcal F$ of linear Schr\"odinger equations fits well into the developed framework,  
which we use in the present paper to study the group classification problem for this class in general 
and completely solve it in dimension 1+2 
in Sections~\ref{sec:PreliminarySymAnalisysOfMultiDSchEqs} and~\ref{sec:GroupClassificationOf(1+2)DLinSchEqs}, 
respectively.   

\begin{remark}
There exist classes of homogeneous linear systems of differential equations 
that are not uniformly semi-normalized with respect to the linear superposition of solutions
but disjointedly uniformly semi-normalized with respect to families of point-symmetry subgroups 
that are wider than the corresponding point-symmetry subgroups associated with the linear superposition of solutions, 
see Remark~\ref{rem:igammax subclass}.
\end{remark}

\section{Algebraic method of group classification\\ for semi-normalized classes}\label{sec:AlgMethodForSemi-normalizedClasses}

The property of $(H,N_{\mathcal S})$-semi-normalization of the class~$\mathcal L|_{\mathcal S}$ 
with known~$H_0$ and~$N_{\mathcal S}$
allows one to solve the group classification problem for this class using an algebraic approach.%
\footnote{%
See \cite[Section~2.2]{opan2022a} for a modern revisited statement of group classification problems
for classes of differential equations.
} 
 
The subgroups $G^{\rm ess}_\theta=G_\theta\cap\pi_*H$ of the groups~$G_\theta$
and the subalgebras $\mathfrak g^{\rm ess}_\theta=\mathfrak g_\theta\cap\pi_*\mathfrak h$ of the algebras~$\mathfrak g_\theta$, 
where $\theta$ runs through~$\mathcal S$,
are in general not known after describing the ingredients of  the semi-normalization structure of~$\mathcal L|_{\mathcal S}$.
According to Theorem~\ref{thm:OnInvAlgsInSemi-normClasses},
we have that for each $\theta\in\mathcal S$, the maximal Lie invariance algebra~$\mathfrak g_\theta$
of the system~$\mathcal L_\theta$ is the sum of~$\mathfrak g^{\rm ess}_\theta$ and~$\mathfrak n_\theta$,
$\mathfrak g_\theta=\mathfrak g^{\rm ess}_\theta+\mathfrak n_\theta$.
Note that in general, the essential Lie invariance algebra~$\mathfrak g^{\rm ess}_\theta$ nontrivially intersects
the uniform Lie invariance algebra~$\mathfrak n_\theta$,
and hence a natural basis of the maximal Lie invariance algebra~$\mathfrak g_\theta$
is constituted by a basis of the intersection $\mathfrak g^{\rm ess}_\theta\cap\mathfrak n_\theta$
and its complements to bases of~$\mathfrak g^{\rm ess}_\theta$ and of~$\mathfrak n_\theta$.
(In the infinite-dimensional case, bases are replaced by spanning sets of vector fields.)
The essential Lie invariance algebras are subalgebras of $\pi_*\mathfrak h$ 
but usually, only some of the subalgebras of $\pi_*\mathfrak h$ can serve as such algebras.  

\begin{definition}\label{def:AppropriateSubalg}
We call a subalgebra of $\pi_*\mathfrak h$ (resp.\ of~$\mathfrak h$) \emph{appropriate}
if it (resp.\ its pushforward by~$\pi_*$) coincides with the essential Lie invariance algebra 
of a system from the class~$\mathcal L|_{\mathcal S}$. 
\end{definition}

The pushforwards by transformations from~$\pi_*H$ (resp.\ from~$H$) preserve the set of appropriate subalgebras 
of $\pi_*\mathfrak h$ (resp.\ of~$\mathfrak h$)
since $\mathfrak g^{\rm ess}_{\mathscr T_*\theta}=(\pi_*\mathscr T)_*\mathfrak g^{\rm ess}_\theta$ 
for any $\mathscr T\in H$ and any $\theta\in\mathcal S$,%
\footnote{% 
$(\pi_*\mathscr T)_*\mathfrak g^{\rm ess}_\theta
=(\pi_*\mathscr T)_*(\mathfrak g_\theta^{}\cap\pi_*\mathfrak h)
=\big((\pi_*\mathscr T)_*\mathfrak g_\theta^{}\big)\cap\big((\pi_*\mathscr T)_*\pi_*\mathfrak h\big)
=\mathfrak g_{\mathscr T_*\theta}^{}\cap\pi_*(\mathscr T_*\mathfrak h)
=\mathfrak g_{\mathscr T_*\theta}^{}\cap\pi_*\mathfrak h
=\mathfrak g^{\rm ess}_{\mathscr T_*\theta}$.
The second, the third and the fourth equalities hold since 
the mapping $(\pi_*\mathscr T)_*$ is a bijection, 
all elements of~$G^\sim$ and of~$\mathfrak g^\sim$ are projectable to the space with the coordinates $(x,u)$ 
and $\mathscr T_*\mathfrak h=\mathfrak h$, respectively. 
We also have that $(\pi_*\mathscr T)_*\mathfrak g_\theta^{}=\mathfrak g_{\mathscr T_*\theta}^{}$.
} 
which generates an equivalence relation within this set. 
It is natural to call it the $H$-equivalence of appropriate subalgebras.
In view of Theorem~\ref{thm:EquivGroupoidOfSemi-normClasses}(iii), 
the $H$-equivalence can be replaced by the $G^\sim$- and even the $\mathcal G^\sim$-equivalence.

Making a preliminary analysis of the determining equations for Lie-symmetry vector fields of systems from~$\mathcal L|_{\mathcal S}$, 
one can find constraints that are satisfied by appropriate subalgebras of~$\pi_*\mathfrak h$ but not by general ones. 
These constraints often include upper bounds for the dimensions of the appropriate subalgebras 
or of their specific parts. 
This crucially reduces the amount of subalgebras to be classified 
and makes the classification problem tractable even in the case of infinite-dimensional~$\pi_*\mathfrak h$. 
Instead of classifying the appropriate subalgebras of~$\pi_*\mathfrak h$, 
one can classify the appropriate subalgebras of~$\mathfrak h$ and then push them forward by~$\pi$. 

Consequently, we derive the following assertion.

\begin{proposition}\label{pro:GroupClassificationOfSemi-normClasses}
Let the class~$\mathcal L|_{\mathcal S}$ be semi-normalized
with respect to a subgroup~$H$ of~$G^\sim$ and a symmetry-subgroup family that are known,
and let $\mathfrak h$ be a subalgebra of~$\mathfrak g^\sim$ that corresponds to~$H$.
Then the solution of the group classification problem for the class~$\mathcal L|_{\mathcal S}$
reduces to the classification, up to the $H$-equivalence,
of appropriate subalgebras of~$\pi_*\mathfrak h$
or, equivalently, of the algebra~$\mathfrak h$ itself.
\end{proposition}

In view of Theorem~\ref{thm:EquivGroupoidOfSemi-normClasses}(iii), Theorem~\ref{thm:OnInvAlgsInSemi-normClasses} 
and the above discussion, we suggest the following procedure of solving the group classification problem 
for a semi-normalized class~$\mathcal L|_{\mathcal S}$ of systems of differential equations. 
We formulate this procedure in terms of the algebra~$\pi_*\mathfrak h$.

\begin{itemize}\itemsep=0ex

\item
When computing the equivalence groupoid~$\mathcal G^\sim$ and analyzing its structure,
construct a subgroup~$H$ of~$G^\sim$ and
a family of uniform point symmetry groups $N_{\mathcal S}=\{N_\theta\mid\theta\in\mathcal S\}$
such that the class~$\mathcal L|_{\mathcal S}$ is $(H,N_{\mathcal S})$-semi-normalized.

\item
Find the subalgebra~$\mathfrak h$ of~$\mathfrak g^\sim$ associated with~$H$ 
and the uniform Lie invariance algebras~$\mathfrak n_\theta$ associated with~$N_\theta$.

\item
Make a preliminary analysis of the system ${\rm DE}$ of determining equations 
for Lie-symmetry vector fields of systems from the class~$\mathcal L|_{\mathcal S}$. 
The system ${\rm DE}$ usually splits into two subsystems, 
the ``essential'' and the ``uniform'' ones, 
which are satisfied by the components of vector fields from the respective parts
of the maximal Lie invariance algebras of systems from the class~$\mathcal L|_{\mathcal S}$.   

\item
Derive constraints satisfied by the appropriate subalgebras of~$\pi_*\mathfrak h$. 

\item
Classify, modulo the $H$-equivalence, the subalgebras of~$\pi_*\mathfrak h$ 
that satisfy the derived constraints. 

\item
For each subalgebra~$\mathfrak s$ in the constructed set~$\mathcal A$ 
of $H$-inequivalent families of subalgebras of~$\pi_*\mathfrak h$, 
\begin{itemize}\itemsep=0ex
\item
successively substitute the components of its basis/spanning elements into~${\rm DE}$, 
\item
merge all the obtained systems into a single system~${\rm DE}_{\mathfrak s}$, 
\item
solve the system~${\rm DE}_{\mathfrak s}$ with respect to the arbitrary elements of the class~$\mathcal L|_{\mathcal S}$,
\item
for each solution~$\theta$ of~${\rm DE}_{\mathfrak s}$, 
check whether $\mathfrak g^{\rm ess}_\theta=\mathfrak s$ 
(it may happen that $\mathfrak g^{\rm ess}_\theta\varsupsetneq\mathfrak s$).
\end{itemize}

\item
The subalgebra~$\mathfrak s$ is appropriate if and only if
the system~${\rm DE}_{\mathfrak s}$ is compatible and 
$\mathfrak g^{\rm ess}_\theta=\mathfrak s$ for a solution~$\theta$ of~${\rm DE}_{\mathfrak s}$.
This allows one to finally select those subalgebras in~$\mathcal A$ 
that are indeed appropriate subalgebras of~$\pi_*\mathfrak h$. 

\item
For each subalgebra~$\mathfrak s$ in the constructed set~$\mathcal B$ 
of $H$-inequivalent families of appropriate subalgebras of~$\pi_*\mathfrak h$, 
analyze the equivalence within the set 
$\mathcal S_{\mathfrak s}:=\{\theta\in\mathcal S\mid\mathfrak g^{\rm ess}_\theta=\mathfrak s\}$ 
that is induced by the stabilizer of~$\mathfrak s$ in~$\pi_*H$.
If there are $\pi_*H$-equivalent subalgebras~$\mathfrak s$ in~$\mathcal B$, 
one can also consider the equivalence between the corresponding sets~$\mathcal S_{\mathfrak s}$. 
When possible, use these equivalences to reduce the sets~$\mathcal B$ and~$\mathcal S_{\mathfrak s}$.

\item
Assemble the set 
$\{(\theta,\mathfrak s)\mid\mathfrak s\in\mathcal B,\,\theta\in\mathcal S_{\mathfrak s}\}$. 

\end{itemize}

The last set constitutes 
the solution of the group classification problem for the class~$\mathcal L|_{\mathcal S}$ 
in terms of essential Lie invariance algebras modulo both the $G^\sim$- and the $\mathcal G^\sim$-equivalences 
\cite[Section~2.2]{opan2022a}.
We call the above procedure \emph{the algebraic method of group classification for semi-normalized classes 
of differential equations}.
To present the solution of the group classification problem for the class~$\mathcal L|_{\mathcal S}$ 
in terms of maximal Lie invariance algebras, 
one needs to indicate the algebra $\mathfrak n_\theta$ for each $\theta\in\mathcal S$ 
and modify the set of classification cases to 
$\{(\theta,\mathfrak g_\theta)\mid
\mathfrak s\in\mathcal A,\,\theta\in\mathcal S,\,\mathfrak g^{\rm ess}_\theta=\mathfrak s,\,
\mathfrak g_\theta=\mathfrak s+\mathfrak n_\theta\}$. 

In fact, the above method is effective for the case of uniform semi-normalization, 
especially disjointed uniform semi-normalization, 
including uniform semi-normalized with respect to the linear superposition of solutions.
Then all the listed steps of the classification procedure can be properly implemented, 
and the mere presentation of inequivalent essential Lie-symmetry extensions is definitely sufficient. 

It may happen that 
the same class of systems of differential equations is semi-normalized under different choices 
of subgroups of its equivalence group and of families of uniform symmetry-subgroups of its systems. 
The proper selection of these objects helps to solve the group classification problem for such a class
in a more efficient way. 

\begin{remark}
The class~$\mathcal E$ from Remark~\ref{rem:igammax subclass}
is non-disjointedly semi-normalized with respect to the entire group~\smash{$G^\sim_{\mathcal E}$}
and the family~$\mathcal N=\{G^{\rm unf}_\gamma\}$.
At the same time, the interpretation of the class~$\mathcal E$
as disjointedly semi-normalized with respect to the family $\mathcal N=\{G^{\rm unf}_\gamma\}$
and the proper subgroup~$H$ of~$G^\sim_{\mathcal E}$ 
singled out from~$G^\sim_{\mathcal E}$ by the constraint $c=1$, see Remark~\ref{rem:igammax subclass},
is more convenient and effective for group classification of this class.
Indeed, then one can apply Theorem~\ref{thm:DisjontedlylySemi-normClasses},
which is slightly stronger than Theorems~\ref{thm:EquivGroupoidOfSemi-normClasses}
and~\ref{thm:OnInvAlgsInSemi-normClasses},
and, following Proposition~\ref{pro:GroupClassificationOfSemi-normClasses},
handle the smaller group~$H$.
One can also construct the generalized extended equivalence group of the class~$\mathcal E$ 
or its efficient part \cite{opan2017a,opan2020b},  
considering the reparameterized class~$\bar{\mathcal E}$ for~$\mathcal E$, 
where additional arbitrary elements~$\hat\gamma$ and~$\check\gamma$ defined by the constraints
$\hat\gamma_t=\gamma$ and $\check\gamma_t=t\gamma$,
and prove that the reparameterized class~$\bar{\mathcal E}$ is 
uniformly semi-normalized with respect to the linear superposition of solutions in the generalized sense, 
i.e., the original class~$\mathcal E$ is 
uniformly semi-normalized with respect to the linear superposition of solutions 
in the generalized extended sense.  
Nevertheless, the group classification of~$\mathcal E$ based on the last interpretation 
is essentially less efficient computationally than those based on 
the above semi-normalizations of~$\mathcal E$ in the usual sense, 
not to mention the absence of the corresponding theoretical background and related techniques for  
the case of semi-normalization in the generalized extended sense.
\end{remark}

\begin{remark}
There are examples of the opposite pattern.
\emph{Given a class of systems of differential equations that is disjointedly uniformly semi-normalized 
with respect to a subgroup of its equivalence group~$G^\sim$ and a family of uniform symmetry-subgroups of its systems, 
nevertheless, the group classification problem for this class can be more efficiently solved
using its non-disjointed semi-normalization with respect to a narrower subgroup of~$G^\sim$
and a family of wider uniform symmetry-subgroups rather than involving the former (disjointed) semi-normalization.}
In particular, this is the case for the class~$\mathcal F_{{\rm l}x_1}$
of (1+2)-dimensional linear Schr\"odinger equations of the form~\eqref{MLinSchEqs}
with potentials $V=U(t,x_2)+i\gamma(t)x_1$,
where $U$ is a complex-valued smooth function of~$(t,x_2)$
and $\gamma$ is a real-valued smooth function of~$t$, 
see the consideration related to the equation~\eqref{EqPotentialsWith2ProportionalChis} below. 
The class~$\mathcal F_{{\rm l}x_1}$ is (disjointedly) uniformly semi-normalized with respect to
its entire equivalence group~$G^\sim_{{\rm l}x_1}$ and linear superposition of solutions.
At the same time, it is more convenient to solve the group classification problem for this class
using its non-disjointed semi-normalization with respect to a proper subgroup of~$G^\sim_{{\rm l}x_1}$
and a family of wider uniform symmetry groups than the symmetry groups of linear superposition of solutions.
This is the way of treating the case $(k_2,r_0)=(0,1)$
in the proof of Theorem~\ref{thm:GroupClassifictionOf(1+2)DLinSchEqs} below.
\end{remark}

\section{Preliminary symmetry analysis of multidimensional\\ linear Schr\"odinger equations with complex-valued potentials}
\label{sec:PreliminarySymAnalisysOfMultiDSchEqs}

To demonstrate the efficiency of the developed theory
and the proposed algebraic method of group classification for semi-normalized classes, 
we study the class~$\mathcal F$ of (1+$n$)-dimensional ($n\geqslant1$) linear Schr\"odinger equations
with complex-valued potentials, which are of the form~\eqref{MLinSchEqs}.
First, in this section we carry out a preliminary analysis of transformational properties and Lie symmetries 
of equations from this class for an arbitrary number~$n$ of space variables~$x$. 
It turns out that it is possible, without fixing a value of~$n$, to find
the equivalence groupoid~$\mathcal G^\sim$, 
the equivalence group~$G^\sim$ and 
the equivalence algebra~$\mathfrak g^\sim$ of the class~$\mathcal F$ 
as well as to preliminarily analyze the determining equations for Lie symmetries of equations from this class
and to derive basic properties of such symmetries. 
In Section~\ref{sec:GroupClassificationOf(1+2)DLinSchEqs}, 
we specify these results to the case $n=2$ and use them 
to carry out the complete group classification of the class~$\mathcal F$ with $n=2$.
Recall that the case $n=1$ was exhaustively studied in~\cite{kuru2018a}. 

Similarly to other Schr\"odinger equations but in contrast to most systems of differential equations 
arising in applications, 
for equations from the class~$\mathcal F$,
the independent variables~$t$, $x_1$,~\dots,~$x_n$ take values in the set of real numbers, 
whereas values of the dependent variable $\psi$ and of the arbitrary element~$V$ are in general complex. 
To avoid considering complex values, one may replace Schr\"odinger equations by 
equivalent systems for two real-valued unknown functions, e.g., 
the real and the imaginary parts of~$\psi$ or 
the absolute value and the argument of~$\psi$. 
Nevertheless, such a replacement complicates the study more essentially 
than the involvement of complex numbers does, 
especially when using the absolute value and the argument of~$\psi$, 
which results in nonlinear systems instead of linear ones.
This is why we work with the form~\eqref{MLinSchEqs} of equations from the class~$\mathcal F$.
In view of the complex nature of~$\psi$ and~$V$,
it is necessary to formally extend the space of variables~$(t,x,\psi)$ with~$\psi^*$ and
the tuple of the single arbitrary element~$V$ with~$V^*$. 
Here and in what follows star denotes the complex conjugation.
In particular, we introduce~$\psi^*$ (resp.\ $V^*$) as an additional argument for all functions 
depending on~$\psi$ (resp.\ $V$), 
which includes, e.g., the components of point transformations and of vector fields 
in the spaces with coordinates $(t,x,\psi)$ or $(t,x,\psi,V)$. 
We thus follow Wirtinger's approach.%
\footnote{%
The variables $z=z_1+iz_2$ and $z^*=z_1-iz_2$ with $z_1,z_2\in\mathbb R$ 
are in fact not independent of each other in the canonical sense  
since they are uniquely related by the complex conjugation. 
However, the equalities $\p_{z^*}z=\p_zz^*=0$ and $\p_zz=\p_{z^*}z^*=1$
with the Wirtinger derivatives $\p_z=\tfrac12(\p_{z_1}-{\rm i}\p_{z_2})$ and $\p_{z^*}=\tfrac12(\p_{z_1}+{\rm i}\p_{z_2})$
allow one to formally treat $z$ and $z^*$ as independent variables 
and to simultaneously indicate them as arguments of relevant functions instead of $(z_1,z_2)$.
}
Confining a differential function of~$\psi$, i.e., of $(\psi,\psi^*)$ in the above interpretation,
to the solution set of an equation from the class~$\mathcal F$, 
we should also take into account its complex conjugate $-i\psi^*_t+\psi^*_{aa}+V^*(t,x)\psi^*=0$,
thus considering the system of two equations instead of the single one, 
cf.\ \cite[Section~1]{kuru2018a}.
At the same time, there two essential simplifications due to using Wirtinger's approach. 
In view of the fact that 
the invariance (resp.\ equivalence) conditions hold or do not hold simultaneously 
for equations from the class~$\mathcal F$ and their complex conjugate counterparts,  
it suffices to check such conditions only for the former equations. 
Moreover, presenting point transformations in the spaces with coordinates $(t,x,\psi)$ or $(t,x,\psi,V)$, 
we can omit the transformation components for~$\psi^*$ and~$V^*$ 
since they are just the complex conjugates of those for~$\psi$ and~$V$, respectively.

\begin{notation*}
In this and the next sections, 
$t$ and $x=(x_1,\dots,x_n)$ are the real independent variables,
$\psi$~is the unknown complex-valued function of $t$ and $x$.
For a complex value $\beta$, its conjugate is denoted by $\beta ^*$ , and we define
\[
\hat\beta=\beta \quad\mbox{if}\quad T_t>0 \quad\mbox{and}\quad
\hat\beta=\beta^* \quad\mbox{if}\quad T_t<0,
\]
where $T$ is the $t$-component of point transformations in the space with $t$ as a coordinate.
The indices $a$ and $b$ run from $1$ to $n$, 
$\nabla\psi=(\psi_1,\dots,\psi_n)$, $\nabla\psi^*=(\psi^*_1,\dots,\psi^*_n)$, $\p_a:=\p/\p x_a$ and $|x|^2:=x_ax_a$.
$E$ is the $n\times n$ identity matrix.
The total derivative operators~$\mathrm D_t$ and~$\mathrm D_a$ are defined as
\begin{gather*}
\mathrm D_t=\p_t+\psi_t\p_\psi+\psi^*_t\p_{\psi^*}+\psi_{tt}\p_{\psi_t}
+\psi_{tb}\p_{\psi_b}+\psi^*_{tt}\p_{\psi^*_t}+\psi^*_{tb}\p_{\psi^*_b}+\cdots,\\
\mathrm D_a=\p_a+\psi_a\p_\psi+\psi^*_a\p_{\psi^*}+\psi_{ta}\p_{\psi_t}
+\psi_{ab}\p_{\psi_b}+\psi^*_{ta}\p_{\psi^*_t}+\psi^*_{ab}\p_{\psi^*_b}+\cdots.
\end{gather*}
\end{notation*}

\subsection{Equivalence groupoid}%\label{sec:EquivGroupoidOfClassOfMultiDSchEqs}

We find the equivalence groupoid~$\mathcal G^\sim$ of the class~$\mathcal F$ using the direct method, 
see, e.g., \cite[Section~2]{opan2017a} and \cite[Section~2]{vane2020b} for a general description of this method and 
\cite{bihl2012b,king1991c,king1998a} for computational techniques involved. 
Let $\mathcal L_V$ denote the equation from~$\mathcal F$ with potential~$V$.
We seek for all point transformations of the form
\begin{equation}\label{multequivptrans}
\Phi\colon\
\tilde t=T(t,x,\psi,\psi^*),\
\tilde x_a =X^a(t,x,\psi,\psi^*),\
\tilde\psi=\Psi(t,x,\psi,\psi^*),\
\tilde\psi^*=\Psi^*(t,x,\psi,\psi^*),
\end{equation}
where $\det\p(T,X,\Psi,\Psi^*)/\p(t,x,\psi,\psi^*)\ne 0$ with $X=(X^1,\dots,X^n)$,
that map a fixed equation~$\mathcal L_V$ from the class~$\mathcal F$
to an equation~$\mathcal L_{\tilde V}$:
$i\tilde\psi_{\tilde t}+\tilde\psi_{\tilde x_a\tilde x_a}+\tilde V(\tilde t,\tilde x)\tilde\psi=0$
from the same class.

The following assertion is proved similarly to~\cite[Lemma~1]{popo2010a}.

\begin{lemma}\label{Lemmaequivalmultidim}
Any point transformation $\Phi$ connecting two equations from the class~$\mathcal F$
satisfies the conditions
\begin{gather*}%\label{MCondtransf}
\begin{split}
&T_a=T_\psi=T_{\psi^*}=0, \quad X^a_\psi=X^a_{\psi^*}=0,\\
&\Psi_\psi=0\quad\text{if}\quad T_t<0 \quad\text{and}\quad\Psi_{\psi^*}=0\quad\text{if}\quad T_t>0.
\end{split}
\end{gather*}
\end{lemma}

\begin{theorem}\label{thmeMquivptransf}
The equivalence groupoid~$\mathcal G^\sim$ of the class~$\mathcal F$
consists of the triples of the form $(V,\Phi,\tilde V)$,
where $\Phi$ is a point transformation in the space of variables,
whose components are
\begin{subequations}\label{Mequgroupoidsub}
\begin{gather}\label{Mequgroupoidsub_a}
\tilde t=T, \quad
\tilde x_a=|T_t|^{1/2}O^{ab}x_b+\mathcal X^a,
\\ \label{Mequgroupoidsub_b}
\tilde \psi=\exp\left(
\frac i8\frac{T_{tt}}{|T_t|}\,|x|^2+
\frac i2\frac{\varepsilon'\mathcal X^b_t}{|T_t|^{1/2}}\,O^{ba}x_a+i\Sigma+\Upsilon
\right)(\hat\psi+\hat\Lambda),
\end{gather}
and the target potential~$\tilde V=\tilde V(\tilde t,\tilde x)$ is expressed via the source potential~$V=V(t,x)$ as
\begin{equation}\label{Mequgroupoidsub_c}
\tilde V=\frac{\hat V}{|T_t|}
+\frac{2T_{ttt}T_t-3T_{tt}{}^2}{16\varepsilon'T_t{}^3}|x|^2
+\frac{\varepsilon'}{2|T_t|^{1/2}}\left(\frac{\mathcal X^b_t}{T_t}\right)_{\!t}O^{ba}x_a
+\frac{\Sigma_t-i\Upsilon_t}{T_t}-\frac{\mathcal X^a_t\mathcal X^a_t+inT_{tt}}{4T_t{}^2}.
\end{equation}
\end{subequations}
Here $T$, $\mathcal X^a$, $\Sigma$ and~$\Upsilon$ are arbitrary smooth real-valued functions of $t$ with
$T_t\ne 0$, $\varepsilon'=\sgn T_t$,
$O=(O^{ab})$~is an arbitrary constant $n\times n$ orthogonal matrix,
and $\Lambda=\Lambda(t,x)$ is an arbitrary solution of the initial equation.
\end{theorem}

\begin{proof}
Let a point transformation $\Phi$ connect equations~$\mathcal L_V$ and~$\mathcal L_{\tilde V}$ from the class~$\mathcal F$.
Lemma~\ref{Lemmaequivalmultidim} implies
that $T=T(t)$, $X^a=X^a(t,x)$ and $\Psi=\Psi(t,x,\hat\psi)$
with $T_t\det(\p X/\p x)\Psi_{\hat\psi}\ne 0$.
Acting the total derivative operators~$\mathrm D_t$ and $\mathrm D_a$ 
on the equality $\tilde \psi(\tilde t,\tilde x)=\Psi(t,x,\hat\psi)$, we derive
\begin{gather*}\label{LInSchEqsMultequiveqder}
\mathrm D_t\tilde \psi(\tilde t,\tilde x)=\tilde \psi_{\tilde t}T_t+\tilde \psi_{\tilde x_b}X^b_t= \mathrm D_t\Psi,\quad
\mathrm D_a\tilde \psi(\tilde t,\tilde x)=\tilde \psi_{\tilde x_c} X^c_a=\mathrm D_a\Psi,\\
\mathrm D_b\mathrm D_a\tilde \psi(\tilde t,\tilde x)=\tilde \psi_{\tilde x_c\tilde x_d}X^c_bX^d_a+\tilde \psi_{\tilde x_d}X^d_{ab}=\mathrm D_b\mathrm D_a\Psi.
\end{gather*}
The above equations are equivalent to
\begin{gather}\label{Multderivative}
\tilde\psi_{\tilde t}=\frac{\mathrm D_t\Psi-Y^a_bX^b_t\mathrm D_a\Psi}{T_t},\quad
\tilde\psi_{\tilde x_c}=Y^a_c\mathrm D_a\Psi,\quad
%\label{Multderivatives}
\tilde \psi_{\tilde x_c\tilde x_d}=Y^a_cY^b_d(\mathrm D_b\mathrm D_a\Psi- Y^d_cX^c_{ab}\mathrm D_d\Psi),
\end{gather}
where $X^a_cY^c_b=Y^a_cX^c_b=\delta_{ab}$, and $\delta_{ab}$ is the Kronecker delta.
In fact, the vector-function $Y=(Y^1,\dots,Y^n)$ is the inverse
of the vector-function $X=(X^1,\dots,X^n)$ with respect to $x$,
and $Y^a_c=\p Y^a/\p\tilde x_c$.
We substitute the expressions~\eqref{multequivptrans} and~\eqref{Multderivative}
for~$\tilde\psi$, $\tilde\psi_{\tilde t}$ and $\tilde \psi_{\tilde x_c\tilde x_d}$
and then the expression for $\hat\psi_t$, $\hat\psi_t=i\varepsilon'(\hat\psi_{aa}+\hat V\hat\psi)$,
into the equation~$\mathcal L_{\tilde V}$.
As a result, we derive the equation
\begin{gather*}
\frac i {T_t}\left(\Psi_t+\Psi_{\hat\psi}(i\varepsilon'\hat\psi_{aa}+i\varepsilon'\hat V\hat\psi)-Y^a_b
(\Psi_a+\Psi_{\hat\psi}\hat\psi_a) X^b_t\right)
+Y^a_cY^b_c\left(\Psi_{ab}+2\Psi_{a\hat\psi}\hat\psi_b\right)
\\
\qquad{}+Y^a_cY^b_c\left(\Psi_{\hat\psi \hat\psi}\hat\psi_b\hat\psi_a+\Psi_{\hat\psi}\hat\psi_{ab}-Y^d_c(\Psi_d+\Psi_{\hat\psi}\hat\psi_d)X^c_{ab}\right)+\tilde V\Psi=0.
\end{gather*}
Then splitting this equation with respect to various derivatives of
$\hat\psi$ and additionally arranging lead to the system
\begin{align}\label{Multderivativesubstitutionsplit}
&Y^a_cY^b_c=\frac{\delta_{ab}}{|T_t|},\quad
\Psi_{\hat\psi\hat\psi}=0,
\\[1ex]\label{Multderivativesubstitutionsplita}
&\frac 2{|T_t|}\Psi_{a\hat\psi}-\frac i{T_t}Y^a_b\Psi_{\hat\psi}X^b_t-\frac1
{|T_t|}Y^a_c\Psi_{\hat\psi}X^c_{aa}=0,
\\[1ex]\label{Multderivativesubstitutionsplitare}
&\frac i{T_t}\Psi_t-\frac 1{|T_t|}\hat V \Psi_{\hat\psi}\hat\psi-\frac i{T_t}Y^a_b\Psi_aX^b_t
+\frac1{|T_t|}\Psi_{aa}-\frac1{|T_t|}Y^d_c\Psi_d X^c_{aa}+\tilde V\Psi=0.
\end{align}

The first equation in~\eqref{Multderivativesubstitutionsplit} together with the condition $Y^a_cX^c_b=\delta_{ab}$
imply that $X^b_a= |T_t|Y^a_b$.
Therefore, $X^c_aX^c_b= |T_t|\delta_{ab}$, i.e., $X_a\cdot X_b=|T_t|\delta_{ab}$ in terms of the tuple~$X$.
We differentiate this equation for each $(a,b)$ with respect to~$x_c$ and permute the indices~$a$, $b$ and~$c$,
which gives the equations\looseness=-1
\[
X_{ac}\cdot X_b+X_a\cdot X_{bc}=0,\quad
X_{ab}\cdot X_c+X_b\cdot X_{ac}=0,\quad
X_{bc}\cdot X_a+X_c\cdot X_{ab}=0
\]
implying $X_a\cdot X_{bc}=0$ for all values of~$(a,b,c)$.
Since for each $(t,x)$ the tuples $X_1$, \dots, $X_n$ constitute an (orthogonal) basis of~$\mathbb R^n$,
this means that $X_{bc}=0$, i.e.,
the vector-function~$X$ is affine in~$x$ with coefficients depending on~$t$.
The equations $X_a\cdot X_b=|T_t|\delta_{ab}$ constrain the matrix~$(X^a_b)$,
and thus we have
$
X^a=|T_t|^{1/2}O^{ab}(t)x_b+\mathcal X^a(t),
$
where $O=(O^{ab})$ is a time-dependent orthogonal matrix and $\mathcal X^a$ is a time-dependent vector.

The general solution of the second equation in~\eqref{Multderivativesubstitutionsplit}
is $\Psi=\Psi^1(t,x)\hat\psi+\Psi^0(t,x)$,
where $\Psi^0$ and $\Psi^1$ are smooth complex-valued functions of $t$ and $x$ with $\Psi^1\ne 0$.

We substitute the above expression of $X^a$ and~$\Psi$ into~\eqref{Multderivativesubstitutionsplita} to get
\begin{gather}\label{LinSchEqsmultsimplifiednext}
\frac{\Psi^1_a}{\Psi^1}=\frac{i}{2T_t}X^b_aX^b_t.
\end{gather}
Cross-differentiating copies of the equation~\eqref{LinSchEqsmultsimplifiednext} with different values of~$a$
and subtracting the results of differentiation from each other
give $\p_c(X^b_aX^b_t)=\p_a(X^b_cX^b_t)$.
Hence $O^{ba}O^{bc}_t=O^{bc}O^{ba}_t$,
or $O^{\mathsf T}O_t-O^{\mathsf T}_tO=0$ in the matrix terms,
where $O^{\mathsf T}$ is the transpose of~$O$.
At the same time, differentiating the orthogonality condition $O^{\mathsf T}O=E$,
where $E$ is the $n\times n$ identity matrix, with respect to~$t$
results in $O^{\mathsf T}O_t+O^{\mathsf T}_tO=0$.
Therefore, $O^{\mathsf T}O_t=0$, from which we have $O_t=0$,
and thus $O$ is a constant orthogonal matrix.

The system consisting of the equations~\eqref{LinSchEqsmultsimplifiednext} when $a$ varying integrates to the following expression for $\Psi^1$:
\begin{gather}\label{LInSchEqsMultwavefunct}
\Psi^1= \exp\left(\frac i8\frac{T_{tt}}{|T_t|}\,|x|^2+
\frac i2\frac{\varepsilon' \mathcal X^b_t}{|T_t|^{1/2}}\,O^{ba}x_a+i\Sigma+\Upsilon\right),
\end{gather}
where~$\Sigma$ and~$\Upsilon$ are arbitrary smooth real-valued functions of $t$.
Finally, we consider the equation~\eqref{Multderivativesubstitutionsplitare},
which reduces, under the derived conditions, to the equation
\begin{gather*}
\frac i{T_t}\Psi_t-\frac 1{|T_t|}\hat V \Psi_{\hat\psi}\hat\psi-\frac i{T_t}Y^a_bX^b_t\Psi_a+\frac1
{|T_t|}\Psi_{aa}+\tilde V\Psi=0.
\end{gather*}
Splitting this equation with respect to $\hat\psi$ in view of the representation for $\Psi$ and rearranging,
we obtain
\begin{gather}\label{LInSchEqsMultarbitr}
\tilde V=\frac{\hat V}{|T_t|}-\frac i{T_t\Psi^1}
\left(\Psi^1_t-\frac{ X^b_aX^b_t}{|T_t|}\Psi^1_a\right)-\frac 1{|T_t|}\frac{\Psi_{aa}^1}{\Psi^1},
\\[1ex]\label{LInSchEqsMultarind}
i\varepsilon'\Psi^0_t-\frac i{T_t}X^b_aX^b_t\Psi^0_a+\Psi^0_{aa}+|T_t|\tilde V\Psi^0=0.
\end{gather}
We introduce the function $\Lambda=\hat\Psi^0/\hat\Psi^1$, i.e., $\Psi^0=\Psi^1\hat\Lambda$.
The equation~\eqref{LInSchEqsMultarind} is equivalent to
the initial linear Schr\"odinger equation in terms of $\Lambda$.
After substituting the expression~\eqref{LInSchEqsMultwavefunct} for $\Psi^1$ into~\eqref{LInSchEqsMultarbitr}
then additionally collecting coefficients of $x$, we derive the final expression for $\tilde V$.
\end{proof}

\begin{corollary}\label{free Schrodingerequation}
A (1+n)-dimensional linear Schr\"odinger equation of the form~\eqref{MLinSchEqs}
is reduced to the free linear Schr\"odinger equation by a point transformation if and only if
\[V=\varrho(t)|x|^2+\varrho^a(t)x_a+\varrho^0(t)+i\tilde \varrho^0(t)\]
for some real-valued smooth functions $\varrho$, $\varrho^a$, $\varrho^0$ and $\tilde \varrho^0$ of $t$.
\end{corollary}

\subsection{Equivalence group and equivalence algebra}%\label{sec:EquivGroupOfClassOfMultiDSchEqs}

The concept of equivalence group, 
which was suggested in~\cite{ovsy1982A} and properly formalized in~\cite[Section~2.3]{popo2010a}, 
see also \cite{vane2020b} and references therein,
plays a central role in the theory of group classification of differential equations.
The equivalence algebra and the equivalence group of a class of differential equations 
can be independently computed using the infinitesimal method and the direct or the algebraic method, respectively. 
Nevertheless, it is much easier to find them from the equivalence groupoid of the class once this groupoid is known.
The following assertions follow from Theorem~\ref{thmeMquivptransf}
in a similar way as that we used in~\cite{kuru2018a} to obtain consequences of~\cite[Theorem~6]{kuru2018a}.

\begin{corollary}%\label{Multieqgpe}
The (usual) equivalence group~$G^\sim$ of the class~$\mathcal F$ consists 
of the point transformations in the space with the coordinates~$(t,x,\psi,\psi^*,V,V^*)$
whose components are of the form~\eqref{Mequgroupoidsub} with $\Lambda=0$.%
\footnote{%
Strictly following the definition of (usual) equivalence groups, 
we should assume that elements of the group~$G^\sim$ are point transformations 
in the space with the coordinates $(t,x,\psi_{(2)},\psi^*_{(2)},V,V^*)$.
Here the subscript ``$(2)$'' of~$\psi$ denotes the collection of jet coordinates corresponding to
the derivatives of~$\psi$ up to order two, including~$\psi$ itself as its zeroth-order derivative, 
$\psi_{(2)}=(\psi,\psi_t,\psi_a,\psi_{tt},\psi_{ta},\psi_{ab})$, and similarly for $\psi^*_{(2)}$.
Nevertheless, the arbitrary element~$V$ depends only $(t,x)$ 
and the relation~\eqref{Mequgroupoidsub_c} between the source and target values of~$V$ 
for admissible transformations within the class~$\mathcal F$ 
does not involve nonzero-order derivatives of~$\psi$. 
Therefore, analogously to~\cite[footnote~1]{kuru2020a} and~\cite[Section~1]{kuru2018a},
we can assume that elements of the group~$G^\sim$ act in the space with the coordinates $(t,x,\psi,\psi^*,V,V^*)$, 
thus shrinking the space underlying this group.
It~suffices to present only the transformation components for $(t,x,\psi,V)$.
The transformation components for derivatives of~$\psi$ are constructed from the $(t,x,\psi)$-components
by the standard prolongation using the chain rule.
The transformation components for derivatives of~$\psi^*$ and for~$V^*$
are obtained by conjugating their counterparts for derivatives of~$\psi$ and for~$V$.
}
\end{corollary}

\begin{remark}
The identity component of~$G^\sim$ consists of the transformations
of the form~\eqref{Mequgroupoidsub},
where $\Lambda=0$, $\det O=1$ and $T_t>0$, i.e., $\varepsilon'=1$.
The whole equivalence group~$G^\sim$ is generated
by the transformations from its identity component
and two discrete transformations,
the space reflection
$\tilde t=t,$ $\tilde x_a=-x_a,$ $\tilde x_b=x_b,$ $b\ne a$, $\tilde \psi=\psi,$ $\tilde V=V$ for a fixed~$a$
and the Wigner time reflection
$\tilde t=-t,$ $\tilde x=x,$ $\tilde \psi=\psi^*$, $\tilde V=V^*$.
\end{remark}

\begin{theorem}\label{thm:EquivAlgOfF}
The equivalence algebra of the class~$\mathcal F$ is the algebra
\[
\mathfrak g^\sim=\big\langle\hat D(\tau),\,\hat J_{ab},\,a<b,\,\hat P(\chi),\,\hat M(\sigma),\,\hat I(\rho)\big\rangle,
\]
where $\tau$, $\chi=(\chi^1,\dots,\chi^n) $, $\sigma$ and~$\rho$ run through the set of smooth real-valued functions of~$t$. 
The vector fields spanning~$\mathfrak g^\sim$ are defined by
\begin{gather*}
\hat D(\tau)=\tau\p_t+\frac12\tau_tx_a\p_a+\frac1{8}\tau_{tt}|x|^2 \left(i\psi\p_\psi-i\psi^*\p_{\psi^*}\right)
\\ \phantom{\hat D(\tau)=}
{}-\left(\tau_tV-\frac18\tau_{ttt}|x|^2+i\frac{\tau_{tt}}{4}\right)\p_V-
\left(\tau_tV^{*}-\frac18\tau_{ttt}|x|^2-i\frac{\tau_{tt}}{4}\right)\p_{V^*},
\\
\hat J_{ab}=x_a\p_b-x_b\p_a,\quad a\ne b,
\\
\hat P(\chi)=\chi^a\p_a+\frac i2\chi^a_tx_a\left(\psi\p_\psi-\psi^*\p_{\psi^*}\right)+
\frac12\chi^a_{tt}x_a\left(\p_V+\p_{V^*}\right),
\\
\hat M(\sigma)=i\sigma(\psi\p_{\psi}-\psi^*\p_{\psi^*})+\sigma_t(\p_V+\p_V^*),
\\
\hat I(\rho)=\rho(\psi\p_{\psi}+\psi^*\p_{\psi^*})-i\rho_t(\p_V+\p_V^*).
\end{gather*}
\end{theorem}

\begin{proof}
The equivalence algebra~$\mathfrak g^\sim$ of the class~$\mathcal F$
is obtained using the knowledge of the identity component of the equivalence group~$G^\sim$.
Searching for infinitesimal generators of one-parameter subgroups of~$G^\sim$,
we represent the parameter function~$\Sigma$ in the form $\Sigma=\frac14\mathcal X^a\mathcal X^a_t+\bar\Sigma$,
where $\bar\Sigma$ is a function of~$t$,
for better consistency of the group parameterization with the one-parameter subgroup structure of~$G^\sim$.
Then we successively assume one of the transformation parameters~$T$, $O$, $\mathcal X^a$, $\bar\Sigma$ and~$\Upsilon$
to depend on a continuous parameter~$\delta$
and set the other transformation parameters to the trivial values
corresponding to the identity transformation,
which are $t$ for $T$, $E$ for $O$ and zeroes
for $\mathcal X^a$, $\bar\Sigma$ and~$\Upsilon$, where $E$ is the $n\times n$ identity matrix.
The components of the associated infinitesimal generator
$Q=\tau\p_t+\xi^a\p_a+\eta\p_\psi+\eta^*\p_{\psi^*}+\theta\p_V+\theta^*\p_{V^*}$
are given by
\[
\tau=\frac{\mathrm d\tilde t}{\mathrm d\delta}\Big|_{\delta=0},\quad
\xi^a=\frac{\mathrm d\tilde x}{\mathrm d\delta}\Big|_{\delta=0},\quad
\eta=\frac{\mathrm d\tilde\psi}{\mathrm d\delta}\Big|_{\delta=0},\quad
\theta=\frac{\mathrm d\tilde V}{\mathrm d\delta}\Big|_{\delta=0}.
\]
The above procedure results in the set of vector fields spanning the algebra~$\mathfrak g^\sim$.
\end{proof}

\begin{corollary}\label{col:UniformSemi-normalizationOfClassF}
The class~$\mathcal F$ is uniformly semi-normalized with respect to the linear superposition of solutions.
\end{corollary}

The uniform semi-normalization of the class~$\mathcal F$ guarantees a specific factorization
of point symmetry groups of all equations from this class.
For any potential~$V$, each element~$\Phi$ of the point symmetry group~$G_V$ of an equation~$\mathcal L_V$
generates an admissible point transformation $(V,\Phi,V)$ of~$\mathcal F$.
Therefore, the transformation~$\Phi$ necessarily has
the form~\eqref{Mequgroupoidsub_a}--\eqref{Mequgroupoidsub_b},
where the transformation parameters additionally satisfy the equation~\eqref{Mequgroupoidsub_c}
with $\tilde V(\tilde t,\tilde x)=V(\tilde t,\tilde x)$.
The symmetry transformations associated with the linear superposition of solutions to the equation~$\mathcal L_V$
which are of the above form with $T=t$, $O=E$ and $\mathcal X^a=\Sigma=\Upsilon=0$,
constitute a normal subgroup of the group~$G_V$,
which is denoted by~$G^{\rm lin}_V$ and is the trivial part of~$G_V$.
In view of the discussion on the uniform semi-normalization with respect to the linear superposition of solutions
in~\cite[Section~3]{kuru2018a},
Corollary~\ref{col:UniformSemi-normalizationOfClassF} implies that
the group~$G_V$ splits over~$G^{\rm lin}_V$, $G_V=G^{\rm ess}_V\ltimes G^{\rm lin}_V$,
the subgroup~$G^{\rm ess}_V$ of~$G_V$ is singled out from~$G_V$ by the constraint $\Lambda=0$
and will be considered as the only essential part of~$G_V$.

\subsection{Analysis of determining equations for Lie symmetries}\label{sec:AnalysOfDetEqsForLieSymsOfMultiDSchEqs}

Using the infinitesimal criterion for Lie symmetries \cite{blum1989A,olve1993A,ovsy1982A}, 
for each potential $V$ we can derive the determining equations 
which are satisfied by the components of vector fields 
from the maximal Lie invariance algebra $\mathfrak g_V$ of the equation $\mathcal L_V$ from the class~$\mathcal F$.
In this section, we partially integrate the determining equations 
and preliminary analyze properties of the collection of $\mathfrak g_V$ when $V$ varies for $\mathcal L_V\in\mathcal F$.

The infinitesimal criterion states that a vector field $Q=\tau\p_t+\xi^a\p_a+\eta\p_\psi+\eta^*\p_{\psi^*}$,
where the components~$\tau$, $\xi^a$ and $\eta$ are smooth functions of~$(t,x,\psi,\psi^*)$, and $\eta^*$ is the complex conjugate of~$\eta$,
belongs to the algebra $\mathfrak g_V$ if an only if
\[Q_{(2)}\big(i\psi_t+\psi_{aa}+V(t,x)\psi\big)\big|_{\mathcal L_V}=0\]
with $Q_{(2)}$ being the second prolongation of the vector field $Q$, see Section~\ref{sec:BasicsNotionsOfGroupClassification}.
Expanding this expression, we obtain
\begin{equation}\label{Deteqmult}
i\eta^t+\eta^{aa}+(\tau V_t+\xi^a V_a)\psi+V\eta=0,
\end{equation}
\looseness=-1
where $\eta^t=\mathrm D_t\left(\eta-\tau\psi_t-\xi^a\psi_a\right)+\tau\psi_{tt}+\xi^a\psi_{ta}$ and $\eta^{ab}=\mathrm D_a\mathrm D_b\left(\eta-\tau\psi_t-\xi^c\psi_c\right)+\tau\psi_{tab}+\xi^c\psi_{abc}$.
We recall that $\mathrm D_t$ and $\mathrm D_a$ are the operators of total derivatives with respect to $t$ and $x_a$, respectively.
The substitution of $\psi_t=i\psi_{aa}+iV\psi$ and $\psi_t^*=-i\psi_{aa}^*-iV^*\psi^*$
into~\eqref{Deteqmult} and splitting with respect to the various derivatives of~$\psi$
and~$\psi^*$ lead to a linear overdetermined system of determining equations for the coefficients of~$Q$,
\begin{subequations}%\label{Multoverdeterequation}
\begin{gather}\label{Multoverdeterequation_a}
\tau_\psi=\tau_{\psi^*}=\tau_a=0,\quad
\xi^a_\psi=\xi^a_{\psi^*}=0,\quad
\tau_t=2\xi^1_1=\dots=2\xi^n_n,\quad \xi^a_b+\xi^b_a=0,\ a\ne b,
\\[1ex]\label{Multoverdeterequation_b}
\eta_{\psi^*}=\eta_{\psi\psi}=0,\quad
2\eta_{\psi a}=i\xi^a_t,
\\[1ex]\label{Multoverdeterequation_c}
i\eta_t+\eta_{aa}+\tau V_t\psi+\xi^a V_a\psi+V\eta-(\eta_\psi-\tau_t)V\psi=0.
\end{gather}
\end{subequations}
The subsystem consisting of the equations~\eqref{Multoverdeterequation_a} and~\eqref{Multoverdeterequation_b},
which do not contain the potential $V$, is integrated immediately.
Its general solution is
\begin{gather*}
\tau=\tau(t),\quad \xi^a=\frac12\tau_tx_a+\kappa_{ab}x_b+\chi^a,\quad
\eta=\left(\frac i8\tau_{tt}|x|^2+\frac i2\chi^a_t x_a+\rho+i\sigma\right)\psi+\eta^0(t,x),
\end{gather*}
where~$\tau$, $\chi^a$, $\rho$ and~$\sigma$ are smooth real-valued functions of~$t$,
$\eta^0$ is a complex-valued function of~$t$ and~$x$,
and $(\kappa_{ab})$ is a constant skew-symmetric matrix.
Then additionally splitting the last determining equation with respect to~$\psi$,
we derive two equations,
\begin{gather}\label{MultclasscondForEta0}
i\eta^0_t+\eta^0_{aa}+V\eta^0=0,
\\\label{Multclasscond}
\tau V_t+\left(\frac12\tau_tx_a+\kappa_{ab}x_b+\chi^a\right)V_a+\tau_tV
=\frac18\tau_{ttt}|x|^2+\frac12\chi^a_{tt}x_a+\sigma_t-i\rho_t-i\frac n4\tau_{tt}.
\end{gather}
Both the equations involve the potential~$V$.
At the same time, the first equation means
that the parameter function~$\eta^0$ is an arbitrary solution of the equation~$\mathcal L_V$.
This is why the second equation is only the truly \emph{classifying condition}
for Lie symmetry generators of equations from the class~$\mathcal F$ depending on the potential $V$.

\begin{theorem}\label{MultLinsymmetryoper}
The maximal Lie invariance algebra~$\mathfrak g_V$ of an equation~$\mathcal L_V$ from~$\mathcal F$
is constituted by the vector fields of the form $D(\tau)+\sum_{a<b}\kappa_{ab}J_{ab}+P(\chi)+\sigma M+\rho I+Z(\eta^0)$,
where
\begin{gather*}
D(\tau)=\tau\p_t+\frac12\tau_tx_a\p_a+\frac18\tau_{tt}|x|^2M,\quad
J_{ab}=x_a\p_b-x_b\p_a,\quad a\ne b,\\
P(\chi)=\chi^a\p_a+\frac12\chi^a_tx_aM,\quad
M=i\psi\p_\psi-i\psi^*\p_{\psi^*},\quad
I=\psi\p_\psi+\psi^*\p_{\psi^*},\\
Z(\eta^0)=\eta^0\p_\psi+\eta^0{^*}\p_{\psi^*},
\end{gather*}
the parameters $\tau$, $\chi=(\chi^1,\dots,\chi^n)$, $\rho$ and $\sigma$ are arbitrary real-valued smooth functions of $t$
and the matrix~$(\kappa_{ab})$ is an arbitrary constant skew-symmetric matrix
that jointly satisfy the classifying condition~\eqref{Multclasscond},
and the parameter $\eta^0=\eta^0(t,x)$ runs through the solution set of the equation~$\mathcal L_V$.\looseness=-1
\end{theorem}

The kernel invariance algebra $\mathfrak g^\cap$ of the class~$\mathcal F$, $\mathfrak g^\cap:=\bigcap_{V}\mathfrak g_V$,
is obtained by splitting the conditions~\eqref{MultclasscondForEta0} and~\eqref{Multclasscond} with respect to the potential $V$
and its derivatives.
This results in $\tau=\chi^a=0$, $\kappa_{ab}=0$, $\eta^0=0$ and $\sigma_t=\rho_t=0$.

\begin{proposition}
The kernel Lie invariance algebra of the class~$\mathcal F$ is $\mathfrak g^\cap=\langle M,I\rangle$.
\end{proposition}

Following~\cite{kuru2018a},
denote by~$\mathfrak g_\spanindex$ the linear span of all vector fields given in Theorem~\ref{MultLinsymmetryoper}
when the potential $V$ varies, i.e.,
\[
\mathfrak g_\spanindex:=\left\{D(\tau)+\sum_{a<b}\kappa_{ab}J_{ab}+P(\chi)+\sigma M+\rho I+Z(\zeta)\right\}=\sum_V\mathfrak g_V.
\]
Here and in what follows
the parameters $\tau$, $\chi^a$, $\sigma$ and $\rho$ run through the set of real-valued
smooth functions of $t$, $\zeta$ runs through the set of complex-valued smooth functions of $(t,x)$
and $\eta^0$ runs through the solution set of the equation~$\mathcal L_V$ when the potential~$V$ is fixed.
We have \smash{$\mathfrak g_\spanindex=\textstyle\sum_V\mathfrak g_V$}
since each vector field~$Q$ from $\mathfrak g_\spanindex$
either with nonvanishing~$\tau$ or~$(\kappa_{ab})$ or~$\chi$
or with jointly vanishing~$\tau$, \smash{$\kappa_{ab}$}, $\chi$, $\sigma$ and $\rho$ necessarily
belongs to~$\mathfrak g_V$ for some~$V$.
Up to the antisymmetry of the Lie bracket of vector fields,
the nonzero commutation relations between vector fields spanning~$\mathfrak g_\spanindex$ are
\begin{gather*}\label{MultiLinSchEqs(1+1)comrel}
[D(\tau^1),D(\tau^2)]=D(\tau^1\tau^2_t-\tau^2\tau^1_t),\quad
[D(\tau),P(\chi)]=P\left(\tau\chi_t-\frac12\tau_t\chi\right),\\
[D(\tau),\sigma M]=\tau\sigma_t M,\quad
[D(\tau),\rho I]=\tau\rho_t I,\\
[D(\tau),Z(\zeta)]=Z\left(\tau\zeta_t+\frac12\tau_tx_a\zeta_a-\frac i8\tau_{tt}|x|^2\zeta\right),\\
[J_{ab},J_{bc}]=J_{ac},\quad a\ne b\ne c\ne a,\\[.7ex]
[J_{ab},P(\chi)]=P(\hat\chi) \quad\mbox{with}\quad \hat\chi^a=\chi^b,\quad\hat\chi^b=-\chi^a,\quad \chi^c=0,\quad a\ne b\ne c\ne a,\\[.7ex]
[J_{ab},Z(\zeta)]=Z(x_a\zeta_b-x_b\zeta_a),\quad a\ne b,\\
[P(\chi),P(\tilde \chi)]=\frac12\left(\chi^a\tilde\chi^a_t-\tilde\chi^a\chi^a_t\right)M,\quad
[P(\chi),Z(\zeta)]=Z\left(\chi^a\zeta_a-\frac i2\chi^a_t x_a\zeta\right),\\
[\sigma M,Z(\zeta)]=-Z(i\sigma \zeta),\quad
[\rho I,Z(\zeta)]=-Z(\rho\zeta).
\end{gather*}

The commutation relations between elements of~$\mathfrak g_\spanindex$ imply
that the span~$\mathfrak g_\spanindex$ is closed with respect to the Lie bracket of vector fields and hence it is a Lie algebra.
The algebra~$\mathfrak g_\spanindex$ can be represented as a semidirect sum
of the subalgebra $\mathfrak g^{\rm ess}_\spanindex:=\langle D(\tau),J_{ab},P(\chi),\sigma M ,\rho I\rangle$
and the ideal $\mathfrak g^{\rm lin}_\spanindex:=\langle Z(\zeta)\rangle$,
$\mathfrak g_\spanindex=\mathfrak g^{\rm ess}_\spanindex\lsemioplus\mathfrak g^{\rm lin}_\spanindex$.
The kernel invariance algebra $\mathfrak g^\cap$ is
an ideal in $\mathfrak g^{\rm ess}_\spanindex$ and in the entire $\mathfrak g_\spanindex$.
The above representation for $\mathfrak g_\spanindex$ induces the similar representation
for each~$\mathfrak g_V$,
\[
\mathfrak g_V=\mathfrak g^{\rm ess}_V\lsemioplus\mathfrak g^{\rm lin}_V,
\quad\mbox{where}\quad
\mathfrak g^{\rm ess}_V:=\mathfrak g_V\cap\mathfrak g^{\rm ess}_\spanindex
\quad\mbox{and}\quad
\mathfrak g^{\rm lin}_V:=\mathfrak g_V\cap\mathfrak g^{\rm lin}_\spanindex=\langle Z(\eta^0),\,\eta^0\in\mathcal L_V\rangle
\]
are a finite-dimensional subalgebra (see Lemma~\ref{LinMultdimgess} below) and an infinite-dimensional abelian ideal of~$\mathfrak g_V$, respectively.
We call $\mathfrak g^{\rm ess}_V$ the {\it essential Lie invariance algebra} of the corresponding equation~$\mathcal L_V$.
The ideal~$\mathfrak g^{\rm lin}_V$ consists of vector fields associated with transformations of {\it linear superposition}
on the solution set of the equation~$\mathcal L_V$ and therefore it is a trivial part of~$\mathfrak g_V$.

\begin{definition}
A subalgebra~$\mathfrak s$ of $\mathfrak g^{\rm ess}_\spanindex$ is called \emph{appropriate}
if there exists a potential $V$ such that~$\mathfrak s=\mathfrak g^{\rm ess}_V$.
\end{definition}

The relation of the equivalence algebra~$\mathfrak g^\sim $ and $\mathfrak g^{\rm ess}_\spanindex$ is given by
$\mathfrak g^{\rm ess}_\spanindex=\pi_*\mathfrak g^\sim$,
where $\pi_*$ is the mapping from $\mathfrak g^\sim$ onto $\mathfrak g^{\rm ess}_\spanindex$ that is induced by
the projection~$\pi$ of the joint space of the variables and the arbitrary element onto the space of the variables only.
Although the projection~$\pi$ is not a (local) diffeomorphism,
it properly pushes forward vector fields from~$\mathfrak g^\sim$ due to their structure
and hence the mapping~$\pi_*$ is well defined.
It maps the vector fields~$\hat D(\tau)$, $\hat J_{ab}$, $\hat P(\chi)$, $\hat M(\sigma)$, $\hat I(\rho)$ spanning~$\mathfrak g^\sim$
to the vector fields~$D(\tau)$, $J_{ab}$, $P(\chi)$, $\sigma M$, $\rho I$ spanning~$\mathfrak g^{\rm ess}_\spanindex$, respectively.
The above relation is stronger than
the inclusion $\mathfrak g^{\rm ess}_\spanindex\subseteq\pi_*\mathfrak g^\sim$
implied by the specific uniform semi-normalization of the class~$\mathcal F$.

As the algebra~$\mathfrak g^{\rm ess}_\spanindex$ coincides with
the set~$\pi_*\mathfrak g^\sim$ of infinitesimal generators of one-parameter subgroups
of the group~$\pi_*G^\sim$ which is the projection of~$G^\sim$ by~$\pi$, the structure of $\mathfrak g^{\rm ess}_\spanindex$ is agreed with
the action of~$\pi_*G^\sim$ on this algebra.
Moreover, both $\mathfrak g^{\rm ess}_\spanindex$ and $\mathfrak g^{\rm lin}_\spanindex$
are invariant with respect to the action of the group~$\pi_*G^\sim$.
This is why the action of~$G^\sim$ on equations from the class~$\mathcal F$
induces the well-defined action of~$\pi_*G^\sim$ on
the essential Lie invariance algebras of these equations,
which are subalgebras of $\mathfrak g^{\rm ess}_\spanindex$.
The kernel $\mathfrak g^\cap$ is obviously an ideal in~$\mathfrak g^{\rm ess}_V$ for any~$V$.
\noprint{
Moreover, the problem of group classification of the class~$\mathcal F$ is solved
if we classify all appropriate subalgebras of the algebra~$\mathfrak g^{\rm ess}_\spanindex$ with respect to the equivalence relation generated by~$\pi_*G^\sim$.
}

Summarizing altogether we have the following proposition:
\begin{proposition}
The problem of group classification of (1+n)-dimensional linear Schr\"odinger equations 
with complex-valued potentials of the form~\eqref{MLinSchEqs} 
reduces to the classification of appropriate subalgebras of the algebra~$\mathfrak g^{\rm ess}_\spanindex$ 
with respect to the equivalence relation generated by the action of~$\pi_*G^\sim$.
\end{proposition}

\subsection{Further properties of Lie invariance algebras}\label{sec:FurtherPropertiesOfLieInvAlgs}

In order to classify all appropriate subalgebras of~$\mathfrak g^{\rm ess}_\spanindex$,
we start with describing the action of the transformations from~$\pi_*G^\sim$
on the vector fields from~$\mathfrak g^{\rm ess}_\spanindex$.
For this purpose, we compute the pushforwards of spanning vector fields of~$\mathfrak g^{\rm ess}_\spanindex$
by elementary transformations from~$\pi_*G^\sim$.

Given $\Phi\in\pi_*G^\sim $ and $Q\in\mathfrak g^{\rm ess}_\spanindex$,
the pushforward~$Q$ by~$\Phi$ is computed as
\[
\tilde Q:=\Phi_*Q=Q(T)\p_{\tilde t}+Q(X^a)\p_{\tilde a}+Q(\Psi)\p_{\tilde \psi}+Q(\Psi^*)\p_{\tilde\psi^*},
\]
where into each coefficient of~$\tilde Q$ we substitute
the expressions of the variables without tildes in terms of the variables with tildes,
$(t,x_a,\psi,\psi^*)=\Phi^{-1}(\tilde t,\tilde x_a,\tilde\psi,\tilde\psi^*)$,
and $\Phi^{-1}$ denotes the inverse of~$\Phi$.

By~$\mathcal D(T)$, $\mathcal J(O)$, $\mathcal P(\mathcal X)$, $\mathcal M(\Sigma)$ and $\mathcal I(\Upsilon)$ we denote
the transformations of the form~\eqref{Mequgroupoidsub_a}--\eqref{Mequgroupoidsub_b} with $\Lambda=0$ and $\varepsilon=1$,
where the parameters $T$, $O$, $\mathcal X$, $\Sigma$ and~$\Upsilon$,
successively excluding one of them, are set to the trivial values,
which are $t$ for $T$, $E$ for $O$ and zeroes for $\mathcal X$, $\Sigma$ and~$\Upsilon$.
These transformations, referred to as \emph{elementary transformations}, generate the entire group~$\pi_*G^\sim$.
The nonidentity pushforward actions of elementary transformations
to the vector fields spanning~$\mathfrak g^{\rm ess}_\spanindex$
are exhausted by
\begin{gather*}
\mathcal D_*(T)D(\tau)=D(\tilde \tau),\quad
\mathcal D_*(T)P(\chi)=P(\tilde \chi),\quad
\mathcal D_*(T)(\sigma M)=\tilde \sigma\tilde M,\quad
\mathcal D_*(T)(\rho I)=\tilde\rho\tilde I,\\[.5ex]
%\mathcal D_*(T)J_{ab}=\tilde J_{ab},\quad
%\mathcal D_*(T)(\rho I)=\rho\tilde I,\\
%{\mathcal J_*(O)D(\tau)=\tilde D(\tau),\\
\mathcal J_*(O)P(\chi)=\tilde P(O\chi),\quad
\mathcal J_*(O)J_{ab}=\sum_{c<d}(O^{db}O^{ca}-O^{da}O^{cb})\tilde J_{cd},\\
%\mathcal J_*(O)(J_{bc})=\tilde J_{bc},\quad
%\mathcal J_*(O)(\sigma M)=\sigma\tilde M,\quad
%\mathcal J_*(O)(\rho I)=\rho\tilde I,\quad
\mathcal P_*(\mathcal X)D(\tau)=\tilde D(\tau)+\tilde P\left(\tau \mathcal X_t-\frac{\tau_t}2\mathcal X\right)
+\left(\frac{\tau_{tt}}8\mathcal X^a\mathcal X^a-\frac{\tau_t}4\mathcal X^a\mathcal X^a_t
-\frac\tau4(\mathcal X^a\mathcal X^a_{tt}-\mathcal X^a_t\mathcal X^a_t)\right)\tilde M,\\
\mathcal P_*(\mathcal X)J_{ab}=\tilde J_{ab}+P(\hat{\mathcal X})-\frac12(\mathcal X^a\mathcal X^b_t-\mathcal X^b\mathcal X^a_t)\tilde M,\\
\mathcal P_*(\mathcal X)P(\chi)=\tilde P(\chi)+\frac12(\chi^a\mathcal X^a_t-\chi^a_t\mathcal X^a)\tilde M,\\[.3ex]
%\mathcal P_*(\mathcal X)(\rho I)=\rho\tilde I,\quad
\mathcal M_*(\Sigma)D(\tau)=\tilde D(\tau)+\tau\Sigma_t\tilde M,\quad
\mathcal I_*(\Upsilon)D(\tau)=\tilde D(\tau)+\tau\Upsilon_t\tilde I,\quad
\end{gather*}
where
$\tilde \tau(\tilde t)=(T_t\tau)(T^{-1}(\tilde t))$,
$\tilde \chi(\tilde t)=(|T_t|^{1/2}\chi)(T^{-1}(\tilde t))$,
$\hat{\mathcal X}^a=\mathcal X^b$,
$\hat{\mathcal X}^b=-\mathcal X^a$,
$\hat{\mathcal X}^c=0,\,c\ne a,b$,
$\tilde \sigma=\sigma(T^{-1}(\tilde t))$,
$\tilde\rho=\rho(T^{-1}(\tilde t))$
and into each pushforward by~$\mathcal D_*(T)$ we should substitute
the expression for~$t$ given by inverting the relation $\tilde t=T(t)$;
$t=\tilde t$ for the other pushforwards.
Tildes over vector fields mean that these vector fields are represented in the new variables.

\begin{lemma}\label{LinMultdimgess}
$\dim\mathfrak g^{\rm ess}_V\leqslant\dfrac{n(n+3)}2+5$ for any potential $V$,
and this upper bound is least.
\end{lemma}

\begin{proof}
Fix a potential~$V$. Then
similarly to the proof of Lemma~1 in \cite{kuru2018a}, the classifying condition~\eqref{Multclasscond} implies
a system of linear ordinary differential equations of the normal form
\begin{gather*}
\tau_{ttt}=\gamma^{00}\tau_t+\gamma^{01}\tau+\gamma^{0,a+1}\chi^a+\theta^{0ab}\kappa_{ab},\\
\chi^c_{tt}=\gamma^{c0}\tau_t+\gamma^{c1}\tau+\gamma^{c,a+1}\chi^a+\theta^{cab}\kappa_{ab},\\
\sigma_t=\gamma^{n+1,0}\tau_t+\gamma^{n+1,1}\tau+\gamma^{n+1,a+1}\chi^a+\theta^{n+1,ab}\kappa_{ab},\\
\rho_t=-\frac n4\tau_{tt}+\gamma^{n+2,0}\tau_t+\gamma^{n+2,1}\tau+\gamma^{n+2,a+1}\chi^a+\theta^{n+2,ab}\kappa_{ab},
\end{gather*}
where the coefficients~$\gamma^{pq}$ and $\theta^{pab}$, $p=0,\dots,n+2$, $q=0,\dots,n+1$, $a<b$, are functions of~$t$.
Therefore, the dimension of $\mathfrak g^{\rm ess}_V$ is bounded by the sum of the number of pairs $(a,b)$ with $a<b$
(which is associated with rotations and equals $n(n-1)/2$)
and the number of arbitrary constants in the general solution of the above system (which equals $2n+5$).
The derived upper bound is least since it coincides with $\dim\mathfrak g^{\rm ess}_V$
for the zero~$V$ corresponding to the free Schr\"odinger equation.
\end{proof}

\begin{corollary}
$\dim\mathfrak g^{\rm ess}_V\cap\langle P(\chi),\sigma M,\rho I\rangle \leqslant 2n+2$.
\end{corollary}

\begin{proof}
We follow the proof of Lemma~\ref{LinMultdimgess} by omitting the first equation of the system from this proof and setting $\tau=0$ and $\kappa_{ab}=0$.
\end{proof}

\begin{corollary}\label{cor:LinMultiDSchEqsKernelAlgebra}
 $\mathfrak g^{\rm ess}_V\cap\langle \sigma M,\rho I\rangle=\mathfrak g^\cap$ for any potential $V$.
\end{corollary}

\begin{proof}
Analogously, we preserve only the last two equations of the system from the proof of Lemma~\ref{LinMultdimgess}
and set $\tau=\chi^a=0$ and $\kappa_{ab}=0$.
\end{proof}

\begin{lemma}\label{LinMultgclas}
For all $V$, $\pi^0_*\mathfrak g^{\rm ess}_V $ is a Lie algebra and $\dim\pi^0_*\mathfrak g^{\rm ess}_V\leqslant 3$.
Moreover,
\[
\pi^0_*\mathfrak g^{\rm ess}_V\in
\{0,\langle \p_t\rangle,\langle \p_t,t\p_t\rangle,\langle \p_t,t\p_t,t^2\p_t\rangle\}
\bmod\pi^0_*G^\sim,
\]
where $\pi^0$ denotes the projection on the space of the variable $t$ and
$\pi^0_*\mathfrak s\subset\pi^0_*\mathfrak g^{\rm ess}_\spanindex=\langle \tau\p_t\rangle$.
\end{lemma}

\begin{proof}
The proof is similar to the one given in~\cite{kuru2018a}.
It is based on the famous Lie classification of finite dimensional algebras of vector fields on the real line.
\end{proof}

\section{Complete group classification of (1+2)-dimensional\\ 
linear Schr\"odinger equations with complex-valued potentials}\label{sec:GroupClassificationOf(1+2)DLinSchEqs}

We complete the group classification of linear Schr\"odinger equations 
of the general form~\eqref{MLinSchEqs} only in the case $n=2$.
Even the case $n=3$ is much more cumbersome and requires more study, not to mention the case of general~$n$.
Specifying results of Section~\ref{sec:AnalysOfDetEqsForLieSymsOfMultiDSchEqs} to the case $n=2$,
we obtain that for any $V$
\begin{gather*}
\dim\mathfrak g^{\rm ess}_V\leqslant 10,\quad
\mathfrak g^{\rm ess}_V \cap\langle\sigma M,\rho I\rangle =\langle M,I\rangle=\mathfrak g^\cap,\\
2\leqslant\dim\mathfrak g^{\rm ess}_V \cap\langle P(\chi),\sigma M,\rho I\rangle \leqslant 6,\quad
\dim\pi^0_*\mathfrak g^{\rm ess}_V\leqslant 3.
\end{gather*}
For each appropriate subalgebra~$\mathfrak s$ of $\mathfrak g^{\rm ess}_{\langle \,\rangle}$, similarly to~\cite{kuru2018a}
we introduce five~$\pi_*G^\sim$-invariant integers in order to arrange the process of the group classification:
\begin{gather*}
k_0:=\dim\mathfrak s\cap\langle \sigma M,\rho I\rangle=\dim\mathfrak g^\cap=2,\\
k_1:=\dim\mathfrak s\cap\langle P(\chi),\sigma M,\rho I\rangle-k_0\in \{0,1,2,3,4\},\\
k_2:=\dim\mathfrak s\cap\langle J,\,P(\chi),\,\sigma M,\,\rho I \rangle-k_1-k_0\in \{0,1\},\\
k_3:=\dim\mathfrak s-k_2-k_1-k_0=\dim\pi^0_*\mathfrak s\in \{0,1,2,3\},\\
r_0:=\rank\{\chi\mid \exists\,\sigma,\rho \colon P(\chi)+\sigma M+\rho I\in \mathfrak s\}\in \{0,1,2\},
\end{gather*}
where $\pi^0$ denotes the projection on the space of the variable $t$.
Note that \[\dim\mathfrak s=k_0+k_1+k_2+k_3\leqslant 10.\]
The parameter tuple $(k_2,r_0,k_3)$ is used for the labeling of possible classification cases.  
In fact, not all values of $(k_2,r_0,k_3)$ in $\{0,1\}\times\{0,1,2\}\times\{0,1,2,3\}$ are appropriate.

\begin{lemma}\label{lem:k2r0}
 $(k_2,r_0)\ne(1,1)$.
\end{lemma}

\begin{proof}
Suppose that $k_2=1$ and $r_0\geqslant1$.
Then the algebra $\mathfrak g^{\rm ess}_V$ necessarily contains vector fields of the form
$Q^0=J+P(\chi^{01},\chi^{02})+\sigma^0 M+\rho^0 I$ and
$Q^1=P(\chi^{11},\chi^{12})+\sigma^1M+\rho^1 I$, where $(\chi^{11},\chi^{12})\ne (0,0)$.
The commutator of these vector fields
\[
[Q^0,Q^1]=P(\chi^{12},-\chi^{11})+\frac12 (\chi^{01}\chi^{11}_t+\chi^{02}\chi^{12}_t-\chi^{01}_t\chi^{11}-\chi^{02}_t\chi^{12})M
\]
should belong to $\mathfrak g^{\rm ess}_V$, and hence $r_0=2$.
Therefore, the case $k_2=r_0=1$ is impossible.
\end{proof}

The other constraint for appropriate values of $(k_2,r_0,k_3)$, $k_3\ne2$ if $(k_2,r_0)\ne(0,0)$, 
is not obvious at all and can be derived only by analyzing the complete proof of group classification 
for the class~$\mathcal F$ with $n=2$.

\begin{lemma}\label{Multgclaslema}
Up to the $\pi_*G^\sim$-equivalence, a vector field $P(\chi)+\sigma M+\rho I$,
where the components of the tuple~$\chi$ are linearly independent,
reduces to $P(h\cos t,h\sin t)+\tilde \rho I$ with nonzero $h=h(t)$.
\end{lemma}

\begin{proof}
We can represent the tuple $\chi$ as $\chi=(h\cos T,h\sin T)$,
where $h$ and~$T$ are smooth functions of~$t$ with $T_t\ne0$.
Then we successively apply the transformations $\mathcal D_*(T)$ and $\mathcal P_*(\mathcal X)$
with an appropriate tuple $\mathcal X$ of smooth functions of $t$
to the vector field and thus set $T=t$ and $\sigma=0$, which gives the required form.
\end{proof}

\begin{lemma}\label{Multlemmatauzero}
If $P(1,0)+\rho^1I\in\mathfrak g^{\rm ess}_V$,
then the vector field  $P(t,0)+\rho^2 I$ with $\rho^2=\int t\rho^1_t \,{\rm d}t$
also belongs to $\mathfrak g^{\rm ess}_V$.
\end{lemma}

\begin{proof}
Suppose that $P(1,0)+\rho^1I \in \mathfrak g^{\rm ess}_V$.
We substitute the components of this vector field into the classifying condition~\eqref{Multclasscond},
which gives $V_1=-i\rho^1_t$.
This equation coincides with the one obtained by evaluating the classifying condition~\eqref{Multclasscond}
at $\tau=\sigma=0$, $\chi=t$ and $\rho^2=\int t\rho^1_t \,{\rm d}t$.
\end{proof}

Putting all the above ingredients together, we have the following:
any appropriate subalgebra of $\mathfrak g^{\rm ess}_\spanindex$ is spanned by
\begin{itemize}\itemsep=0ex
\item the basis vector fields $M$ and $I$ of the kernel $\mathfrak g^\cap$,
\item $k_1$ vector fields $P(\chi^{p1},\chi^{p2})+\sigma^pM+\rho^pI$, $p=1,\dots,k_1$,
with linearly independent tuples $\chi^1$, \dots, $\chi^{k_1}$,
\item $k_2$ vector fields $J+P(\chi^{01},\chi^{02})+\sigma^0M+\rho^0I$,
\item $k_3$ vector fields $D(\tau^q)+\kappa_q J+P(\chi^{q1},\chi^{q2})+\sigma^qM+\rho^qI$, $q=k_1+1,\dots,k_1+k_3$,
with linearly independent $\tau^{k_1+1}$, \dots, $\tau^{k_1+k_3}$.
\end{itemize}

In what follows we use the notation
\[
|x|=\sqrt{x_1^2+x_2^2},\quad \phi=\arctan x_2/x_1,\quad
\omega_1=x_1\cos t+x_2\sin t,\quad \omega_2=-x_1\sin t+x_2\cos t.
\]

\begin{theorem}\label{thm:GroupClassifictionOf(1+2)DLinSchEqs}
\looseness=-1
A complete list of inequivalent essential Lie-symmetry extensions 
in the class~$\mathcal F$ of (1+2)-dimensional linear Schr\"odinger equations
with complex-valued potentials is exhausted by the cases presented below.
Here $U$ is an arbitrary complex-valued smooth function of its arguments or an arbitrary complex constant
that satisfies constraints indicated in the corresponding cases, and the other functions and constants take real values.
\end{theorem}

\begin{enumerate}\setcounter{enumi}{-1}
\newcommand{\cc}{\item}\itemsep=.5ex

\cc
$V=V(t,x)$:\quad $\mathfrak g^{\rm ess}_V=\mathfrak g^\cap=\big\langle M,\,I\big\rangle$.

\cc
$V=U(x_1,x_2)$:\quad $\mathfrak g^{\rm ess}_V=\big\langle M,\,I,\, D(1)\big\rangle$.

\cc
$V=U(\omega_1,\omega_2)$:\quad $\mathfrak g^{\rm ess}_V=\big\langle M,\, I,\, D(1)+J\big\rangle$.

\cc
$V=|x|^{-2}U(\zeta)$, $\zeta=\phi-2\beta\ln|x|$, $\beta>0$, $U_\zeta\ne0$:\quad
$\mathfrak g^{\rm ess}_V=\big\langle M,\, I,\, D(1),\, D(t)+\beta J\big\rangle.$

\cc
$V=|x|^{-2}U(\phi)$, $U_\phi\ne0$:\quad
$\mathfrak g^{\rm ess}_V=\big\langle M,\, I,\, D(1),\, D(t),\, D(t^2)-tI\big\rangle$.

\cc
$V=U(t,|x|)+(\sigma_t-i\rho_t)\phi$, $\sigma\in \{0,t\}\bmod G^\sim$
and, if $\sigma=0$, $\rho \in\{0,t\}\bmod G^\sim $:\\
$\mathfrak g^{\rm ess}_V=\big\langle M,\, I,\, J+\sigma M+\rho I\big\rangle$.

\cc
$V=U(|x|)+(\alpha_1-i\beta_1)\phi$:\quad $\mathfrak g^{\rm ess}_V=\big\langle M,\, I,\, J+\alpha_1 tM+\beta_1 tI,\, D(1)\big\rangle$.

\cc
$V=|x|^{-2}U,\ U\ne0$:\quad $\mathfrak g^{\rm ess}_V=\big\langle M,\, I,\, J,\, D(1),\, D(t),\, D(t^2)-tI\big\rangle$.

\cc
$V=U(t,x_2)+i\gamma (t)x_1$:\\
$\mathfrak g^{\rm ess}_V=\big\langle M,\,I,\,P(1,0)-\left(\int \gamma(t)\,{\rm d}t\right)I,\,P(t,0)-\left(\int t\gamma(t)\,{\rm d}t\right)I\big\rangle$.

\cc
$V=U(\zeta)+i\beta x_1$, $\zeta=x_2$:\\
$\mathfrak g^{\rm ess}_V=\big\langle M,\, I,\, P(1,0)-\beta tI,\, P(t,0)-\frac{\beta}2 t^2I,\, D(1)\big\rangle$.

\cc
$V=t^{-1}U(\zeta)+i\beta |t|^{-3/2}x_1$, $\zeta=|t|^{-1/2}x_2$:\\
$\mathfrak g^{\rm ess}_V=\big\langle M,\, I,\, P(1,0)-2\beta t|t|^{-3/2} I,\, P(t,0)-2\beta t|t|^{-1/2} I,\, D(t)\big\rangle$.

\cc
$V=(t^2+1)^{-1}U(\zeta)+i\beta (t^2+1)^{-3/2}x_1$, $\zeta=(t^2+1)^{-1/2}x_2$:\\
$\mathfrak g^{\rm ess}_V=\big\langle M,\, I,\, P(1,0)-\beta t(t^2+1)^{-1/2} I,\, P(t,0)+\beta (t^2+1)^{-1/2}I,\, D(t^2+1)-tI\big\rangle$.

\cc
$V=U x_2^{-2}$, $U\ne0$:
$\mathfrak g^{\rm ess}_V=\big\langle M,\, I,\, P(1,0),\, P(t,0),\, D(1),\, D(t),\, D(t^2)-tI\big\rangle$.

\cc
$V=U(t,\omega_2)+\frac14(h^{-1}h_{tt}-1)\omega^2_1+h_th^{-1}\omega_1\omega_2-i\rho_th^{-1}\omega_1$, $h=h(t)\ne0$:\\
$\mathfrak g^{\rm ess}_V=\big\langle M,\, I,\, P(h\cos t,h\sin t)+\rho I\big\rangle$.

\cc
$V=U(\omega_2)+\frac14(\beta-1)\omega_1^2-\beta\omega_1\omega_2-i\alpha\beta\omega_1$, $\beta\ne0$:\\
$\mathfrak g^{\rm ess}_V=\big\langle M,\, I,\, P(e^{\beta t}\cos t,e^{\beta t}\sin t)+\alpha e^{\beta t}I,\, D(1)+J\big\rangle$.

\cc
$V=U(\omega_2)-\frac14\omega_1^2+(\tilde\alpha-i\alpha)\omega_1$:\quad
$\mathfrak g^{\rm ess}_V=\big\langle M,\, I,\, P(\cos t,\sin t)+\tilde\alpha tM+\alpha tI,\, D(1)+J\big\rangle$.

\cc
$V=\frac14h^{ab}(t)x_ax_b+ih^{0a}(t)x_a$, $h^{12}=h^{21} $:\\
$\mathfrak g^{\rm ess}_V=\big\langle M,\, I,\, P(\chi^{p1},\chi^{p2})+\rho^pI,\, p=1,\dots,4\big\rangle,$
where $\{(\chi^{p1}(t),\chi^{p2}(t))\}$ is a fundamental set of solutions
of the system $\chi^a_{tt}=h^{ab}\chi^b$, and $\rho^p =-\int h^{0a}\chi^{pa} {\rm d}t$.

\cc
$V=\frac14\alpha x_1^2+\frac14\beta x_2^2+i\nu_ax_a$, $\alpha\ne\beta $ or $(\nu_1,\nu_2)\ne (0,0)$:\\
$\mathfrak g^{\rm ess}_V=\big\langle M,\, I,\, P(\chi^{11},0)+\rho^1I,\, P(\chi^{21},0)+\rho^2I,\, P(0,\chi^{32})+\rho^3I,\, P(0,\chi^{42})+\rho^4I,\, D(1)\big\rangle$,
where $\{\chi^{11}(t),\chi^{21}(t)\}$ is a fundamental set of solutions of the equation $\chi^1_{tt}=\alpha\chi^1$,
$\rho^p=-\nu_1\int\chi^{p1}{\rm d}t$, $p=1,2$, and
$\{\chi^{32}(t),\chi^{42}(t)\}$ is a fundamental set of solutions of the equation
$\chi^2_{tt}=\beta\chi^2$, $\rho^p=-\nu_2\int\chi^{p2}{\rm d}t$, $p=3,4$.

\cc
$V=\frac14\alpha\omega_1^2+\frac14\beta\omega_2^2+i\nu_a\omega_a$, $\alpha\ne\beta$ or $(\nu_1,\nu_2)\ne(0,0)$:\\
$\mathfrak g^{\rm ess}_V=\big\langle M,\, I,\, P(\theta^{p1}\cos t-\theta^{p2}\sin t,\theta^{p1}\sin t+\theta^{p2}\cos t)+\rho^pI,\, D(1)+J,\, p=1,\dots,4\big\rangle$,
where $(\theta^{p1}(t),\theta^{p2}(t))$ are linearly independent solutions of the system $\theta^1_{tt}-2\theta^2_t=(1+\alpha)\theta^1$, $\theta^2_{tt}+2\theta^2_t=(1+\beta)\theta^2$, and $\rho^p=-\int\chi^{pa}h^{0a}{\rm d}t$, $p=1,\dots,4$.

\cc
$V=0$:\quad $\mathfrak g^{\rm ess}_V= \big\langle M,\, I,\, P(1,0),\, P(t,0),\, P(0,1),\, P(0,t),\, J,\, D(1),\, D(t),\, D(t^2)-tI\big\rangle$.

\end{enumerate}

\begin{remark}
The essential Lie invariance algebras listed in Theorem~\ref{thm:GroupClassifictionOf(1+2)DLinSchEqs}
are really maximal for the corresponding potentials
if these potentials are $G^\sim$-inequivalent to listed potentials with larger essential Lie invariance algebras.
In some cases, we have presented simple necessary and sufficient conditions that provide such inequivalence.
In other cases, these conditions are not so obvious.
For example, in Cases~9--11 the condition of maximal Lie symmetry extension is
($\beta\ne0$ or $(\zeta^2U)_\zeta\ne0$) and ($U_{\zeta\zeta\zeta}\ne0$ or $\mathop{\rm Im}U_{\zeta\zeta}\ne0$),
which excludes the values of~$V$ that are $G^\sim$-equivalent to those in Cases~12 and~16--19.
Case~8 is already more complicated since for it the similar condition
($\gamma\ne0$ or $(x_2^2U)_2\ne0$) and ($U_{222}\ne0$ or $\mathop{\rm Im}U_{22}\ne0$)
should be extended with the condition excluding potentials $G^\sim$-equivalent to those in Cases~9--11.
Analogously, potentials in Cases~13--15 are $G^\sim$-inequivalent to those in Cases~16--19
if and only if $U_{\omega_2\omega_2\omega_2}\ne0$ or $\mathop{\rm Im}U_{\omega_2\omega_2}\ne0$.
This condition is necessary and sufficient for the maximality of Lie symmetry extensions given in Cases~14 and~15,
and for Case~13 it should be extended to guarantee the exclusion of potentials related to Cases~14 and~15.
Schr\"odinger equations related to Case~16 are not similar to the free Schr\"odinger equation presented in Case~19
if and only if the parameter functions~$h$ satisfy at least one of the conditions
$h^{12}\ne h^{21}$, $h^{12}=h^{21}\ne0$, $h^{11}\ne h^{22}$, $h^{01}\ne0$ and $h^{02}\ne0$.
We need additional constraints on these functions to exclude, up to the $G^\sim$-equivalence, potentials of Cases~17 and~18.
\end{remark}

\begin{remark}
In Theorem~\ref{thm:GroupClassifictionOf(1+2)DLinSchEqs}, we in general neglect
possible gauges of constant parameters in~$V$ by discrete and scaling equivalence transformations. 
For example, 
a nonzero parameter among~$\alpha_1$ and~$\beta_1$ in Case~6, 
the parameter~$\beta$ in Case~9
and the parameter~$\nu_1$ in Case~17
can be set to one. 
The parameter~$\beta$ in Cases~10 and~11 and the imaginary part of~$U$ in Cases~7 and~12 can be made nonnegative. 
In Case~18 and, if $\nu_1=\nu_2$, in Case~17, we can also assume $\alpha<\beta$.
\end{remark}

\begin{remark}
More generally, the cases that are listed in Theorem~\ref{thm:GroupClassifictionOf(1+2)DLinSchEqs} are not equivalent to each other. 
At the same time, each of Cases~0--18 is a family of cases of essential Lie symmetry extensions 
within the class~$\mathcal F$ rather than a single case, 
and there is an equivalence within each of these case families 
that is induced by the action of a subgroup of~$G^\sim$. 
For example, see below the consideration of the subclass~$\mathcal F_{{\rm l}x_1}$ 
of (1+2)-dimensional linear Schr\"odinger equations
with potentials of the form~\eqref{EqPotentialsWith2ProportionalChis}, 
in particular the formula~\eqref{EqRelationBetweenPotentialsWith2ProportionalChis}, 
for the equivalence of potentials in Case~8.
Nevertheless, it is usually not convenient or even not feasible to take such an equivalence into account. 
\end{remark}

\begin{proof}
We single out different cases depending on values of the integers~$k_3$, $k_2$ and~$r_0$.
Basis vector fields of a Lie symmetry extension is denoted by
\[
Q^s=D(\tau^s)+\kappa_s J+P(\chi^{s1},\chi^{s2})+\sigma^sM+\rho^sI,
\]
where the range of~$s$ is equal to $\dim\mathfrak g^{\rm ess}_V-2$,
$\kappa_s$ are real constants,
and all the other parameters are real-valued functions of~$t$.

\medskip\par\noindent
$\boldsymbol{k_2=r_0=0}$.
The corresponding Lie symmetry extensions are spanned by~$\{Q^s\}$ with linearly independent~$\tau^s$.
By Lemma~\ref{LinMultgclas}, $\langle\pi^0_*Q^s\rangle$
is an algebra isomorphic to a subalgebra of the algebra ${\rm sl}(2,\mathbb R)$.
Hence the group classification in this case reduces to
the classification of subalgebras of the algebra ${\rm sl}(2,\mathbb R)$.
Varying $k_3$ we single out the following four subcases.

\medskip\par\noindent
$k_3=0.$ This is the general case with no extension,
i.e., $\mathfrak g^{\rm ess}_V=\mathfrak g^\cap$ (Case~0).

\medskip\par\noindent
$k_3=1$.
The algebra $\mathfrak g^{\rm ess}_V$ contains
a single linearly independent vector field $Q^1$ with $\tau^1\ne0$.
Pushing $Q^1$ forward by a transformation from $\pi_*G^\sim$,
we can set $\tau^1=1$, $\kappa_1\in\{0,1\}$ and $\chi^{1a}=\sigma^1=\rho^1=0$,
i.e., the vector field $Q^1$ reduces to the form $Q^1=D(1)+\kappa_1 J$.
Evaluating the classifying condition~\eqref{Multclasscond} at~$Q^1$ with $\kappa_1=0$ and $\kappa_1=1$
gives differential equations in~$V$
whose general solutions are presented in Cases~1 and~2, respectively.

\medskip\par\noindent
$k_3=2.$
The Lie symmetry extension is given by vector fields~$Q^1$ and~$Q^2$ with linearly independent~$\tau^1$ and~$\tau^2$.
Modulo the $\pi_*G^\sim$-equivalence, we can set $\tau^1=1$ and $\tau^2=t$ (cf.\ Lemma~\ref{LinMultgclas})
and, as in the previous case, $\chi^{1a}=\sigma^1=\rho^1=0$.
Then, the condition $[Q^1,Q^2]\in \mathfrak g^{\rm ess}_V$ implies that
\[D(1)+P(\chi^{21}_t,\chi^{22}_t)+\sigma^2_tM+\rho^2_tI=Q^1+\alpha_1M+\beta_1I,\]
where $\alpha_1$ and $\beta_1$ are real constants.
This equality splits into the constraints $\kappa_1=0$, $\chi^{2a}_t=0$, $\sigma^2_t=\alpha_1$, $\rho^2_t=\beta_1$.
Taking these constraints into account,
we push~$Q^1$ and $Q^2$ forward by transformations from $\pi_*G^\sim$
and linearly combine them with $M$ and $I$.
This allows us to set $\chi^{2a}=\sigma^2=\rho^2=0$ and $\beta:=\kappa_2\geqslant0$ when preserving $Q^1$.
The integration of the system in $V$  obtained
by evaluating the classifying condition~\eqref{Multclasscond} at~$Q^1$ and at~$Q^2$
leads to~Case~3.
For the extension to be maximal, we need
$\beta\ne0$ (otherwise we get Case~4 with $k_3=3$) and
$U_\zeta\ne0$ (otherwise we get Case~7).

\medskip\par\noindent
$k_3=3.$
The algebra~$\mathfrak g^{\rm ess}_V$ is spanned by the basis elements of the kernel algebra
and vector fields $Q^q$ with linearly independent coefficients~$\tau^q$, $q=1,2,3$.
In view of Lemma~\ref{LinMultgclas}, we can assume that $\tau^1=1$, $\tau^2=t$ and $\tau^3=t^2$.
The commutation relations of~$\mathfrak g^{\rm ess}_V$ imply $\kappa_q=0$.
Similarly to the previous case, we can also set $\chi^{qa}=\sigma^q=\rho^q=0$, $q=1,2$.
The commutation relations of~$Q^3$ with~$Q^1$ and~$Q^2$ implies $\chi^{3a}=0$,
$\sigma^3_t=\alpha_1$ and $\rho^3_t=\alpha_2$, where $\alpha_1$ and $\alpha_2$ are real constants.
Up to linearly combining with~$M$ and~$I$, the vector field~$Q^3$ reduces to
$Q^3=D(t^2)+\alpha_1t M+\alpha_2tI$.
The classifying condition~\eqref{Multclasscond} evaluated successively at~$Q^1$, $Q^2$ and~$Q^3$ 
gives the system
\[
V_t=0,\quad
tV_t+\frac 12 x_aV_a+V=0,\quad
t^2V_t+tx_aV_a+2tV=\alpha_1-i\alpha_2-i,
\]
whose compatibility with respect to~$V$ implies $\alpha_1=0$ and $\alpha_2=-1$.
Thus, we have~Case~4,
where the extension is maximal if and only if $U_\phi\ne0$
since otherwise it becomes Case~7.

\medskip\par\noindent
$\boldsymbol{k_2=1,r_0=0}.$
The algebra $\mathfrak g^{\rm ess}_V$
contains a vector field $Q^0$ with $\tau^0=0$ and $\kappa_0=1$,
and additional extensions are realized by vector fields $Q^q$ with linearly independent coefficients~$\tau^q$.
The commutation relation for~$Q^0$ and~$Q^q$ implies that $\chi^{qa}=0$.
Acting by $\mathcal P_*(\mathcal X)$ with $\mathcal X^1=\chi^{02}$ and $\mathcal X^2=-\chi^{01}$ on $\mathfrak g^{\rm ess}_V$,
we set $\chi^{0a}=0$, and thus $Q^0=J+\sigma^0 M+\rho^0 I$.
Similarly to the case $k_2=r_0=0$, we partition Lie symmetry extensions into three subcases depending on values of~$k_3$.
Moreover, the simplification of the parameter functions~$\tau^q$, $\sigma^q$ and~$\rho^q$ up to the $\pi_*G^\sim$-equivalence
jointly with linearly combining completely coincides with that in the case $k_2=r_0=0$.

\medskip\par\noindent
$k_3=0$.
If the function $\sigma^0$ (resp.\ $\rho^0$) is constant, then it can be made zero by linearly combining~$Q^0$ with~$M$ (resp.\ $I$).
If one of these functions is not constant, up to the $\pi_*G^\sim$-equivalence it can be set to be equal to $t$.
Then the classifying condition~\eqref{Multclasscond} implies an equation in~$V$
%\[x_1V_2-x_2V_1=\sigma_t-i\rho_t,\]
whose general solution is presented in~Case~5.

\medskip\par\noindent
$k_3=1.$
The vector field $Q^1$ is reduced to $Q^1=D(1)$ up to the $\pi_*G^\sim$-equivalence.
As the Lie algebra $\mathfrak g^{\rm ess}_V$ is closed with respect to Lie bracket of vector fields, we obtain
\[
[Q^1,Q^0]=\sigma^0_t M+\rho^0_t I=\alpha_1 M+\beta_1I,
\]
where $\alpha_1$ and $\beta_1$ are real constants.
Equating the components of the vector fields on both the sides of the last equality
and then solving the obtained equations,
we derive $\sigma^0=\alpha_1 t+\alpha_0$ and $\rho^0=\beta_1t+\beta_0$,
where $\alpha_0$ and $\beta_0$ are real constants, which can be set to zero due to linearly combining with~$M$ and~$I$.
The vector field~$Q^0$ is reduced to $Q^0=J+\alpha_1 t M+\beta_1t$.
Successively substituting the components of~$Q^0$ and~$Q^1$
into the classifying condition~\eqref{Multclasscond} yields two independent equations in $V$:
$V_t=0$, $x_1V_2-x_2V_1=\alpha_1-i\beta_1$, whose general solution is presented in Case~6.

\medskip\par\noindent
$k_3\geqslant 2.$
The algebra $\mathfrak g^{\rm ess}_V$ necessarily contains
the vector fields~$M$ and~$I$ from~$\mathfrak g^\cap$
and vector fields $Q^0$ and $Q^q$, $q=1,2$, with linearly independent $\tau^1$ and $\tau^2$.
Up to the $\pi_*G^\sim$-equivalence the vector fields $Q^1$ and $Q^2$
are reduced to the form $Q^1=D(1)$ and $Q^2=D(t)$.
As the commutators $[Q^q,Q^0]$ belong to $ \mathfrak g^{\rm ess}_V$, we have
\begin{gather*}
[Q^1,Q^0]=\sigma^0_tM+\rho^0_tI=\alpha_1M+\beta_1I,\\
[Q^2,Q^0]=t\sigma^0_tM+t\rho^0_tI=\alpha_2M+\beta_2I,
\end{gather*}
where $\alpha_j$ and $\beta_j$, $j=1,2$, are real constants.
The above commutation relations give the system
$\sigma^0_t=\alpha_1$, $t\sigma^0_t=\alpha_2$,
$\rho^0_t=\beta_1$, $t\rho^0_t=\beta_2$,
which implies that $\sigma^0$ and $\rho^0$ are constants.
Therefore, by linearly combining with~$M$ and~$I$,
they can be set to be equal to zero.
We successively evaluate the classifying condition~\eqref{Multclasscond}
at~$Q^0$, $Q^1$ and~$Q^2$ and obtain the system
$V_t=0$, $x_1V_2-x_2V_1=0$, $x_aV_a+2V=0$,
whose general solution is $V=|x|^{-2}U$.
At the same time, any such potential admits the Lie symmetry vector field $Q^3=D(t^2)-tI$.
This means that in fact $k_3=3$.
For the Lie symmetry extension to be maximal, we require $U\ne 0$.
Therefore, we obtain~Case~7.

\medskip\par\noindent
$\boldsymbol{k_2=0,\,r_0=1}$.
We definitely have a vector field $Q^1=P(\chi^{11},\chi^{12})+\sigma^1 M+\rho^1 I$ with a nonzero tuple $(\chi^{11},\chi^{12})$
in $\mathfrak g^{\rm ess}_V$.
An additional Lie symmetry extension may be provided
by one more vector field of the same form
or vector fields $Q^q$, $q=k_1+1,\dots,k_1+k_3$, with linearly independent coefficients~$\tau^q$.
We should suppose separately whether the tuple $(\chi^{11},\chi^{12})$ is either proportional
or not proportional to a constant tuple.

\medskip\par\noindent{\bf 1.}
Let the parameter functions~$\chi^{11}$ and~$\chi^{12}$ be linearly dependent.
Up to the $\pi_*G^\sim$-equivalence, we can set $(\chi^{11},\chi^{12})=(1,0)$, $\sigma^1=0$, 
and thus the vector field~$Q^1$ reduces to $P(1,0)+\rho^1I$.
Lemma~\ref{Multlemmatauzero} implies
that then the algebra $\mathfrak g^{\rm ess}_V$ also contains
the vector field $P(t,0)+\rho^2 I$ with \mbox{$\rho^2=\int t\rho^1_t\,{\rm d}t$}.
Then substituting the components of the vector fields into the classifying condition~\eqref{Multclasscond}
yields two dependent equations for~$V$, $V_1=-i\rho^1_t$ and $tV_1=-it\rho^1_t$,
whose general solution is
\begin{equation}\label{EqPotentialsWith2ProportionalChis}
V=U(t,x_2)+i\gamma(t)x_1,
\end{equation}
where $U$ is a complex-valued smooth function of~$(t,x_2)$
and $\gamma=-\rho^1_t$ is a real-valued smooth function of~$t$.
Thus, any equation from the class~$\mathcal F$ with potential of the form~\eqref{EqPotentialsWith2ProportionalChis}
is invariant with respect to the vector fields
$Q^1_\gamma=P(1,0)-\left(\int\gamma(t)\,{\rm d}t\right)I$ and
$Q^2_\gamma=P(t,0)-\left(\int t\gamma(t)\,{\rm d}t\right)I$.
The function~$U$ satisfies the conditions $U_{222}\ne 0$ or $\mathop{\rm Im}U_{22}\ne 0$
since otherwise $r_0=2$ (see the case $r_0=2$ below).

Consider the subclass~$\mathcal F_{{\rm l}x_1}$ 
of (1+2)-dimensional linear Schr\"odinger equations
with potentials of the form~\eqref{EqPotentialsWith2ProportionalChis}
constrained by the above condition for~$U$.
We can reparameterize the class~$\mathcal F_{{\rm l}x_1}$,
assuming the parameter functions~$U$ and~$\gamma$ as arbitrary elements instead of~$V$.
The equivalence groupoid~$\mathcal G^\sim_{{\rm l}x_1}$  of~$\mathcal F_{{\rm l}x_1}$
can be singled out from the equivalence groupoid~$\mathcal G^\sim$ of the entire class~$\mathcal F$,
which is described in Theorem~\ref{thmeMquivptransf}.
A~point transformation connects two equations~$\mathcal L_V$ and~$\mathcal L_{\tilde V}$
from the class~$\mathcal F_{{\rm l}x_1}$ if and only if
it is of the form~\eqref{Mequgroupoidsub_a}--\eqref{Mequgroupoidsub_b},
where the function~$T$ is linear fractional in $t$,
the matrix~$O$ is diagonal with $O^{11},O^{22}=\pm1$,
and $\mathcal X^1=\nu_1 T+\nu_0$ with arbitrary real constants $\nu_1$ and $\nu_0$.
The corresponding arbitrary-element tuples are related in the following way:
\begin{gather}\label{EqRelationBetweenPotentialsWith2ProportionalChis}
\begin{split}
&\tilde U=\frac{\hat U}{|T_t|}
+\frac{\varepsilon'O^{22}}{2|T_t|^{1/2}}\left(\frac{\mathcal X^2_t}{T_t}\right)_{\!t}x_2
+\frac{\Sigma_t-i\Upsilon_t}{T_t}-\frac{(\mathcal X^2_t)^2+2iT_{tt}}{4T_t{}^2}-\frac{\nu_1{}^{\!2}}4-i\tilde \gamma\mathcal X^1,
\\
&\tilde \gamma=\frac{\varepsilon'O^{11}}{|T_t|^{3/2}}\gamma.
\end{split}
\end{gather}
The equivalence group~$G^\sim_{{\rm l}x_1}$ of the class~$\mathcal F_{{\rm l}x_1}$
consists of the point transformations of the form \eqref{Mequgroupoidsub_a}--\eqref{Mequgroupoidsub_b}
prolonged to the arbitrary elements according to~\eqref{EqRelationBetweenPotentialsWith2ProportionalChis},
where the parameters satisfy, in addition to all the above conditions for elements of~$\mathcal G^\sim_{{\rm l}x_1}$,
the constraint $\Lambda=0$.
The class~$\mathcal F_{{\rm l}x_1}$ is (disjointedly) uniformly semi-normalized with respect to
the entire equivalence group~$G^\sim_{{\rm l}x_1}$ and linear superposition of solutions.
At the same time, it is already known that each equation~$\mathcal L_V$,
where the potential~$V$ is of the form~\eqref{EqPotentialsWith2ProportionalChis},
admits the one-parameter groups of point symmetry transformations 
\begin{gather*}
\mathcal Q^1_\gamma(\nu_0)\colon\ \  
\tilde t=t,\ \  
\tilde x_1=x_1+\nu_0,\ \ 
\tilde x_2=x_2,\ \ 
\tilde\psi=\exp\left(-\nu_0\int\gamma(t)\,{\rm d}t\right)\psi,
\\
\mathcal Q^2_\gamma(\nu_1)\colon\ \  
\tilde t=t,\ \  
\tilde x_1=x_1+\nu_0 t,\ \ 
\tilde x_2=x_2,\ \ 
\tilde\psi=\exp\left(\frac i2\nu_1x_1+\frac i4\nu_1^2t-\nu_1\int t\gamma(t)\,{\rm d}t\right)\psi
\end{gather*}
generated by the vector fields~$Q^1_\gamma$ and~$Q^2_\gamma$.
This is why it is convenient to include these symmetry transformations in the corresponding uniform symmetry group~\smash{$G^{\rm unf}_V$}
and to consider the class~$\mathcal F_{{\rm l}x_1}$
as (non-disjointedly and uniformly) semi-normalized with respect to a proper subgroup~$H$ of~$G^\sim_{{\rm l}x_1}$
and a family of wider uniform symmetry groups than the symmetry groups of linear superposition of solutions.

For the group~\smash{$G^{\rm unf}_V$} to be closed with respect to the transformation composition,
we also need to include the elementary transformations~$\mathcal M(\Sigma)$
(cf.\ Section~\ref{sec:FurtherPropertiesOfLieInvAlgs}) with constant values of~$\Sigma$ in~\smash{$G^{\rm unf}_V$} 
since they are the commutators of transformations~$\mathcal Q^1_\gamma(\nu_0)$ and~$\mathcal Q^2_\gamma(\nu_1)$.
% since $[P(1,0),P(t,0)]=M(\frac12)$
The elementary transformations~$\mathcal I(\Upsilon)$ with constant values of~$\Upsilon$ should be in~\smash{$G^{\rm unf}_V$}
since the expressions for the $\psi$-components of the transformations~$\mathcal Q^1_\gamma(\nu_0)$ and~$\mathcal Q^2_\gamma(\nu_1)$
contain antiderivatives of~$\gamma$ and~$t\gamma$, which are defined up to constant summands, 
and there is no canonical choice of these antiderivatives for general values of~$\gamma$.
Thus, for each equation~$\mathcal L_V$ from the class~$\mathcal F_{{\rm l}x_1}$,
we consider the group~\smash{$G^{\rm unf}_V$} of its point symmetry transformations
of the form~\eqref{Mequgroupoidsub_a}--\eqref{Mequgroupoidsub_b},
where $T=t$, $O$ coincides with the identity matrix, $\mathcal X^1=\nu_1 t+\nu_0$, $\mathcal X^2=0$,
$\Sigma=\frac14\nu_1{}^{\!2}t+\mu_0$, $\Upsilon_t=\gamma\mathcal X^1$,
$\mu_0$, $\nu_1$ and~$\nu_0$ are arbitrary constants,
and $\Lambda=\Lambda(t,x)$ is an arbitrary solution of~$\mathcal L_V$.
The corresponding Lie algebra~$\mathfrak g^{\rm unf}_V$
is spanned by the vector fields~$M$, $I$, $Q^1_\gamma$, $Q^2_\gamma$ and~$Z(\eta^0)$
with~$\eta^0$ running through the solution set of~$\mathcal L_V$.
The suitable subgroup~$H$ of the group~\smash{$G^\sim_{{\rm l}x_1}$} is singled out from this group
by the constraint~$\mathcal X^1=0$,
and the Lie algebra associated with~$H$ is
$\mathfrak h=\langle\hat D(1),\hat D(t),\hat D(t^2),\hat P(0,\chi^2),\hat M(\sigma),\hat I(\rho)\rangle$.

Since all the conditions of Definition~\ref{def:Semi-normClass} are satisfied 
and the subgroups $G^{\rm unf}_V$ with $V$ 
running through the set of potentials of the form~\eqref{EqPotentialsWith2ProportionalChis} 
are known and of the same structure,
the class~$\mathcal F_{{\rm l}x_1}$ is (uniformly but not disjointedly) semi-normalized
with respect to the subgroup~$H$ of~$\mathcal G^\sim_{{\rm l}x_1}$
and the family $\mathcal N=\{G^{\rm unf}_V\}$, where
\smash{$\pi_*H\cap G^{\rm unf}_V=G^\cap_{{\rm l}x_1}=G^\cap\ne\{{\rm id}\}$}.
In view of Proposition~\ref{pro:GroupClassificationOfSemi-normClasses},
the group classification of the class~$\mathcal F_{{\rm l}x_1}$
reduces to the classification of $\pi_*H$-inequivalent appropriate subalgebras
of the algebra \[\pi_*\mathfrak h=\langle D(1),D(t),D(t^2),P(0,\chi^2),\sigma M,\rho I\rangle.\]
Denote $\mathfrak g^{\rm ext}_V:=\mathfrak g^{\rm ess}_V\cap\pi_*\mathfrak h$.
The condition $r_0=1$ and Corollary~\ref{cor:LinMultiDSchEqsKernelAlgebra} imply that
any Lie symmetry extension in the subclass~$\mathcal F_{{\rm l}x_1}$
is spanned by vector fields of the form $D(\tau)+P(0,\chi^2)+\sigma M+\rho I$,
where $\tau$ runs through a set of linearly independent quadratic polynomial in~$t$.
This shows that the computation of inequivalent Lie symmetry extensions reduces to
the classification of subalgebras of the algebra ${\rm sl}(2,\mathbb R)$.
Depending on the values of $k_3$ we obtain the following subcases.

\medskip\par\noindent
$k_3=0$. There is no additional extension, i.e., we have~Case~8.

\medskip\par\noindent
$k_3=1.$
The algebra $\mathfrak g^{\rm ext}_V$ necessarily contains a vector field~$Q^3$ with nonzero $\tau^3$.
Up to the $\pi_*H$-equivalence, the vector field~$Q^3$ takes one of the following values: $D(1)$, $D(t)$ and $D(t^2+1)-tI$.
Taking into account the form~\eqref{EqPotentialsWith2ProportionalChis} of $V$,
we successively evaluate the classifying condition~\eqref{Multclasscond} at each value of~$Q^3$
and obtain equations for~$U$ and~$\gamma$.
% D(1): \gamma_t=0, U_t=0,
% D(t): t\gamma_t+\frac32\gamma=0, tU_t+\frac12x_2U_2+U=0,
% D(t^2+1)-tI: (t^2+1)\gamma_t+3t\gamma=0, (t^2+1)U_t+tx_2U_2+2tU=0,
The solution of these equations results in Cases~9, 10 and~11, respectively,
where we change the notation of~$U$.
Note that we choose the special form of the last value of~$Q^3$
for simplifying the representation of the corresponding potentials.

\medskip\par\noindent
$k_3\geqslant 2.$
Apart from the vector fields~$M$ and~$I$,
the algebra $\mathfrak g^{\rm ext}_V$ additionally contains
at least vector fields $Q^q$, $q=3,4$,
with linearly independent $\tau^3$ and $\tau^4$.
Up to the $\pi_*H$-equivalence and linearly combining with~$M$ and~$I$,
the vector fields $Q^q$ reduce to $Q^3=D(1)$ and $Q^4=D(t)$.
Evaluating the classifying condition~\eqref{Multclasscond} simultaneously at these values of~$Q^3$ and~$Q^4$
for the expression~\eqref{EqPotentialsWith2ProportionalChis} of~$V$,
we obtain $U_t=0$, $x_2U_2+2U=0$ and $\gamma=0$, and thus $V=\tilde U x_2^{-2}$ with a nonzero complex constant $\tilde U$.
At the same time, the equation~$\mathcal L_V$ with this~$V$ possesses one more Lie symmetry vector field $Q^5=D(t^2)-tI$.
This shows that the Lie symmetry extension with $k_3=2$ is not maximal.
Therefore, we obtain~Case~12.

\medskip\par\noindent {\bf 2.}
Suppose now that the parameter functions~$\chi^{11}$ and~$\chi^{12}$ are linearly independent.

If we have no additional extension,
then the algebra $\mathfrak g^{\rm ess}_V$ is spanned
by the vector fields $M$, $I$ and $Q^1=P(\chi^{11},\chi^{12})+\sigma^1 M+\rho^1 I$.%
\footnote{%
In contrast to the previous case,
where the parameter functions~$\chi^{11}$ and~$\chi^{12}$ are linearly dependent,
in the case under consideration the algebra~$\mathfrak g^{\rm ess}_V$ contains,
up to linear dependence, only one vector field of the form as~$Q^1$.
Indeed, suppose that this algebra contains one more vector field
$Q^2=P(\chi^{21},\chi^{22})+\sigma^2 M+\rho^2 I$,
where the tuple $(\chi^{21},\chi^{22})$ is not linearly dependent with $(\chi^{11},\chi^{12})$.
Then the case condition $r_0=1$ implies that $\chi^{2a}=\lambda\chi^{1a}$
where $\lambda$ is a nonconstant function of $t$.
Successively evaluating the classifying condition~\eqref{Multclasscond} at the vector fields $Q^1$ and $Q^2$,
we derive two equations for~$V$,
for which the difference of the second equation and the first equation multiplied by $\lambda$
leads, in view of the proportionality of the tuples, to the condition
$(\lambda_t\chi^{1a}_t+\frac12\lambda_{tt}\chi^{1a})x_a+\sigma^2_t
-\lambda\sigma^1_t-i(\rho^2_t-\lambda\rho^1_t)=0$.
Splitting it with respect to~$x_a$ and integrating the obtained equations,
we derive $\chi^{1a}=c_a|\lambda_t|^{-1/2}$, where $c_a$ are real constants, 
which contradicts the linear independence of~$\chi^{11}$ and~$\chi^{12}$.
}
Lemma~\ref{Multgclaslema} implies that the vector field~$Q^1$ reduces to
the form $Q^1=P(h \cos t,h \sin t)+\rho I$,
where $h$ and $\rho$ are smooth functions of $t$ with $h\ne 0$.
The substitution of the components of~$Q^1$ into the classifying condition~\eqref{Multclasscond}
yields the equation
\[
 V_1\cos t+V_2\sin t =\frac 12 (h^{-1}h_{tt}-1)\omega_1+h^{-1}h_t\omega_2-ih^{-1}\rho_t,
\]
whose general solution is presented in~Case~13.

Otherwise, the algebra $\mathfrak g^{\rm ess}_V$ also contains
a vector field $Q^2$ with nonzero $\tau^2$, which takes,
up to the $\pi_*G^\sim$-equivalence, the form $Q^2=D(1)+\kappa_2 J$, $\kappa_2\in\{0,1\}$.
The condition $[Q^2,Q^1]\in \mathfrak g^{\rm ess}_V$ means that
\[P(\chi^{11}_t,\chi^{12}_t)+\kappa_2P(\chi^{12},-\chi^{11})+\sigma^1_tM+\rho^1_t=\beta_1 Q^1+\beta_2M+\beta_3I,
\]
where $\beta_j,\ j=1,2,3$ are real constants.
Equating the corresponding components of the vector fields on both the sides, we derive the system
\begin{gather}\label{non-proportional tuple}
\chi^{11}_t+\kappa_2\chi^{12}=\beta_1\chi^{11},\quad \chi^{12}_t-\kappa_2\chi^{11}=\beta_1\chi^{12},\quad
\sigma^1_t=\beta_1\sigma^1+\beta_2,\quad\rho^1_t=\beta_1\rho^1+\beta_3.
\end{gather}
If $\kappa_2=0$, then the parameter functions~$\chi^{11}$ and~$\chi^{12}$ are linearly dependent.
Hence $\kappa_2=1$, i.e., $Q^2=D(1)+ J$.
Integrating the systems~\eqref{non-proportional tuple} and linearly combining~$Q^1$ with elements from the kernel,
we obtain a reduced form of $Q^1$ depending on the value of $\beta_1$.

If $\beta:=\beta_1\ne 0$, then
$Q^1=P(e^{\beta t}\cos t,e^{\beta t}\sin t)+\tilde\alpha e^{\beta t}M+\alpha e^{\beta t}I$
with real constants $\tilde\alpha$ and $\alpha$.
Pushing vector fields from $\mathfrak g^{\rm ess}_V$ forward
by $\mathcal P_*(-\tilde\alpha\sin t,\tilde\alpha\cos t)$, we set $\tilde\alpha=0$.
The successive substitution of the components of the vector fields~$Q^1$ and $Q^2$
into the classifying condition~\eqref{Multclasscond} gives the system
\[
V_1\cos t+V_2\sin t =\frac12(\beta^2-1)\omega_1+\beta\omega_2-i\beta\alpha,\quad
V_t+x_1V_2-x_2V_1=0
\]
whose general solution is presented in Case~14.

If $\beta_1=0$, then the solution of the system~\eqref{non-proportional tuple} provides,
after linearly combining with $M$ and~$I$,
$Q^1=P(\cos t,\sin t)+\beta_2tM +\beta_3tI$. In this case no further simplifications are possible.
For convenience, we re-denote $\beta_2:=\tilde\alpha$ and $\beta_3:=\alpha$.
Evaluating the classifying condition~\eqref{Multclasscond} at~$Q^1$ and at~$Q^2$ yields two equations for~$V$,
\[V_t+x_2V_1-x_1V_2=0,\quad V_1\cos t+V_2\sin t =-\frac12\omega_1+\tilde\alpha-i\alpha.\]
Therefore, we obtain~Case~15.

\medskip\par\noindent {\bf Further classification}. 
In view of Lemma~\ref{lem:k2r0}, 
the values of $(k_2,r_0)$ not considered so far are $(0,2)$ and $(1,2)$.

Since $r_0=2$, the algebra $\mathfrak g^{\rm ess}_V$ contains at least two vector fields 
of the form~$Q^a=P(\chi^{a1},\chi^{a2})+\sigma^a M+\rho^a I$, $a=1,2$,
where $\chi^{11}\chi^{22}-\chi^{12}\chi^{21}\ne 0$.
The substitution of the components of the vector fields $Q^a$
into the classifying condition~\eqref{Multclasscond} leads to the following system in $V$:
\[
\chi^{ab}V_b=\frac12 \chi^{ab}_{tt}x_b+\sigma^a_t-i\rho^a_t.
\]
This system can be written as
$V_a=\frac 12h^{ab}(t)x_b+\tilde h^{0a}(t)+ih^{0a}(t)$,
where the coefficients~$h^{ab}$, $\tilde h^{0a}$ and~$h^{0a}$ are real-valued functions of $t$ 
that satisfy the conditions
\begin{gather*}
\chi^{ab}_{tt}=\chi^{ac}h^{cb},\quad \sigma^a_t=\chi^{ac}\tilde h^{0c},\quad \rho^a_t=-\chi^{ac} h^{0c}.
\end{gather*}
The condition $V_{12}=V_{21}$ implies that the matrix $(h^{ab})$ is symmetric.
Thus, the potential~$V$ is a quadratic polynomial in $x_1$ and $x_2$ with coefficients depending on $t$,
\begin{gather}\label{QuadraticPotential}
V=\frac 14h^{ab}(t)x_ax_b+\tilde h^{0b}(t)x_b+ih^{0b}(t)x_b+\tilde h^{00}(t)+ih^{00}(t),
\end{gather}
where the functions $\tilde h^{0b}$, $\tilde h^{00}$ and $ h^{00}$ can be set equal to zero up to the $G^\sim$-equivalence 
but we will use this possibility later.

Next consider the subclass~$\mathcal F_{\rm q}$ of equations from the class~$\mathcal F$
with potentials of the form~\eqref{QuadraticPotential}.
Theorem~\ref{thmeMquivptransf} implies that, similarly to the entire class $\mathcal F$,
the subclass~$\mathcal F_{\rm q}$ is uniformly semi-normalized with respect to the linear superposition of solutions,
and its equivalence group coincides with the equivalence group $G^\sim$ of $\mathcal F$.
The further group classification of the subclass $\mathcal F_{\rm q}$ splits into cases 
depending on the value of~$k_2$.
The conditions $k_2=0$ and $k_2=1$ partition the subclass $\mathcal F_{\rm q}$
into the subclasses $\mathcal F^0_{\rm q}$ and $\mathcal F^1_{\rm q}$, respectively.
These subclasses are uniformly semi-normalized with respect to the linear superposition of solutions as well, 
and their equivalence groups also coincide with the equivalence group $G^\sim$ of the entire class $\mathcal F$.
The classifying condition~\eqref{Multclasscond} implies
that the subclass~$\mathcal F^1_{\rm q}$ is singled out from the class~$\mathcal F_{\rm q}$
by the conditions $h^{12}=h^{21}=0$, $h^{11}=h^{22}$, $h^{01}=0$ and $h^{02}=0$,
and the subclass~$\mathcal F^0_{\rm q}$ is then singled out by requiring that one of these conditions fails to hold.
In view of Corollary~\ref{free Schrodingerequation} the above condition means
that each equation in the subclass~$\mathcal F^1_{\rm q}$ is
$G^\sim$-equivalent to the (1+2)-dimensional free Schr\"odinger equation (with $V=0$).
In other words, the subclass~$\mathcal F^1_{\rm q}$ is the $G^\sim$-orbit of the free Schr\"odinger equation,
which gives Case~19.

Therefore, the last consideration concerns the subclass $\mathcal F^0_{\rm q}$,
for which $k_2=0$ and $r_0=2$.
Up to the $G^\sim$-equivalence, the simplified form of potentials for equations from the subclass $\mathcal F^0_{\rm q}$ is
\begin{gather}\label{QuadraticSimplifiedPotential}
V=\frac 14h^{ab}(t)x_ax_b+i h^{0b}(t)x_b.
\end{gather}
Substituting $V$ of this form
into the classifying condition~\eqref{Multclasscond} with $\tau=0$ and $\kappa=0$,
splitting with respect to $x_a$ and separating real and imaginary parts,
we derive the following system for $(\chi^1,\chi^2,\sigma,\rho)$:
\begin{gather}\label{MlinSchSimirltycl}
\chi^a_{tt}=h^{ab}\chi^b,\quad \sigma_t=0,\quad \rho_t=- h^{0b}\chi^b.
\end{gather}
This system admits a fundamental set of solutions $(\chi^{p1},\chi^{p2},0,\rho^p)$, $p=1,\dots,4$, $(0,0,1,0)$ and $(0,0,0,1)$,
where the tuples $(\chi^{p1},\chi^{p2})$ are linearly independent, and $\rho^p=-\int h^{0b}\chi^{pb}{\rm d}t$.
Thus, the algebra $\mathfrak g^{\rm ess}_V$ contains the four vector fields $Q^p=P(\chi^{p1},\chi^{p2})+\rho^p I$.
The last two solutions correspond to the vector fields $M$ and $I$ from the kernel algebra~$\mathfrak g^{\cap}$.
Case~16 presents general equations from the subclass $\mathcal F^0_{\rm q}$ with no
additional Lie symmetry extension.

Otherwise, the algebra $\mathfrak g^{\rm ess}_V$ contains a vector field $Q^5$ with nonzero $\tau^5$.
Up to the $G^\sim$-equivalence, the parameter function~$\tau^5$ is reduced to~1, $\kappa_5\in\{0,1\}$,
and the simplified form~\eqref{QuadraticSimplifiedPotential} of~$V$ is preserved.
We substitute the components of $Q^5$ jointly with the form~\eqref{QuadraticSimplifiedPotential} of~$V$
into the classifying condition~\eqref{Multclasscond} and split with respect to different powers of~$x_a$.
As a result, we obtain that the tuple $(\chi^{51},\chi^{52},\sigma^5,\rho^5)$ satisfies
the system~\eqref{MlinSchSimirltycl},
and the parameter functions~$h^{ab}$ and~$h^{0b}$ are constrained by the equations
\begin{equation}\label{equationskappa}
\begin{split}
& h^{11}_t+2\kappa_5 h^{12}=0,\quad h^{12}_t+\kappa_5 (h^{22}-h^{11})=0,\quad h^{22}_t-2\kappa_5 h^{12}=0,\\
& h^{01}_t+\kappa_5 h^{02}=0,\quad h^{02}_t-\kappa_5 h^{01}=0.
\end{split}
\end{equation}
Therefore, up to linearly combining $Q^5$ with~$Q^p$, $M$ and~$I$,
we can set $\chi^{5a}=\sigma^5=\rho^5=0$. This reduces $Q^5$ to $D(1)+\kappa_5J$. 
Consider the cases $\kappa_5=0$ and $\kappa_5=1$ separately. 

For $\kappa_5=0$, the system~\eqref{equationskappa} implies that all $h^{ab}$ and $h^{0a}$ are constants.
Up to rotations, we can reduce the matrix $(h^{ab})$ to a diagonal matrix
${\rm diag}(\alpha,\beta)$. The maximality of the Lie symmetry extension requires
$\alpha \ne \beta$ or $(\nu_1,\nu_2)\ne (0,0)$ with $\nu_a:=h^{0a}$, which gives~Case~17.

If $\kappa_5=1$, then, up to translations of time, the general solution of the system~\eqref{equationskappa} is
$h^{11}=\alpha \cos^2 t+\beta\sin^2 t$, $h^{12}=h^{21}=(\alpha-\beta) \cos t \sin t$,
$h^{22}=\alpha\sin^2 t+\beta\cos^2 t$, $h^{01}=\nu_1\cos t-\nu_2\sin t$ and
$h^{02}=\nu_1\sin t+\nu_2\cos t$, where $\alpha$, $\beta$, $\nu_1$ and $\nu_2$ are arbitrary real constants.
The Lie symmetry extension is maximal if these constants satisfy
the inequalities $\alpha \ne \beta$ or $(\nu_1,\nu_2)\ne (0,0)$.
Rearranging the potential $V$ in terms of $\omega_a$ leads to~Case~18. 
\looseness=-1

Let us show that the dimension of additional Lie symmetry extension cannot exceed one.
Suppose that this is not the case.
Then the algebra~$\mathfrak g^{\rm ess}_V$ contains a two-dimensional subalgebra
spanned by vector fields~$Q^5$ and~$Q^6$ with linearly independent~$\tau^5$ and~$\tau^6$.
Up to the $\pi_*G^\sim$-equivalence, we can assume that $\tau^5=1$ and $\tau^6=t$.
In a similar way as above, the vector field~$Q^5$ takes the form $Q^5=D(1)+\kappa_5J$.
The commutation relation for~$Q^5$ and~$Q^6$ implies that $\kappa_5=0$.
Successively evaluating the classifying condition~\eqref{Multclasscond} at~$Q^5$ and at~$Q^6$
for the form~\eqref{QuadraticSimplifiedPotential} of~$V$
and splitting with respect to different powers of $x_a$ provides
the equations $h^{ab}=h^{0b}=0$, which contradicts the condition
singling out the subclass~$\mathcal F^0_{\rm q}$ from the class~$\mathcal F_{\rm q}$.
\end{proof}

\section{Subclass with real-valued potentials}\label{sec:Real-ValuedPotentials}

Consider the subclass~$\mathcal F_{\mathbb R}$ of the class~$\mathcal F$, 
which consists of the equations of the form~\eqref{MLinSchEqs} with real-valued potentials~$V$, 
i.e., it is singled out from the class~$\mathcal F$ by the constraint $\mathop{\rm Im}V:=-\frac12i(V-V^*)=0$.
Since the class~$\mathcal F$ is uniformly semi-normalized with respect to the linear superposition of solutions, 
all results on symmetry analysis of its subclass~$\mathcal F_{\mathbb R}$ 
can be easily derived from the corresponding results for this entire class 
via singling out objects that are consistent with the constraint $\mathop{\rm Im}V=0$.
In particular, using Theorem~\ref{thmeMquivptransf}, we construct 
the equivalence groupoid~$\mathcal G^\sim_{\mathbb R}$ of the class~$\mathcal F_{\mathbb R}$ 
as the maximal subgroupoid of~$\mathcal G^\sim$  
with the set of objects $\mathcal S_{\mathbb R}\subset\mathcal S$ 
consisting of the real-valued potentials~$V$.

\begin{corollary}\label{thm:EquivGroupoidOfFR}
The equivalence groupoid~$\mathcal G^\sim_{\mathbb R}$ of the class~$\mathcal F_{\mathbb R}$ 
consists of the triples of the form $(V,\Phi,\tilde V)$,
where $\Phi$ is a point transformation in the space of variables,
whose components are of the form~\eqref{Mequgroupoidsub_a}--\eqref{Mequgroupoidsub_b}
and the target potential~$\tilde V$ is expressed via the source potential~$V$ 
as in~\eqref{Mequgroupoidsub_c},  
where in addition we omit the hat accent $\hat{\ }$ over~$V$ and
\[
\Upsilon_t=-\frac{nT_{tt}}{4T_t}.
\]
\end{corollary}

\begin{corollary}%\label{cor:Reduction ToFreeSchrodingerEqInFR}
A (1+n)-dimensional linear Schr\"odinger equation of the form~\eqref{MLinSchEqs} 
with a real-valued potential~$V$
is reduced to the free linear Schr\"odinger equation by a point transformation if and only if
\[V=\varrho(t)|x|^2+\varrho^a(t)x_a+\varrho^0(t)\]
for some real-valued smooth functions $\varrho$, $\varrho^a$ and $\varrho^0$ of $t$.
\end{corollary}

\begin{corollary}%\label{Multieqgpe}
The (usual) equivalence group~$G^\sim_{\mathbb R}$ of the class~$\mathcal F_{\mathbb R}$ consists 
of point transformations in the space with the coordinates~$(t,x,\psi,\psi^*,V)$
whose components are of the form~\eqref{Mequgroupoidsub}, 
where in addition we omit the hat accent $\hat{\ }$ over~$V$ and
\[
\Upsilon_t=-\frac{nT_{tt}}{4T_t},\quad \Lambda=0.
\]
\end{corollary}

\begin{corollary}%\label{cor:UniformSemi-normalizationOfClassFR}
The class~$\mathcal F_{\mathbb R}$ is uniformly semi-normalized with respect to the linear superposition of solutions.
\end{corollary}

\begin{remark}
The identity component of~$G^\sim_{\mathbb R}$ is constituted by the transformations from this group,
where in addition $\det O=1$ and $T_t>0$, i.e., $\varepsilon'=1$.
Similarly to~$G^\sim$, the entire equivalence group~$G^\sim_{\mathbb R}$ is generated
by the transformations from its identity component and two discrete transformations,
the space reflection
$\tilde t=t,$ $\tilde x_a=-x_a,$ $\tilde x_b=x_b,$ $b\ne a$, $\tilde \psi=\psi,$ $\tilde V=V$ for a fixed~$a$
and the Wigner time reflection
$\tilde t=-t,$ $\tilde x=x,$ $\tilde \psi=\psi^*$, $\tilde V=V$.
\end{remark}

\begin{corollary}
In the notation of Theorem~\ref{thm:EquivAlgOfF}, 
where we omit the $V^*$-components in all vector fields, 
the equivalence algebra of the class~$\mathcal F_{\mathbb R}$ is the algebra
\[
\mathfrak g^\sim_{\mathbb R}=\left\langle
\hat D(\tau)-\frac n4\hat I(\tau_t),\,\hat J_{ab},\,a<b,\,\hat P(\chi),\,\hat M(\sigma),\,\hat I(1)\right\rangle.
\]
\end{corollary}

Lie-symmetry vector fields of equations from the class~$\mathcal F_{\mathbb R}$ 
have the same properties as those within the class~$\mathcal F$ 
that are presented in Sections~\ref{sec:AnalysOfDetEqsForLieSymsOfMultiDSchEqs} and~\ref{sec:FurtherPropertiesOfLieInvAlgs}
and additionally satisfy the constraint $\rho_t+\frac14n\tau_{tt}=0$.
In particular, 
the kernel invariance algebra of the class~$\mathcal F_{\mathbb R}$ is 
$\mathfrak g^\cap_{\mathbb R}=\mathfrak g^\cap=\langle M,I\rangle$, and 
the linear span of essential Lie-symmetry vector of equations from the class~$\mathcal F_{\mathbb R}$ is 
\[
\mathfrak g^{\rm ess}_{\spanindex\mathbb R}:=\sum_{V\in\mathcal S_{\mathbb R}}\mathfrak g^{\rm ess}_V
=\left\{D(\tau)-\frac n4\tau_tI+\sum_{a<b}\kappa_{ab}J_{ab}+P(\chi)+\sigma M+cI\right\}.
\]
The problem of group classification of (1+n)-dimensional linear Schr\"odinger equations 
with real-valued potentials
reduces to the classification of appropriate subalgebras of the algebra~$\mathfrak g^{\rm ess}_{\spanindex\mathbb R}$ 
with respect to the equivalence relation generated by the action of~$\pi_*G^\sim_{\mathbb R}$.

In view of the relation~\eqref{Mequgroupoidsub_c} between source and target potentials 
for admissible transformations within the class~$\mathcal F$, 
an equation~$\mathcal L_V$ from this class can be mapped by a point transformation 
to an equation from the class~$\mathcal F_{\mathbb R}$ if and only if 
the imaginary part of~$V$ depends at most on~$t$.
This is why we do not need to solve the group classification problem for the class~$\mathcal F_{\mathbb R}$ 
with $n=2$ from the very beginning, 
but its solution can be easily derived from that for the class~$\mathcal F$ 
presented in Theorem~\ref{thm:GroupClassifictionOf(1+2)DLinSchEqs}.
It suffices to single out, among the cases listed in Theorem~\ref{thm:GroupClassifictionOf(1+2)DLinSchEqs}, 
those where the imaginary parts of the corresponding potentials are zero. 
More specifically, we should set 
$\mathop{\rm Im}V=0$ in Case~0, 
$\mathop{\rm Im}U=0$ in Cases~1--15, 
$\rho=0$ in Cases~5 and~13, 
$\beta_1=0$ in Case~6, 
$\gamma=0$ in Case~8, 
$\beta=0$ in Cases~9--11, 
$\alpha=0$ in Cases~14 and~15,
$h^{01}=h^{02}=0$ in Case~16 and 
$\nu_1=\nu_2=0$ in Cases~17 and~18.

\section{Conclusion}\label{sec:Conclusion}

The initial aim of the paper was to study transformational properties 
of (1+$n$)-dimensional linear Schr\"odinger equations with complex-valued potentials, 
which are of the form~\eqref{MLinSchEqs} and constitute the class denoted by~$\mathcal F$, 
for arbitrary values of $n\in\mathbb N$ 
and to carry out the complete group of such equations with $n=2$.
It turned out that this effort requires advancing the theoretical background 
of symmetry analysis in classes of differential equations, 
which became the main subject of the paper.
We have revisited and extended the entire framework of normalized classes of differential equations, 
including its basic notions and terminology, 
and proved the theorems on factoring out point-symmetry groups and maximal Lie invariance algebras within semi-normalized classes 
and the stronger theorems on splitting such groups and algebras within disjointedly semi-normalized classes.
We have also constructed several proof-of-concept examples.
This has allowed us to extend the algebraic method of group classification to classes 
that are semi-normalized in a more general sense,  
whereas previously only the disjointed semi-normalization with respect to the entire corresponding equivalence groups 
(in the present terminology) was considered. 

For the class~$\mathcal F$ with an arbitrary~$n$, we first computed its equivalence groupoid~$\mathcal G^\sim$  
by the direct method and thus proved that this class is uniformly semi-normalized with respect to the linear superposition of solutions,
which motivates the use of the developed version of the algebraic method of group classification.
Knowing~$\mathcal G^\sim$, we have easily found the equivalence group~$G^\sim$ of the class~$\mathcal F$ 
and constructed the associated equivalence algebra~$\mathfrak g^\sim$ 
as the set of infinitesimal generators of one-parameter subgroups of the group~$G^\sim$.

Working within the framework of the algebraic method,
we have reduced the group classification of equations from the class~$\mathcal F$ 
to the classification of specific low-dimensional subalgebras of the algebra~$\mathfrak g^\sim$ 
or, equivalently,~$\pi_*\mathfrak g^\sim$.
After analyzing the determining equations for Lie symmetries of equations from the class~$\mathcal F$,
we have constructed the linear span $\mathfrak g_\spanindex$ of the vector fields from the maximal Lie invariance algebras 
of these equations.
This linear span can be represented as the semidirect sum of the so-called essential subalgebra
$\mathfrak g^{\rm ess}_\spanindex$ and an ideal $\mathfrak g^{\rm lin}_\spanindex$
related to the %transformations of 
linear superposition of solutions,
\smash{$\mathfrak g_\spanindex=\mathfrak g^{\rm ess}_\spanindex\lsemioplus\mathfrak g^{\rm lin}_\spanindex$}. 
We have shown that $\mathfrak g^{\rm ess}_\spanindex=\pi_*\mathfrak g^\sim$.
For each equation~$\mathcal L_V$ from the class~$\mathcal F$, the representation for $\mathfrak g_\spanindex$
induces the similar representation for the maximal Lie invariance algebra~$\mathfrak g_V$ of~$\mathcal L_V$,
\smash{$\mathfrak g_V=\mathfrak g^{\rm ess}_V\lsemioplus\mathfrak g^{\rm lin}_V$}. 
The analysis of the determining equations has also resulted in the principal constraints 
for the appropriate subalgebras of the algebra~$\pi_*\mathfrak g^\sim$.

The group classification problem of the class~$\mathcal F$ with $n=2$ 
has been completely solved using the results obtained for the general value of~$n$. 
All possible inequivalent families of potentials
possessing essential Lie symmetry extensions have been listed in Theorem~\ref{thm:GroupClassifictionOf(1+2)DLinSchEqs}.

In contrast to the case of single space variable $(n=1)$ studied in~\cite{kuru2018a},
the maximal Lie invariance algebras of equations from the class~$\mathcal F$ with $n=2$
can be of essentially greater dimensions, 
and some of these algebras involve the rotation vector field, which cannot appear for $n=1$.  
These two facts make the group classification more complex.
To come up with this challenge and make the classification rigorous and efficient,
we have split the group classification of the class~$\mathcal F$ 
into different cases depending on three integers, $k_2\in\{0,1\}$, $r_0\in\{0,1,2\}$ and $k_3\in\{0,1,2,3\}$.
These integers
characterize the dimensions of parts of the corresponding essential Lie invariance subalgebra 
that are related to transformations of time, rotations 
and generalized time-dependent shifts with respect to space variables, respectively,
and are invariant under acting by the equivalence group~$G^\sim$ of the class~$\mathcal F$.
Not all values of the tuple $(k_2,r_0,k_3)$ are appropriate. 
The constraint $(k_2,r_0)\ne(1,1)$ for the appropriate values of $(k_2,r_0,k_3)$ 
is quite obvious and is presented in Lemma~\ref{lem:k2r0} 
in the course of the preliminary analysis of Lie symmetries of equations from the class~$\mathcal F$.
This is not the case for the other constraint $k_3\ne2$ if $(k_2,r_0)\ne(0,0)$,
which can be derived only from the entire proof of Theorem~\ref{thm:GroupClassifictionOf(1+2)DLinSchEqs}.

We have also comprehensively studied the subclass~$\mathcal F_{\mathbb R}$ of~$\mathcal F$, 
which is constituted by equations of the form~\eqref{MLinSchEqs} with real-valued potentials. 
One can easily convert the above results for the class~$\mathcal F$ 
into those for the subclass~$\mathcal F_{\mathbb R}$, 
merely selecting all the transformational objects related to the class~$\mathcal F$ 
that are relevant for real-valued potentials.
In this way, we have straightforwardly constructed 
the equivalence groupoid and the equivalence group of the subclass~$\mathcal F_{\mathbb R}$, 
proved its uniform semi-normalization with respect to the linear superposition of solutions,
described properties of Lie symmetries of its elements 
and obtained its complete group classification in the case of two space variables. 

The similar study for the subclasses~$\mathcal F'$ and~$\mathcal F'_{\mathbb R}$ 
consisting of equations of the form~\eqref{MLinSchEqs} with time-independent complex- and real-valued potentials, 
respectively, is not straightforward at all. 
Even singling out the equivalence groupoids of these subclasses as subgroupoids 
of the groupoids~$\mathcal G^\sim$ and~$\mathcal G^\sim_{\mathbb R}$ 
given in Theorem~\ref{thmeMquivptransf} and Corollary~\ref{thm:EquivGroupoidOfFR} 
requires significant efforts. 
The entire group analysis of the subclasses~$\mathcal F'$ and~$\mathcal F'_{\mathbb R}$ 
is complicated due to the fact that these subclasses are not semi-normalized. 
A particular consequence of this fact is that one should consider 
two different group classification problems for each of the subclasses~$\mathcal F'$ and~$\mathcal F'_{\mathbb R}$, 
up to the equivalence generated by the corresponding equivalence groupoid and 
up to the equivalence generated by the corresponding equivalence group.  
The group classification lists for~$\mathcal F'$ and~$\mathcal F'_{\mathbb R}$ up to the former equivalence 
can be formally derived from Theorem~\ref{thm:GroupClassifictionOf(1+2)DLinSchEqs}, 
but they are not too useful if the equivalence groupoids of~$\mathcal F'$ and~$\mathcal F'_{\mathbb R}$ 
are not known. 
The study of Lie symmetries of Schr\"odinger equations from the class~$\mathcal F'_{\mathbb R}$ 
was initiated by Boyer in~\cite{Boyer1974}. 
The classification list of potentials was presented for $n=3$, 
including the unfounded claim that it can be converted to the lists for $n=1$ and $n=2$ by ``straightforward restriction''.
It is of common knowledge that Boyer's classification list has a number of weaknesses, 
see, e.g.,~\cite{niki2016a}. 
There are missed and equivalent cases in this list and essential gaps in its derivation. 
In particular, it is not clear which kind of equivalence had been used. 
The above weaknesses have not been comprehensively corrected in the literature. 
Hence the rigorous group analysis of the subclasses~$\mathcal F'$ and~$\mathcal F'_{\mathbb R}$ 
with explicit indication of the used equivalence is still an open problem. 

The results on group classification of the class~$\mathcal F$ and its subclasses can be used, in particular, 
to classify Lie reductions of equations from this class and construct their invariant solutions.

\section*{Acknowledgements}

The authors are grateful to Peter Basarab-Horwath, Michael Kunzinger and Galyna Popovych 
for helpful discussions and suggestions.
The research of C.K. was supported by International Science Programme (ISP) 
in collaboration with Eastern Africa Universities Mathematics Programme (EAUMP) 
and the Abdus Salam International Centre for Theoretical Physics (ICTP). 
The research of D.R.P. was undertaken thanks to funding from the Canada Research Chairs program,
the InnovateNL LeverageR{\&}D program and the NSERC Discovery Grant program.
The research was also supported in part by the Ministry of Education, Youth and Sports of the Czech Republic (M\v SMT \v CR)
under RVO funding for I\v C47813059.
R.O.P. expresses his gratitude for the hospitality shown by the University of Vienna during his long staying at the university.
D.R.P. and R.O.P. deeply thank to the Armed Forces of Ukraine and the civil Ukrainian people
for their bravery and courage in defense of peace and freedom in Europe and in the entire world from russism.

\end{document}